\documentclass[12pt]{amsart}
\usepackage{amsmath,amsthm,amsfonts,amssymb,eucal, hyperref}

\renewcommand {\a}{ \alpha }
\renewcommand{\b}{\beta}

\newcommand{\g}{\gamma}

\newcommand{\s}{\sigma}
\renewcommand{\l}{\lambda}
\renewcommand{\L}{\Lambda}
\newcommand{\z}{\zeta}
\renewcommand{\t}{\theta}

\newcommand{\p}{\partial}
\newcommand{\om}{\omega}
\newcommand{\Om}{\Omega}

\newcommand{\oq}{\ {\raise 7pt\hbox{${\scriptstyle\circ}$}}
\kern -7pt{
\hbox{$Q$}}}

\newcommand{\R}{ \mathbb R}
\newcommand{\N}{ \mathbb N}

\newcommand{\Rd}{ \mathbb R^d}

\newcommand {\GA}{\mathfrak A}

\newcommand {\GF}{\mathfrak F}
\newcommand {\GH}{\mathfrak H}

\newcommand {\GL}{\mathfrak L}

\newcommand {\GV}{\mathfrak V}
\newcommand {\GW}{\mathfrak W}
\newcommand {\GU}{\mathfrak U}

\newcommand {\GX}{\mathfrak X}
\newcommand {\ba}{\mathbf a}

\newcommand {\BD}{\mathbf D}

\newcommand {\BM}{\mathbf M}

\newcommand {\BS}{\mathbf S}

\newcommand {\BT}{\mathbf T}

\newcommand {\bx}{\mathbf x}

\newcommand {\be}{\mathbf e}
\newcommand {\bh}{\mathbf h}
\newcommand {\bk}{\mathbf k}

\newcommand {\bz}{\mathbf z}
\newcommand {\by}{\mathbf y}
\newcommand {\bt}{\mathbf t}

\newcommand {\bs}{\mathbf s}
\newcommand {\bu}{\mathbf u}
\newcommand {\bv}{\mathbf v}
\newcommand {\bn}{\mathbf n}

\newcommand {\bnu}{\boldsymbol\nu}

\newcommand {\bmu}{\boldsymbol\mu}
\newcommand {\bth}{\boldsymbol\theta}
\newcommand {\Bth}{\boldsymbol\Theta}
\newcommand {\boldeta}{\boldsymbol\eta}

\newcommand {\bphi}{\boldsymbol\phi}
\newcommand {\bxi}{\boldsymbol\xi}

\newcommand {\BXi}{\boldsymbol\Xi}
\newcommand {\Bxi}{\boldsymbol\Xi}

\newcommand{\bchi}{\boldsymbol\chi}

\newcommand{\lu}{\langle}
\newcommand{\ru}{\rangle}

\newcommand{\CV}{\mathcal V}

\newcommand{\CR}{\mathcal R}

\newcommand{\CX}{\mathcal X}

\newcommand{\CO}{\mathcal O}
\newcommand{\CP}{\mathcal P}

\newcommand{\CA}{\mathcal A}

\newcommand{\CM}{\mathcal M}

\newcommand{\CC}{\mathcal C}

\newcommand{\plainC}[1]{\textup{{\textsf{C}}}^{#1}}

\newcommand{\plainS}{\textup{{\textsf{S}}}}

\DeclareMathOperator{\tr}{{tr}}
\DeclareMathOperator{\card}{{card}}

\newcommand{\1}
{{\,\vrule depth3pt height9pt}{\vrule depth3pt height9pt}
{\vrule depth3pt height9pt}{\vrule depth3pt height9pt}\,}

\DeclareMathOperator \vol{{vol}}

\DeclareMathOperator{\op}{{Op}}

\DeclareMathOperator {\ad}{{ad}}

\newtheorem{thm}{Theorem}[section]
\newtheorem{cor}[thm]{Corollary}
\newtheorem{cla}[thm]{Claim}
\newtheorem{lem}[thm]{Lemma}
\newtheorem{prop}[thm]{Proposition}

\theoremstyle{definition}
\newtheorem{defn}[thm]{Definition}

\newtheorem{rem}[thm]{Remark}

\numberwithin{equation}{section}

\newcommand{\bee}{\begin{equation}}
\newcommand{\ene}{\end{equation}}
\newcommand{\bees}{\begin{equation*}}
\newcommand{\enes}{\end{equation*}}
\newcommand{\bes}{\begin{split}}
\newcommand{\ens}{\end{split}}

\newcommand{\bet}{\begin{thm}}
\newcommand{\ent}{\end{thm}}
\newcommand{\bel}{\begin{lem}}
\newcommand{\enl}{\end{lem}}
\newcommand{\bec}{\begin{cor}}
\newcommand{\enc}{\end{cor}}
\newcommand{\becl}{\begin{cla}}
\newcommand{\encl}{\end{cla}}
\newcommand{\bep}{\begin{proof}}
\newcommand{\enp}{\end{proof}}
\newcommand{\ber}{\begin{rem}}
\newcommand{\enr}{\end{rem}}
\newcommand{\ep}{\varepsilon}
\newcommand{\la}{\lambda}
\newcommand{\La}{\Lambda}
\newcommand{\de}{\delta}
\newcommand{\al}{\alpha}
\newcommand{\Z}{\mathbb Z}
\newcommand{\Ga}{\Gamma}

\newcommand {\BUps}{\boldsymbol\Upsilon}

\newcommand{\CF}{\mathcal F}

\makeatletter
\def\square{\RIfM@\bgroup\else$\bgroup\aftergroup$\fi
  \vcenter{\hrule\hbox{\vrule\@height.6em\kern.6em\vrule}\hrule}\egroup}
\makeatother

\begin{document}

\hoffset -4pc

\title[Density of states]
{Complete asymptotic expansion of the integrated density of states of multidimensional almost-periodic Schr\"odinger operators}
\author[L. Parnovski \& R. Shterenberg]
{Leonid Parnovski \& Roman Shterenberg}
\address{Department of Mathematics\\ University College London\\
Gower Street\\ London\\ WC1E 6BT\\ UK}
\email{Leonid@math.ucl.ac.uk}
\address{Department of Mathematics\\ University of Alabama at Birmingham\\ 1300 University Blvd.\\
Birmingham AL 35294\\ USA}
\email{shterenb@math.uab.edu}

\keywords{Periodic operators, almost-periodic pseudodifferential operators, integrated density of states}
\subjclass[2000]{Primary 35P20, 47G30, 47A55; Secondary 81Q10}

\date{\today}


\begin{abstract}
We prove the complete asymptotic expansion of the integrated density of states of a
Schr\"odinger operator $H=-\Delta+b$ acting in $\R^d$ when the potential $b$ is either
smooth periodic, or generic quasi-periodic (finite linear combination of exponentials),
or belongs to a wide class of almost-periodic functions.
\end{abstract}

\



\maketitle
\vskip 0.5cm

\renewcommand{\uparrow}{{\mathcal {L F}}}
\newcommand{\ssharp}{{\mathcal {L E}}}
\renewcommand{\natural}{{\mathcal {NR}}}
\renewcommand{\flat}{{\mathcal R}}
\renewcommand{\downarrow}{{\mathcal {S E}}}

\section{Introduction}

We consider the Schr\"odinger operator
\bee\label{eq:Sch}
H=-\Delta+b
\ene
acting in $L_2(\R^d)$. The potential $b=b(\bx)$ is assumed to be real, smooth, and either periodic, or almost-periodic;
in the almost-periodic case we assume that all the derivatives of $b$ are almost-periodic as well.
We are interested in the asymptotic behaviour of the (integrated) density of states $N(\la)$ as
the spectral parameter $\la$ tends to infinity. The density of states of $H$
can be defined by the formula
\bee\label{eq:def1}
N(\lambda)=N(\la;H): = \lim_{L\to\infty}
\frac{N(\lambda;H^{(L)}_D)}
{(2L)^d}.
\ene
Here, $H^{(L)}_D$ is the restriction of $H$ to the cube $[-L,L]^d$ with the Dirichlet boundary
conditions, and $N(\lambda;A)$ is the counting function of the discrete spectrum of $A$.
Later, we will give equivalent definitions of $N(\la)$ which are more convenient to work with.
If we denote by $N_0(\lambda)$ the density of states of the unperturbed operator
$H_0=-\Delta$, one can easily see that for positive $\la$ one has
\bee
N_0(\lambda)=C_d\la^{d/2},
\ene
where
\bee
C_d=\frac{w_d}{(2\pi)^d} \ \text{and} \ w_d=\frac{\pi^{d/2}}{\Gamma(1+d/2)}
\ene
is the volume of the unit ball in $\R^d$. There is a long-standing conjecture that, at least in the case of periodic $b$,
the density of states of $H$
enjoys the following asymptotic behaviour as $\la\to\infty$:
\bee\label{eq:intr0}
N(\lambda)\sim \la^{d/2}\Bigl(C_d+\sum_{j=1}^\infty e_j\lambda^{-j}\Bigr),
\ene
meaning that for each $K\in\N$ one has
\bee\label{eq:intr1}
N(\lambda)=\la^{d/2}\Bigl(C_d+\sum_{j=1}^K e_j\lambda^{-j}\Bigr)+R_K(\lambda)
\ene
with $R_K(\lambda)=o(\la^{\frac{d}{2}-K})$. In those formulas, $e_j$ are real numbers which depend on the potential $b$.
They can be calculated relatively easily using the heat kernel invariants (computed in \cite{HitPol}); they are equal to certain
integrals of the potential $b$ and its derivatives.
Indeed, in the paper \cite{KorPush}, all these coefficients were computed; in particular, it turns out
that, if $d$ is even, then $e_j$ vanish whenever $j>d/2$.

Until recently, formula \eqref{eq:intr0} was proved only in the case $d=1$ in \cite{ShuSche} for periodic $b$ and
in \cite{Sav} for almost-periodic $b$. In the recent paper \cite{ParSht}, this formula was proved in the case
$d=2$ and periodic potential.
In the periodic case and $d\ge 3$, only partial results are known, see \cite{HelMoh}, \cite{Kar},
\cite{Kar1}, \cite{ParSob}, \cite{Skr}. In particular, in \cite{Kar} it was shown that
formula \eqref{eq:intr1} is valid with $K=1$ and $R(\la)=O(\lambda^{-\de})$ with some small positive $\de$ when
$d= 3$
and $R(\la)=O(\lambda^{\frac{d-3}{2}}\ln\la)$ when $d>3$. Finally, in the multidimensional almost-periodic case,
formula \eqref{eq:intr1} is known only with $K=0$ and $R(\la)=O(\la^{\frac{d-2}{2}})$, see \cite{Shu}.

The aim of our paper is to prove formula \eqref{eq:intr1} with arbitrary $K$ for all dimensions $d$ and
for periodic or almost-periodic potentials. In the case of periodic potential, we do not impose any additional
assumptions (besides infinite smoothness) on it. However, if the potential $b$ is almost-periodic, we need it
to satisfy certain extra conditions; since the formulation of them requires several definitions,
we will list these conditions and formulate our main result in the next section.

Now we discuss the difference in the approaches of \cite{ParSht} and this paper. To begin with, let us assume
that the potential $b$ is periodic. Then we can perform the Floquet-Bloch decomposition (see, e.g., \cite{ReeSim})
and express the operator $H$ as a direct integral
\bee\label{eq:decper}
H=\int_{\oplus}H(\bk)d\bk,
\ene
quasi-momentum $\bk$ running over $\CO^{\dagger}$ --
the cell of the lattice $\Ga^{\dagger}$, dual to the lattice of periods $\Ga$. The very first thing
we need to do is to replace definition  \eqref{eq:def1} with
a different one. There are two problems with definition  \eqref{eq:def1}. The first problem is that this definition is rather
difficult to work with. The second problem is that this definition makes sense only for differential operators. And, although
we are working with a differential operator \eqref{eq:Sch} in the beginning, our methods require us to replace this operator with a pseudodifferential
one, and so we need a definition which works for pseudodifferential operators as well.
In the periodic case, the alternative definition is given by the formula
\bee\label{eq:densityperiodic}
N(\lambda):=\frac{1}{(2\pi)^d}\int_{\CO^{\dagger}}N(\la,H(\bk))d\bk,
\ene
where $N(\la,H(\bk))$ is the eigenvalue counting function of $H(\bk)$.
The first
step of obtaining information on the density of states is to compute the precise asymptotics of the eigenvalues of
$H(\bk)$. There are two different approaches to doing this. The first method, called the method of spectral projections,
was developed in \cite{Par}. When using this method,  we
study, instead of $H$, the operator $\tilde H=\sum_j P_jHP_j$, where $\{P_j\}$ are spectral projections
of the unperturbed operator $H_0:=-\Delta$. It was shown in \cite{Par} that, if we carefully choose these projections, then
the spectra of $H(\bk)$ and $\tilde H(\bk)$ are close to each other. Next, we can decompose the operator $\tilde H$ into
invariant subspaces. There are two types of such subspaces. The first type (called stable, or non-resonant subspaces) corresponds to
eigenvalues of $H$ which are far away from other eigenvalues; in studying them we can use straightforward perturbation theory
to compute their precise asymptotic behaviour. This was done in \cite{Par}, but in certain cases
such computations were performed earlier (see, for example, \cite{Kar1} and \cite{Vel}). The second type of
subspaces (called unstable, or resonant) corresponds to clusters of eigenvalues of $H$ lying close to each other.
In order to study these eigenvalues, we have to use perturbation theory of multiple eigenvalues, and this theory is much more difficult
and less precise than in the stable case. The methods of \cite{Par} allow us to reduce the study of resonant eigenvalues to the study of a family
of operators $A+\ep B$ when $\ep\to 0$. Here, $A$ and $B$ are finite-dimensional self-adjoint operators and $\ep\sim \la^{-1/2}$
is a small parameter. We are interested in the eigenvalues of $A+\ep B$ which are perturbations of zero eigenvalues of $A$. Of course,
we can write the formula $\la(A+\ep B)\sim\sum \la_j\ep^j$, see \cite{Kat}, but the coefficients $\la_j$ will, in general, be unbounded
functions of the quasi-momentum and, therefore, we cannot integrate these asymptotic expansions against $d\bk$. Paper
\cite{ParSht} deals with this problem in the case $d=2$. We study the operator $PBP$, where $P$ is the orthogonal projection
onto the kernel of $A$, and show that the cluster of eigenvalues of this operator has multiplicity at most two. This allows us,
using the Weierstrass Preparation Theorem, to prove that the eigenvalues of $A+\ep B$ enjoy the asymptotic formula
$\la(A+\ep B)\sim\sum \la_j\ep^j\pm\sqrt{\sum \tilde\la_j\ep^j}$, where the coefficients $\la_j$ and $\tilde\la_j$ are bounded
functions of the quasi-momentum $\bk$ and so can be integrated against $d\bk$. Unfortunately, this approach does not work if $d\ge 3$, since
then the cluster multiplicity of $PBP$ becomes unbounded.

The second method of obtaining asymptotic formulas for the eigenvalues of $H(\bk)$ was developed in \cite{Sob} and \cite{Sob1} and
was also used in \cite{ParSob}. This method, which we call the gauge transform method, consists of constructing two pseudodifferential operators,
$H_1$ and $H_2$. Here, $H_1=e^{i\Psi}H e^{-i\Psi}$, where $\Psi$ is a bounded periodic self-adjoint pseudo-differential operator of order $0$.
Thus, the eigenvalues of $H_1(\bk)$ coincide with the eigenvalues of $H(\bk)$. The operator $H_2$ is close to $H_1$ in norm; also,
operators $H_2(\bk)$ have a lot of invariant subspaces. As in the previous method, these invariant subspaces can be generated by
stable and unstable eigenvalues, and the case of the stable eigenvalues can be treated completely (i.e. the complete asymptotic formula for such
eigenvalues can be obtained). The difference with the previous approach lies in the form of the restriction of the operator $H_2$ to a
subspace generated by a resonant eigenvalue. Let $\bxi$ be a point in the phase space lying in a resonant region generated by a lattice subspace $\GV$
(see Section 5 for the definitions and more details).
Then $H_2$, restricted to the invariant subspace generated by $\bxi$, has a form $r^2I+S(r)$,
where $r$ is, essentially, the distance from $\bxi$
to $\GV$ and $S(r)$ is a finite-dimensional self-adjoint operator which can grow in $r$, but slower than $r^2$. Using the method of spectral projections, we could achieve that $S(r)$ has a simple form, namely, $S(r)=rA+B=r(A+\ep B)$, where $\ep=r^{-1}$, which was the advantage of that method.
The advantage of the method of gauge transform is that all eigenvalues of the reduced operator  $r^2I+S(r)$ contribute to the integrated density
of states, whereas in the method of spectral projections only the eigenvalues coming from the zero eigenvalues of $A$ (and not even all such eigenvalues)
were of interest to us. This observation makes the method of gauge transform much more convenient to use, despite the operator $S(r)$ being more complicated
than in the method of spectral projections. Indeed, the fact that all the eigenvalues contribute to the density of states allows us
to use the residue theorem in order to compute the sum of contributions
from all eigenvalues without computing the contributions from individual eigenvalues, see \eqref{eq:n4} and \eqref{eq:residues}. In fact, formula \eqref{eq:residues} is the most crucial
observation which has enabled us to compute the contribution to the density of states from the resonance regions.

When we were working on the details of this approach, we have realized that, as a matter of fact, 
the decomposition \eqref{eq:decper} is not required, and all
the steps can be written for the `global' operators $H$ without any references to the `fibre' operators $H(\bk)$. This led us to believe that this
method is likely to be applicable in a range of other settings. In particular, it turned out that this method works for quasi-periodic and almost-periodic potentials.
So, let us assume that the potential $b$ is almost-periodic. What is the analogue of definition \eqref{eq:densityperiodic} in this case? The answer to this
question can be found in \cite{Shu}. One possible way of defining the density of states is via the von Neumann algebras; we discuss this approach later in our paper.
However, the ultimate definition is as follows: let $e(\la;\bx,\by)$ be the kernel of the spectral projection of an elliptic pseudo-differential operator
of positive order with almost-periodic coefficients. Then, it was proved in Theorem 4.1 of \cite{Shu}, that the density of states of this operator satisfies
(at least at its continuity points):
\bee\label{eq:densityalmost}
N(\la)=\BM_{\bx}(e(\la;\bx,\bx)),
\ene
where $\BM$ is the mean of an almost-periodic function. Our main tool in the proof will be formula \eqref{eq:densityalmost}, but we will use the operator-algebraic definition
sometimes (for example, to show that density of states decreases when the operator increases, something not immediately obvious from \eqref{eq:densityalmost}).
Another useful observation which helped us in extending our results to the almost-periodic case is this. Let $A$ be an  elliptic pseudo-differential operator
with almost-periodic coefficients. We are usually assuming that $A$ acts in $L_2(\R^d)$. However, we can consider actions of $A$ (via the same
Fourier integral operator formula) in different vector spaces, for example in the Besicovitch space $B_2(\R^d)$. The space $B_2(\R^d)$ is the space of all formal
sums
\bees
\sum_{j=1}^\infty a_{j}\be_{\bth_j}(\bx),
\enes
where
\bee
\be_{\bth}(\bx):=e^{i\bth\bx}
\ene
and $\sum_{j=1}^\infty |a_{j}|^2<+\infty$. It is known (see \cite{Shu0}) that the spectra of $A$ acting in $L_2(\R^d)$ and $B_2(\R^d)$ are the same, although the types of
those spectra can be entirely different. It is very convenient, when working with the gauge transform constructions, to assume that all the operators involved
act in $B_2(\R^d)$, although in the end we will return to operators acting in $L_2(\R^d)$. This trick (working with operators acting in $B_2(\R^d)$)
is similar to working with fibre operators $A(\bk)$ in the periodic case in a sense that we can freely consider the action of an operator on one, or finitely many,
exponentials, without caring that these exponentials do not belong to our function space.

It seems likely that the approach of this paper can be applied to a wider class of operators than \eqref{eq:Sch}. The operators should be of the type $H=H_0+b$, where
$H_0$ has constant coefficients, and $b$ has order smaller than $H_0$. We plan to consider such operators in a subsequent publication.

Now we describe the structure of this paper.
The proof of our main theorem consists of several parts, which are not always immediately related to each other; in particular, there
is no natural order in which these parts should be presented. As a result, it is possible to read different sections of our paper in almost
arbitrary order. The main principles we were following in determining the actual order of the sections were: trying to postpone the most difficult
and technical parts of the proof for as long as possible, and trying to minimize the amount of references to definitions/results stated after the reference.
In particular, Section 6 of this paper can be considered as a further general discussion of our approach
which we have decided to postpone until the definitions and results of Sections 2--5 have been introduced.
In Section 2, we give some basic
definitions, formulate the conditions we impose on the potential and state the main result. In Section 3, we explain why, instead of proving \eqref{eq:intr1}, it would be sufficient to prove a more general asymptotic formula \eqref{eq:main_thm1} (which includes more
powers of $\la$ as well as logarithms). We also explain why
it is enough to prove this asymptotic formula not for all large $\la$, but only for $\la$ inside
a fixed interval. The proof of these statements (as well as reasons why we need them) is similar to the corresponding section in \cite{ParSht}. In Section 4, we describe
the definition of the density of states based on the operator algebraic constructions and prove several useful properties of $N(\la)$ which immediately follow from
these constructions. In Section 5, we define resonance regions and prove their properties. The reader who has read several of the papers  \cite{Sob1}, \cite{Par}, \cite{BarPar}, \cite{ParSht}, \cite{ParSob} may have noticed that in each of these papers the construction of the resonance regions is slightly different. The reason is that
each time we define these regions, we need to fine tune the definition taking care of the problem we are trying to solve. Our present paper is not an exception,
and the construction of the resonance regions in Section 5 is different from the constructions in all papers mentioned above. This new construction will be extremely
convenient when we are going to integrate the contribution from individual eigenvalues to the density of states. In Section 6, we describe this procedure of integrating
the contribution from individual eigenvalues over the resonance zones in more detail. In Section 7, we introduce the coordinates in each resonance region (or rather we
cut each resonance region into pieces and introduce coordinates in each piece). These coordinates are introduced so that the integration, described in Section 6, will be
as painless as it possibly can. Each resonance region will have two types of coordinates. The first type is Cartesian coordinates in $\GV$, where $\GV$ is the quasi-lattice
subspace generating the resonance region. The second set of coordinates is the shifted polar coordinates in $\GV^{\perp}$. These coordinates are ideologically similar
to the shifted polar coordinates we have introduced in \cite{ParSht}, but the details are much more complicated now. Starting from Section 7, until the end of Section 10, we will assume that all the regions where the integration takes place are of the simplest possible type (the simplex case).
In Sections 8 and 9, we discuss the main tool
of this paper, the gauge transform method. A large proportion of the material contained in these two sections is similar to the relevant parts of \cite{ParSob}, the only difference being definition  \eqref{1b1:eq} (we need to change the norm to accommodate it to the case of almost-periodic coefficients) and Lemma \ref{lem:symbol} (this lemma was not required
in \cite{ParSob}). In Section 10, we compute the contribution to the density of states from each resonance region and, first, reduce this contribution to
the explicit integral  \eqref{eq:in2} and then, in Lemma \ref{lem:integral1}, prove that this integral admits a decomposition in the powers of $\la$ and logarithms. Finally, in Section 11, we discuss how to reduce integration over the region of arbitrary shape to the
simplex case.

{\bf Acknowledgement.} We are very grateful to Michael Levitin for running computer experiments which eventually led us
to the proper understanding of formula \eqref{eq:residues} and to Eugene Shargorodsky for useful discussions. Thanks also go to
Gerassimos Barbatis and Sergey Morozov who have read the preliminary version of this paper and made several useful comments and suggestions.
The work of the first author
was partially supported by the EPSRC grant EP/F029721/1. The second author was partially supported by the NSF grant DMS-0901015.

\section{Preliminaries}

Since our potential $b$ is almost-periodic, it has the Fourier series
\bee\label{eq:potential}
b(\bx)\sim\sum_{\bth\in\Bth}a_{\bth}\be_{\bth}(\bx),
\ene
where
\bee
\be_{\bth}(\bx):=e^{i\bth\bx},
\ene
and $\Bth$ is a (countable) set of frequencies.
Without loss of generality we assume that $\Bth$ spans $\R^d$, and contains $0$ and is symmetric about $0$; we also put
\bee\label{eq:algebraicsum}
\Bth_k:=\Bth+\Bth+\dots+\Bth
\ene
(algebraic sum taken $k$ times) and $\Bth_{\infty}:=\cup_k\Bth_k=Z(\Bth)$, where for a set $S\subset \R^d$ by $Z(S)$ we denote the set of all finite linear combinations of elements in
$S$ with integer coefficients. The set $\Bth_\infty$ is countable and non-discrete (unless the potential $b$ is periodic).
The first condition we impose on the potential is:

{\bf Condition A}. Suppose that $\bth_1,\dots,\bth_d\in \Bth_\infty$. Then $Z(\bth_1,\dots,\bth_d)$ is discrete.

It is easy to see that this condition can be reformulated like this:
suppose, $\bth_1,\dots,\bth_d\in \Bth_\infty$.
Then
either $\{\bth_j\}$ are linearly independent, or $\sum_{j=1}^d n_j\bth_j=0$, where $n_j\in\Z$ and not all
$n_j$ are zeros. This reformulation shows that Condition A is generic: indeed, if we are choosing frequencies
of $b$ one after the other, then on each step we have to avoid choosing a new frequency from a countable set of
hyperplanes, and this is obviously a generic restriction. Condition A is obviously satisfied for periodic
potentials, but it becomes meaningful for quasi-periodic potentials 
(we call a function quasi-periodic, if it is a linear combination of finitely many exponentials).

The rest of the conditions we have to impose describe how well we can approximate the potential $b$
by means of quasi-periodic functions. In the proof we are going to work
with quasi-periodic approximations of $b$, and we need these conditions to make sure that
all estimates in the proof are uniform with respect to these approximations.

{\bf Condition B}.
Let $k$ be an arbitrary fixed natural number. Then for each sufficiently large real number $\rho$
there is a finite set $\Bth(k;\rho)\subset(\Bth\cap B(\rho^{1/k}))$ (where $B(r)$ is a ball of radius $r$ centered at
$0$) and a `cut-off' potential
\bee\label{eq:condB1}
b_{(k;\rho)}(\bx):=\sum_{\bth\in\Bth(k;\rho)}\tilde a_{\bth}\be_{\bth}(\bx)
\ene
which satisfies
\bee\label{eq:condB2}
||b-b_{(k;\rho)}||_{\infty}<\rho^{-k}.
\ene

\ber
First of all, notice that we can reformulate this condition like this: for each (small) $\al>0$ and (small) $\epsilon>0$
there is a `cut-off' potential $b_{(\al,\epsilon)}$ so that $||b-b_{(\al;\epsilon)}||_{\infty}<\epsilon$ and
the frequencies of $b_{(\al;\epsilon)}$ lie inside the ball of radius $\epsilon^{-\al}$. However, it will be rather
more convenient in what follows to have Condition B formulated in terms of $k$ and $\rho$.
This condition is obviously satisfied for quasi-periodic potentials; for periodic potentials it is equivalent to the
infinite smoothness. For almost-periodic potentials Condition B does not seem to follow from the infinite smoothness of $b$.
Note that we do not require the coefficients $\tilde a_{\bth}$ to be equal to the `old'
coefficients $a_{\bth}$; indeed, sometimes one can find a better approximation by using procedures different than the trivial
`chopping off' of $b$, like, for example, the Bochner-Fej\'er summation.
\enr

The next condition we need to impose is a version of the Diophantine condition on the frequencies of $b$. First, we need
some definitions. We fix a natural number $\tilde k$ (the choice of $\tilde k$ will be determined later by how many terms in \eqref{eq:intr1}
we want to obtain) and denote $\tilde\Bth:=[\Bth(k;\rho)]_{\tilde k}$
(see \eqref{eq:algebraicsum} for the notation) and
$\tilde\Bth':=\tilde\Bth\setminus\{0\}$.
We say that $\GV$ is a quasi-lattice subspace of dimension $m$, if $\GV$ is a linear
span of $m$ linear independent vectors $\bth_1,\dots,\bth_m$ with $\bth_j\in\tilde\Bth\ \forall j$. Obviously, the zero
space (which we will denote by $\GX$)
is a quasi-lattice subspace of dimension $0$ and $\R^d$ is a quasi-lattice subspace of dimension $d$.
We denote by $\CV_m$ the collection of all quasi-lattice subspaces of dimension $m$ and put
$\CV:=\cup_m\CV_m$.
If $\bxi\in\R^d$
and $\GV$ is a linear subspace of $\R^d$, we denote by $\bxi_{\GV}$ the orthogonal projection of $\bxi$
onto $\GV$, and put $\GV^\perp$ to be an orthogonal complement of $\GV$, so that $\bxi_{\GV^\perp}=\bxi-\bxi_{\GV}$.
Let $\GV,\GU\in\CV$. We say that these subspaces are {\it strongly distinct}, if neither of them is a
subspace of the other one. This condition is equivalent to stating that if we put $\GW:=\GV\cap\GU$, then
$\dim \GW$ is strictly less than dimensions of $\GV$ and $\GU$. We put $\phi=\phi(\GV,\GU)\in [0,\pi/2]$
to be the angle between them, i.e. the angle between $\GV\ominus\GW$ and $\GU\ominus\GW$, where $\GV\ominus\GW$
is the orthogonal complement of $\GW$ in $\GV$. This angle is positive iff $\GV$ and $\GW$ are strongly distinct.
We put $s=s(\rho)=s(\tilde\Bth):=\inf\sin(\phi(\GV,\GU))$, where infimum is over all strongly distinct pairs of subspaces from $\CV$,
$R=R(\rho):=\sup_{\bth\in\tilde\Bth}|\bth|$, and $r=r(\rho):=\inf_{\bth\in\tilde{\Bth}'}|\bth|$. Obviously, $R(\rho)\ll \rho^{1/k}$ (where the
implied constant can depend on $k$ and $\tilde k$; we say that $f\ll g$ if $f=O(g)$).

{\bf Condition C}. For each fixed $k$ and $\tilde k$ the sets $\Bth(k;\rho)$ satisfying \eqref{eq:condB1} and \eqref{eq:condB2}
can be chosen in such a way that for sufficiently large $\rho$ we have
\bee\label{eq:condC1}
s(\rho)\ge\rho^{-1/k}
\ene
and
\bee\label{eq:condC2}
r(\rho)\ge\rho^{-1/k},
\ene
where the implied constant (i.e. how large should $\rho$ be) can depend on $k$ and $\tilde k$.

\ber\label{rem:condB}
First of all, we remark that condition \eqref{eq:condC2} for $\tilde k$ can be derived from condition
\eqref{eq:condC1} for $\tilde k+1$, but we prefer to postulate both conditions. We also note that Condition C
is automatically satisfied for quasi-periodic potentials; for smooth periodic potentials
Condition C is also automatically satisfied
(see, for example, \cite{Par}). Finally, notice that condition \eqref{eq:condC1} is equivalent to $s(\rho)\ge\rho^{-\alpha/k}$
for any fixed positive $\alpha$ (indeed, this equivalence can be proved by considering sets $\Bth([\alpha^{-1} k];\rho)$ instead of $\Bth(k;\rho)$
in Condition B, since 
Condition B holds for all $k$). Thus, if we consider potentials of the form
$b=b_{per}+b_{qua-per}$, where $b_{per}$ is smooth periodic and
$b_{qua-per}$ is quasi-periodic, Condition C amounts to the Diophantine condition on the frequencies of $b_{qua-per}$
and is generic.
\enr

Condition A implies the following
statement.
Suppose, $\bth_1,\dots,\bth_{l}\in \tilde\Bth$, $l\leq d-1$.
Let $\GV$ be the span of $\bth_1,\dots,\bth_{l}$. Then each element of the set
$\tilde\Bth\cap\GV$ is a linear combination of $\bth_1,\dots,\bth_{l}$ with rational coefficients. Since the
set $\tilde\Bth\cap\GV$ is finite, this implies that the set
$Z(\tilde\Bth\cap\GV)$
is discrete and is, therefore, a lattice in $\GV$. We denote this lattice by $\Gamma(\rho;\GV)$.
Our final condition states that this lattice cannot be too dense.

{\bf Condition D}. We can choose $\Bth(k;\rho)$
satisfying conditions B and C
in such a way that for sufficiently large $\rho$ and for each $\GV\in\CV$, $\GV\ne\R^d$, we have
\bee\label{eq:condD}
\vol(\GV/\Gamma(\rho;\GV))\ge\rho^{-1/k}.
\ene

\ber\label{rem:condC} As with Condition C, Condition D is satisfied for quasi-periodic and smooth periodic potentials.
Also, similarly to Remark \ref{rem:condB}, condition \eqref{eq:condD} is equivalent to $\vol(\GV/\Gamma(\rho;\GV))\ge\rho^{-\alpha/k}$.
Condition D is not essential
for our methods and it is likely that this condition can be relaxed. Indeed, the only place we are using this condition is to get an
upper bound on the number of elements in $\BUps(\bxi)$, and this estimate, in turn, is used only to prove \eqref{eq:residues}. However, it seems likely that there may be another way of proving that LHS and RHS of  \eqref{eq:residues}
are the same. This alternative proof is more direct and much more difficult technically. Given that our paper is quite technically
involved the way it is now, we have decided to present a proof which is considerably simpler, paying the price of assuming a slightly
stronger condition on the potential.
\enr
\ber
One final remark that concerns all conditions B--D. Given any symmetric set $\Bth$ of frequencies, we can construct a real smooth almost-periodic potential $b$
such that \eqref{eq:potential} holds, all Fourier coefficients $a_{\bth}$ are non-zero, and  conditions B--D are satisfied (of course, the Fourier coefficients
will have to converge to zero really fast).
For example, if $b$ is a limit-periodic function with Fourier coefficients going to zero exponentially, than all our conditions
A--D are satisfied.
\enr

Now we can formulate our main theorem.

\bet\label{thm:thm}
Let $H$ be an operator \eqref{eq:Sch} with real smooth almost-periodic potential $b$ satisfying Conditions {\rm A,B,C,} and {\rm D}.
Then for each $K\in\N$ we have:
\bee\label{eq:main_thm0}
N(\la)=\la^{d/2}\left(C_d+\sum\limits_{j=1}^{K}e_j\la^{-j}+o(\la^{-K})\right)
\ene
as $\la\to\infty$.
\ent
\ber
Following \cite{HitPol}, \cite{HitPol1}, and \cite{KorPush}, it is straightforward to compute the coefficients $e_j$. For
example, we have
\bees
e_1=-\frac{d w_d}{2(2\pi)^d}\BM(b)
\enes
and
\bees
e_2=
\frac{d(d-2) w_d}{8(2\pi)^d}\BM(b^2),
\enes
where $\BM$ is the mean of an almost-periodic function.
\enr
From now on, we always assume that our potential satisfies all the conditions from this section; we also will denote $\rho:=\sqrt{\la}$.
Given Conditions B--D, we want to introduce the following definition. We say that a positive function
$f=f(\rho)=f(\rho;k,\tilde k)$ satisfies the estimate
$f(\rho)\le\rho^{0+}$ (resp. $f(\rho)\ge\rho^{0-}$), if for each positive $\ep$ and for each $\tilde k$ we can achieve $f(\rho)\le\rho^{\ep}$ (resp. $f(\rho)\ge\rho^{-\ep}$)
for sufficiently large $\rho$ by choosing parameter $k$ from Conditions B--D
sufficiently large. For example, we have $R(\rho)\le\rho^{0+}$, $s(\rho)\ge\rho^{0-}$, $r(\rho)\ge\rho^{0-}$, and $\vol(\GV/\Gamma(\rho;\GV))\ge\rho^{0-}$.
One can also use a standard covering argument to show that the number of elements in $\Bth(k;\rho)$ satisfies $|\Bth(k;\rho)|\le\rho^{0+}$.
Throughout the paper, we always assume that the value of $k$ is chosen
sufficiently large so that all inequalities of the form $\rho^{0+}\le\rho^{\ep}$ or $\rho^{0-}\ge\rho^{-\ep}$
we encounter in the proof are satisfied.
\ber
As we will mention several times in this paper, the very big problem for both the authors and the readers
is the amount of notation one has to keep
in mind. The above definition is the first step in our desire to make as much as possible of the notation obsolete and eventually stop using it.
\enr
The next statement shows a bit more how this new notation is used.
\bel\label{lem:coefficients}
Suppose, $\bth, \bmu_1,\dots,\bmu_d\in\tilde\Bth'$, the set $\{\bmu_j\}$ is linearly independent, and $\bth=\sum_{j=1}^db_j\bmu_j$. Then each non-zero coefficient
$b_j$ satisfies
\bee\label{eq:coefficients}
\rho^{0-}\le |b_j| \le \rho^{0+}.
\ene
\enl
\bep
Let $\GV\in\CV_{d-1}$ be a subspace spanned by $\bmu_j$, $j=2,\dots, d$, and let $\be$
be a unit vector orthogonal to $\GV$.
Then the sine of the angle between $\bth$ and $\GV$ is $\lu\bth,\be\ru|\bth|^{-1}$. Thus, if this angle is non-zero, we have
$|\lu\bth,\be\ru|\ge s(\rho) r(\rho)$ and, hence, if $b_1=\lu\bth,\be\ru\lu\bmu_1,\be\ru^{-1}$ is non-zero, it satisfies
$|b_1|\ge r(\rho) s(\rho) R(\rho)^{-1}\ge\rho^{0-}$. Similarly, since $\lu\bmu_1,\be\ru\ne 0$, we have $|\lu\bmu_1,\be\ru|\ge s(\rho) r(\rho)$ and thus
$|b_1|\le R(\rho) (r(\rho) s(\rho))^{-1}\le \rho^{0+}$. The proof for $j\ne 1$ is similar.
\enp

In this paper, by
$C$ or $c$ we denote positive constants,
the exact value of which can be different each time they
occur in the text,
possibly even each time they occur in the same formula. On the other hand, the constants which are labeled (like $C_1$, $c_3$, etc)
have their values being fixed throughout the text.
Given two positive functions $f$ and $g$,
we say that $f\gg g$, or $g\ll
f$, or $g=O(f)$ if the ratio $\frac{g}{f}$ is bounded. We say
$f\asymp g$ if $f\gg g$ and $f\ll g$.

\section{Reduction to a finite interval of spectral parameter}

The main result of our paper, Theorem \ref{thm:thm}, will follow from the following theorem (
recall that we put
$\rho:=\sqrt{\la}$):

\bet\label{main_thm}
For each $K\in\N$ we have:
\bee\label{eq:main_thm1}
N(\rho^2)=C_d\rho^d+\sum_{p=0}^{d-1}\sum_{j=-d+1}^{K}e_{j,p}\rho^{-j}(\ln\rho)^p
+o(\rho^{-K})
\ene
as $\rho\to\infty$. 
\ent
Once the theorem is proved, it immediately implies
\bec
For each $K\in\N$ we have:
\bee\label{eq:main_cor1}
N(\la)=\la^{d/2}\left(C_d+\sum\limits_{j=1}^{K}e_j\la^{-j}+o(\la^{-K})\right)
\ene
as $\la\to\infty$.
\enc
\bep
First of all, we notice that \cite{HitPol}, \cite{HitPol1} and formula (2.9) from \cite{Shu} imply that
\bee\label{Laplace}
\int_{-\infty}^\infty e^{-t\la}N(\la)d\la\sim t^{-(d+2)/2}\sum_{j=0}^\infty q_jt^j
\ene
as $t\to 0+$, where $q_j$ are constants depending on the potential.
Now the corollary follows from Theorem \ref{main_thm} and calculations similar to that of
\cite{KorPush}. Indeed, consider the following integrals:
\bee
I_1(t;k,p):=\int_1^\infty e^{-t\la}\la^{-k}(\ln\la)^p d\la,\ \ \ p\in\Z_+,\ k\in\Z;
\ene
\bee
I_2(t;k,p):=\int_1^\infty e^{-t\la}\la^{-k-\frac12}(\ln\la)^p d\la,\ \ \ p\in\Z_+,\ k\in\Z.
\ene
Elementary calculations show that
\bee\label{Laplace1}
I_1(t;k,p)=t^{k-1}\left(\Gamma(-k+1)\left(\ln\frac1t\right)^p+\sum\limits_{j=0}^{p-1}a_j\left(\ln\frac1t\right)^j\right)+f_1(t),\ \ \ \hbox{for}\ k\leq0,\ t>0;
\ene
\bee\label{Laplace2}
I_1(t;k,p)=t^{k-1}\left(\frac{1}{p+1}\frac{(-1)^{k-1}}{(k-1)!}\left(\ln\frac1t\right)^{p+1}+
\sum\limits_{j=0}^{p}a'_j\left(\ln\frac1t\right)^j\right)+f_2(t),\ \ \ \hbox{for}\ k\geq1,\ t>0;
\ene
\bee\label{Laplace3}
I_2(t;k,p)=t^{k-\frac12}\left(\Gamma(-k+\frac12)\left(\ln\frac1t\right)^p+
\sum\limits_{j=0}^{p-1}a''_j\left(\ln\frac1t\right)^j\right)+f_3(t),\ \ \ \hbox{for any}\ k\in\Z,\ t>0.
\ene
Here, $a_j=a_j(k,p),\ a'_j=a'_j(k,p),\ a''_j=a''_j(k,p)$ are some constants and $f_j(t)=f_j(t;k,p)$ are {\it entire} functions in $t$. Obviously, $\int_{-\infty}^1 e^{-t\la}N(\la)d\la$ is an entire function in $t$. Comparing \eqref{Laplace} and \eqref{Laplace1}, \eqref{Laplace2}, \eqref{Laplace3}, it is not difficult to see that

1. if $d$ is even then $e_{j,p}$ can be non-zero only if $p=0$ and $j$ is non-positive and even;

2. if $d$ is odd then $e_{j,p}$ can be non-zero only if $p=0$ and $j$ is odd.

\enp

Thus, we can concentrate on proving Theorem \ref{main_thm}.

To begin with, we choose sufficiently large $\rho_0>1$ (to be fixed later on) and put $\rho_n=2\rho_{n-1}=2^n\rho_0$,
$\la_n:=\rho_n^2$;
we also define the interval
$I_n=[\rho_n,4\rho_n]$.
The proof of Theorem \ref{main_thm} will be based on the following lemma:

\bel\label{main_lem}
For each $M\in\N$ and $\rho\in I_n$ we have:
\bee\label{eq:main_lem1}
N(\rho^2)=C_d\rho^d+\sum_{p=0}^{d-1}\sum_{j=-d+1}^{6M}e_{j,p}(n)\rho^{-j}(\ln\rho)^p
+O(\rho_n^{-M}).
\ene
Here, $e_{j,p}(n)$ are some real numbers depending on $j,\ p$ and $n$ (and $M$) satisfying
\bee\label{eq:main_lem2}
e_{j,p}(n)=O(\rho_n^{(2j/3)+a}).
\ene
The constants in the $O$-terms do not depend on $n$ (but they may depend on $M$). The value of $a$
does not depend on either $n$ or $M$.
\enl
\ber\label{rem:new1}
Note that \eqref{eq:main_lem1} is not a `proper' asymptotic formula, since the coefficients
$e_{j,p}(n)$ are allowed to grow with $n$ (and, therefore, with $\rho$).
\enr

Let us prove Theorem \ref{main_thm} assuming that we have proved Lemma \ref{main_lem}. Let $M$ be fixed.
Denote
\bee
N_n(\rho^2):=C_d\rho^d+\sum_{p=0}^{d-1}\sum_{j=-d+1}^{6M}e_{j,p}(n)\rho^{-j}(\ln\rho)^p.
\ene
Then, whenever $\rho\in J_n:=I_{n-1}\cap I_n=[\rho_n,2\rho_n]
$, we have:
\bee
N_n(\rho^2)-N_{n-1}(\rho^2)=\sum_{p=0}^{d-1}\sum_{j=-d+1}^{6M}t_{j,p}(n)\rho^{-j}(\ln\rho)^p,
\ene
where
\bee
t_{j,p}(n):=e_{j,p}(n)-e_{j,p}(n-1).
\ene
On the
other hand, since for $\rho\in J_n$ we have both
$N(\rho^2)=N_n(\rho^2)+O(\rho_n^{-M})$ and
$N(\rho^2)=N_{n-1}(\rho^2)+O(\rho_n^{-M})$, this implies that
$\sum_{p=0}^{d-1}\sum_{j=-d+1}^{6M}t_{j,p}(n)\rho^{-j}(\ln\rho)^p=O(\rho_n^{-M})$.
\becl For each $j=-d+1,\dots,6M$ we have:
$t_{j,p}(n)=O(\rho_n^{j-M}(\ln\rho_n)^{d-1-p})$.
\encl
\bep
Put
$x:=\rho^{-1}$. Then
$\sum_{p=0}^{d-1}\sum_{j=-d+1}^{6M}t_{j,p}(n)x^{j}(-1)^p(\ln
x)^p=O(\rho_n^{-M})$ whenever $x\in
[\frac{\rho_n^{-1}}{2},\rho_n^{-1}]$. Put $y:=x\rho_n$
and
$$\tau_{j,p}(n):=\rho_n^{M-j}\sum\limits_{s=p}^{d-1}\begin{pmatrix}s\\p\end{pmatrix}(-1)^p t_{j,s}(n)(\ln\rho_n)^{s-p}.$$
Then
\bee\label{Cramer}
P(y):=\sum_{p=0}^{d-1}\sum_{j=-d+1}^{6M}\tau_{j,p}(n)y^{j}(\ln y)^p=O(1)
\ene
whenever $y\in [\frac{1}{2},1]$. Consider the following $d(6M+d)$ functions: $y^j(\ln y)^p$ ($j=-d+1,...,6M$, $p=0,\dots,d-1$) and label them $h_1(y),...h_{d(6M+d)}(y)$. These functions are linearly independent on the interval
$[\frac{1}{2},1]$. Therefore, there exist points $y_1,...,y_{d(6M+d)}\in [\frac{1}{2},1]$ such that
the determinant of the matrix $(h_j(y_l))_{j,l=1}^{d(6M+d)}$ is non-zero. Now \eqref{Cramer} and the Cramer's Rule imply that for each $j$ the values $\tau_{j,p}(n)$ are fractions with a bounded expression in the numerator and a fixed non-zero number in the denominator. Therefore, $\tau_{j,p}(n)=O(1)$. This shows first that $t_{j,d-1}(n)=O(\rho_n^{j-M})$ and then subsequently reducing index $p$ from $p=d-1$ to $p=0$ we obtain $t_{j,p}(n)=O(\rho_n^{j-M}(\ln\rho_n)^{d-1-p})$ as claimed.
\enp

Thus, for $j<M$, the series $\sum_{m=0}^\infty t_{j,p}(m)$ is absolutely convergent; moreover, for such
$j$ we have:
\bee
\begin{split}
& e_{j,p}(n)=e_{j,p}(0)+\sum_{m=1}^n
t_{j,p}(m)=e_{j,p}(0)+\sum_{m=1}^\infty
t_{j,p}(m)+O(\rho_n^{j-M}(\ln\rho_n)^{d-1-p})\cr &
=:e_{j,p}+O(\rho_n^{j-M}(\ln\rho_n)^{d-1-p}),
\end{split}
\ene where we have denoted $e_{j,p}:=e_{j,p}(0)+\sum_{m=1}^\infty
t_{j,p}(m)$.

Since $e_{j,p}(n)=O(\rho_n^{(2j/3)+a})$ (it was one of the assumptions
of lemma), we have: \bee
\sum_{j=M}^{6M}|e_{j,p}(n)|\rho_n^{-j}=O(\rho_n^{a-\frac{M}{3}})=O(\rho_n^{-\frac{M}{4}}),
\ene assuming as we can without loss of generality that $M$ is
sufficiently large.
Thus, when $\rho\in I_n$, we have: \bee\label{eq:main_lem11}
N(\rho^2)=C_d\rho^d+\sum_{p=0}^{d-1}\sum_{j=-d+1}^{M-1}e_{j,p}\rho^{-j}(\ln\rho)^p+
O(\rho^{-M}(\ln\rho)^{d-1})+O(\rho^{-\frac{M}{4}}(\ln\rho)^{d-1}).
\ene
Since constants in $O$ terms do not depend on $n$, 
for all $\rho\ge \rho_0$ we have:
\bee\label{eq:main_lem17}
\bes
N(\rho^2)&=C_d\rho^d+\sum_{p=0}^{d-1}\sum_{j=-d+1}^{M-1}e_{j,p}\rho^{-j}(\ln\rho)^p+O(\rho^{-\frac{M}{6}})\\
&=C_d\rho^d+\sum_{p=0}^{d-1}\sum_{j=-d+1}^{[M/6]}e_{j,p}\rho^{-j}(\ln\rho)^p+O(\rho^{-\frac{M}{6}}).
\end{split}
\ene
Taking $M=6K+1$, we obtain \eqref{eq:main_thm1}.

The rest of the paper is devoted to proving Lemma \ref{main_lem}. We will mostly concentrate on obtaining
formula \eqref{eq:main_lem1}, since estimate \eqref{eq:main_lem2} will usually follow by trivial but tedious
arguments (like estimating coefficients in the product of several geometric series). However, in the cases when
estimating the coefficients in our infinite series would present difficulties, we will carry out these
estimates as well. The first step of the proof is fixing $n$ and fixing large $\tilde k$ and $k$. The precise value of
$\tilde k$
will be chosen later 
in order to satisfy estimate \eqref{eq:kM} (this estimate says that
the more asymptotic terms we want to have in \eqref{eq:main_lem1}, the bigger $\tilde k$ we need to choose; note that
the choice of $\tilde k$ does not depend on $k$).
We will have several
requirements on how large $k$ should be (most of them will be of the form $\rho_n^{0+}<\rho_n^{\epsilon}$
or $\rho_n^{0-}>\rho_n^{-\epsilon}$); each time we have such an inequality, we assume that $k$ is chosen sufficiently
large to satisfy it. The first requirement on $k$ we have is that $k>M$.
After fixing $n$ and $k$, we choose the finite set
$\Bth(k;\rho_n)$ and the approximating potential $b_{(k;\rho_n)}$ which satisfy all the conditions B--D.
Then condition \eqref{eq:condB2} and definition \eqref{eq:def1} imply that difference between the densities of states of operators with potentials $b$
and $b_{(k;\rho_n)}$ is smaller than $\rho_n^{-M}$. Thus, from now on we will consider the operator with a potential
$b_{(k;\rho_n)}$ and try to establish \eqref{eq:main_lem1} for this new operator. Following our policy of
getting rid of all indexes as soon as possible (i.e. immediately after we have fixed them), we will denote $\Bth:=\Bth(k;\rho_n)$ and $b:=b_{(k;\rho_n)}$. This means that from now on we will assume that $b$ is a quasi-periodic
potential with $\Bth$ its spectrum of frequencies so that $\Bth$ satisfies conditions A--D with $\rho=\rho_n$.

\section{Abstract results}

In this section,we establish several abstract results concerning density of states for
operators with almost-periodic coefficients. In the periodic setting, these results become either trivial
or already known, so the reader who is mostly interested in the periodic case, can skip this section.

In this and further sections, we will work with pseudo-differential operators with almost-periodic
coefficients (or symbols). These operators were studied in \cite{Shu0} and \cite{Shu}.
In Section 8, we will introduce the classes of such operators. We will also see that one can naturally
consider the action of such operators in both $L_2(\R^d)$ and $B_2(\R^d)$. These actions have many
similarities between them; in particular, the norms and (for elliptic or bounded operators) the spectra
of operators acting in $L_2(\R^d)$ and $B_2(\R^d)$ are the same, see \cite{Shu0}. As a result, often when
we discuss a pseudo-differential operator with almost-periodic
coefficients, we do not specify in which space it acts.
Sometimes, however, it becomes important to emphasize the space where the operator acts,
in which case we will do this.

Following \cite{Shu}, we denote by $\GA_B$ a $II_{\infty}$ factor
acting in $\tilde\GH:=B_2(\R^d)\otimes L_2(\R^d)$. We denote by
$\be_{\bxi}$ both the function $e^{i\bxi\bx}$ and the operator of
multiplication by this function. $T_{\bxi}$ is the operator of
translation by $\bxi$ in $L_2(\R^d)$, i.e. $T_{\bxi}
u(\bx)=u(\bx-\bxi)$. The factor $\GA_B$ is defined as the von
Neumann algebra acting in $\tilde\GH$ generated by two families of
operators: \bee \{\be_{\bxi}\otimes\be_{\bxi},\,\bxi\in\R^d\} \ene
and \bee \{I\otimes T_{\bxi},\,\bxi\in\R^d\}. \ene Let $A=a(\bx,D)$
be a self-adjoint pseudo-differential operator with almost-periodic
coefficients such that $a(\bx,\bxi)\gg|\bxi|^m$ for some
$m>0$. 
We introduce operator
$A^{\sharp}:=a(\bx+\by,D_{\by})$ acting in $\tilde\GH$; here, $\bx$
is a variable of functions in $B_2(\R^d)$ and $\by$ is a variable of functions in $L_2(\R^d)$.
We denote by $E_{\la}(A)$ the spectral projection of $A$; by $\tilde
E_{\la}(A)$ we denote the spectral projection of $A^{\sharp}$. By
$\BD$ and $\BT$ we denote the relative dimension and the relative
trace in $\GA_B$ (see \cite{Nai}).

If $A$ is actually a {\it differential} operator then (see
\cite{Shu}) one can define the density of states of $A$ (denoted by
$N(\la;A)$) by formula \eqref{eq:def1}.
It was
also proved in \cite{Shu} that
\bee\label{eq:density1}
N(\la;A)=\BT(\tilde E_{\la}(A^{\sharp}))=\BD(\tilde
E_{\la}(A^{\sharp})\tilde \GH).
\ene
Note that the relative dimension in a $II_{\infty}$ factor
can take any non-negative value.
Now it is natural 
to define the density of states for
a general elliptic self-adjoint {\it pseudo-differential} operator with
almost-periodic coefficients
by \eqref{eq:density1}.

By $\GL$ we denote a closed linear subspace of $\tilde\GH$ adjoint
to $\GA_B$ (see \cite{Nai} for the explanation of the terminology).
The following lemma gives a variational description of the density
of states.
\bel\label{lem:densty1} 
\bee\label{eq:density2}
N(\la;A)=\sup\{\BD(\GL), \,\, (A^{\sharp}\phi,\phi)\le \la(\phi,\phi),
\forall\phi\in\GL\}.
\ene
\enl
\bep
By taking $\GL:=\tilde
E_{\la}(A^{\sharp})\tilde \GH$ and using \eqref{eq:density1}, we see that the
LHS of \eqref{eq:density2} is at most the RHS. Suppose now that we
have found a subspace $\GL$ such that $\BD(\GL)>\BD(\tilde
E_{\la}(A^{\sharp})\tilde \GH)$. Then Lemma from Section VII.37 of
\cite{Nai} implies that $\GL$ contains a non-zero vector $\phi$
orthogonal to $\tilde E_{\la}(A^{\sharp})\tilde \GH$. But then
$(A^{\sharp}\phi,\phi)> \lambda (\phi,\phi)$, which contradicts our
assumption on $\GL$. This proves \eqref{eq:density2}.
\enp
\bec\label{cor:norms}
If $A\ge B$, then $N(\la;A)\le N(\la;B)$.
\enc
\bec\label{cor:H1H2}
Suppose, $H_1$ and $H_2$ are two  elliptic self-adjoint pseudo-differential operators with
almost-periodic coefficients such that $||H_1-H_2||\ll\rho_n^{-M+(2-d)}$. Suppose,
$N(H_2;\rho^2)$ satisfies asymptotic expansion \eqref{eq:main_lem1}. Then
$N(H_1;\rho^2)$ also satisfies \eqref{eq:main_lem1}.
\enc
\bep
Our assumptions imply that $H_2-\de\le H_1\le H_2+\de$, where $\de\ll \rho_n^{-M+(2-d)}$.
Previous corollary now implies that
\bee\label{eq:H1H2}
N(H_2+\de;\la)\le N(H_1;\la)\le N(H_2-\de;\la).
\ene
It remains to notice that if $N(H_2;\rho^2)$ satisfies \eqref{eq:main_lem1}, then the difference
between RHS and LHS of \eqref{eq:H1H2} is $O(\rho_n^{-M})$.
\enp
\bel\label{lem:density2} Suppose, $A=a(\bx,D)$ and $U=u(\bx,D)$ are
two pseudo-differential operators with almost-periodic coefficients.
Let operator $A$ be elliptic self-adjoint
and operator $U$ be
unitary. Then $N(\la;A)=N(\la;U^{-1}AU)$.
\enl
\bep
Obviously,
operator $U^{\sharp}$ is unitary and
$(U^{-1}AU)^{\sharp}=(U^{\sharp})^{-1}A^{\sharp}U^{\sharp}$. Thus,
\bee
N(\la;U^{-1}AU)=\BT(\tilde
E_{\la}((U^{\sharp})^{-1}A^{\sharp}U^{\sharp}))=
\BT((U^{\sharp})^{-1}\tilde
E_{\la}(A^{\sharp})U^{\sharp})=\BT(\tilde
E_{\la}(A^{\sharp}))=N(\la;A).
\ene
Here, the third equality follows, for example, from Sections 36-37, Chapter 7 of \cite{Nai}.
\enp

\section{Resonance zones}

In this section, we define resonance regions and establish some of their properties. Recall the definition
of the set $\Bth=\Bth(k;\rho_n)$ as well as of the quasi-lattice subspaces from Section 2. As before, by
$\Bth_{\tilde k}$
we denote the algebraic sum of $\tilde k$ copies of $\Bth$; remember that we consider the index $\tilde k$ fixed. We also put
$\Bth'_{\tilde k}:=\Bth_{\tilde k}\setminus\{0\}$. For each $\GV\in\CV$ we put $S_{\GV}:=\{\bxi\in\GV,\ |\bxi|=1\}$.
For each non-zero $\bth\in\R^d$ we put $\bn(\bth):=\bth|\bth|^{-1}$.

Let $\GV\in\CV_m$. We say that $\GF$ is a {\it flag} generated by $\GV$, if $\GF$ is a sequence $\GV_j\in\CV_j$ ($j=0,1,\dots,m$)
such that $\GV_{j-1}\subset\GV_j$ and $\GV_m=\GV$. We say that $\{\bnu_j\}_{j=1}^m$ is a sequence generated by $\GF$
if $\bnu_j\in\GV_j\ominus\GV_{j-1}$ and $||\bnu_j||=1$ (obviously, this condition determines each $\bnu_j$ up to the
multiplication by $-1$). We denote by $\CF(\GV)$ the collection of all flags generated by $\GV$. We also
fix an increasing sequence of positive numbers $\al_j$ ($j=1,\dots,d$) with $\al_d<\frac{1}{2 d}$
(these numbers depend only on $d$) and put $L_j:=\rho_n^{\al_j}$.

Let $\bth\in\Bth'_{\tilde k}$. We call by {\it resonance zone} generated by $\bth$
\bee\label{eq:1}
\Lambda(\bth):=\{\bxi\in\R^d,\ |\lu\bxi,\bn(\bth)\ru|\le L_1\}.
\ene
Suppose, $\GF\in\CF(\GV)$ is a flag and $\{\bnu_j\}_{j=1}^m$ is a sequence generated by $\GF$. We define
\bee\label{eq:2}
\Lambda(\GF):=\{\bxi\in\R^d,\ |\lu\bxi,\bnu_j\ru|\le L_j\}.
\ene
If $\dim\GV=1$, definition \eqref{eq:2} is reduced to \eqref{eq:1}.
Obviously, if $\GF_1\subset\GF_2$, then $\Lambda(\GF_2)\subset\Lambda(\GF_1)$.

Suppose, $\GV\in\CV_j$. We denote
\bee\label{eq:Bxi1}
\Bxi_1(\GV):=\cup_{\GF\in\CF(\GV)}\Lambda(\GF).
\ene
Note that
$\Bxi_1(\GX)=\R^d$ and $\Bxi_1(\GV)=\Lambda(\bth)$ if $\GV\in\CV_1$ is spanned by $\bth$.
Finally, we put
\bee\label{eq:Bxi}
\Bxi(\GV):=\Bxi_1(\GV)\setminus(\cup_{\GU\supsetneq\GV}\Bxi_1(\GU))=\Bxi_1(\GV)
\setminus(\cup_{\GU\supsetneq\GV}\cup_{\GF\in\CF(\GU)}\Lambda(\GF)).
\ene
We call $\Bxi(\GV)$ the resonance region generated by $\GV$. 
Very often, the region $\Bxi(\GX)$ is called the non-resonance region. We, however, will omit using
this terminology since we will treat all regions $\Bxi(\GV)$ in the same way.

Let us establish some basic properties of resonance regions. The first set of properties follows immediately
from the definitions.

\bel\label{lem:propBUps}

(i) We have
\bee
\cup_{\GV\in\CV}\Bxi(\GV)=\R^d.
\ene

(ii) $\bxi\in\Bxi_1(\GV)$ iff $\bxi_{\GV}\in\Omega(\GV)$, where $\Omega(\GV)\subset\GV$
is a certain bounded set (more precisely, $\Omega(\GV)=\Bxi_1(\GV)\cap\GV\subset B(m L_m)$ if
$\dim\GV=m$).

(iii) $\Bxi_1(\R^d)=\Bxi(\R^d)$ is a bounded set, $\Bxi(\R^d)\subset B(d L_d)$; all other sets
$\Bxi_1(\GV)$ are unbounded.

\enl

Now we move to slightly less obvious properties. From now on we always assume that $\rho_0$
(and thus $\rho_n$) is sufficiently large. We also assume, as we always do, that the value of $k$
is sufficiently large so that, for example, $L_j\rho_n^{0+}<L_{j+1}$.

\bel\label{lem:intersect}
Let $\GV,\GU\in\CV$. Then $(\Bxi_1(\GV)\cap\Bxi_1(\GU))\subset \Bxi_1(\GW)$,
where $\GW:=\GV+\GU$ (algebraic sum).
\enl
\bep
Assume, without loss of generality, that $m_1:=\dim\GV\ge\dim\GU=:m_2$. If $\GU\subset\GV$, then
the statement of the lemma is obvious. Consider the case when $\GV$ and $\GU$ are strongly distinct.
Suppose, $\bxi\in(\Bxi_1(\GV)\cap\Bxi_1(\GU))$. Then there is a flag $\GF\in\CF(\GV)$
such that $\bxi\in\Lambda(\GF)$. Let $\GF_1\in\CF(\GW)$ be any flag such that the first
$m_1$ elements of $\GF_1$ coincide with $\GF$. Let us prove that $\bxi\in\Lambda(\GF_1)$.
Let $\{\bnu_j\}_{j=1}^m$ be a sequence generated by $\GF_1$ ($m=\dim\GW$). Then the
inclusion $\bxi\in\Lambda(\GF)$ implies that $|\lu\bxi,\bnu_j\ru|\le L_j$ for $j=1,\dots, m_1$.
Moreover, a simple geometry implies $|\bxi_{\GW}|\le (|\bxi_{\GV}|+|\bxi_{\GU}|)[\sin(\phi(\GV,\GU))]^{-1}\le
2m_1 L_{m_1} s(\rho_n)^{-1}< L_{m_1}\rho_n^{0+}<L_{m_1+1}$. Therefore, for $j\ge m_1+1$ we have
$|\lu\bxi,\bnu_j\ru|=|\lu\bxi_{\GW},\bnu_j\ru|\le|\bxi_{\GW}|\le L_{m_1+1}\le L_j$. This shows
that indeed $\bxi\in\Lambda(\GF_1)$ and, therefore, $\bxi\in\Bxi_1(\GW)$, which proves our lemma.
\enp
The next statement follows immediately from Lemma \ref{lem:intersect}.
\bec
(i) We can re-write definition \eqref{eq:Bxi} like this:
\bee\label{eq:Bxibis}
\Bxi(\GV):=\Bxi_1(\GV)\setminus(\cup_{\GU\not\subset\GV}\Bxi_1(\GU)).
\ene

(ii) If $\GV\ne\GU$, then $\Bxi(\GV)\cap\Bxi(\GU)=\emptyset$.

(iii) We have $\R^d=\sqcup_{\GV\in\CV}\Bxi(\GV)$ (the disjoint union).
\enc

\bel\label{lem:verynew}
Let $\GV\in\CV_m$ and $\GV\subset\GW\in\CV_{m+1}$. Let $\bmu$ be (any) unit vector from $\GW\ominus\GV$.
Then, for $\bxi\in\Bxi_1(\GV)$, we have  $\bxi\in\Bxi_1(\GW)$ if and only if the estimate
$|\lu\bxi,\bmu\ru|=|\lu\bxi_{\GV^{\perp}},\bmu\ru|\le L_{m+1}$ holds.
\enl
\bep
In one direction the statement is obvious. Now, we assume that $\bxi\in\Bxi_1(\GV)\cap\Bxi_1(\GW)$. Let $\GF=\{\GW_0,\dots\GW_{m},\GW\}$ be a flag for which $\bxi\in\Lambda(\GF)$. If $\GW_m=\GV$ then the statement of the lemma is straightforward. Otherwise, we can apply the construction from the proof of Lemma~\ref{lem:intersect} with $\GU=\GW_m$. This completes the proof.
\enp

\bel\label{lem:newXi}
We have
\bee\label{eq:newBxi}
\Bxi_1(\GV)\cap\cup_{\GU\supsetneq\GV}\Bxi_1(\GU)=\Bxi_1(\GV)\cap\cup_{\GW\supsetneq\GV, \dim\GW=1+\dim\GV}\Bxi_1(\GW).
\ene
\enl
\bep
Indeed, obviously, the RHS of \eqref{eq:newBxi} is a subset of the LHS. On the other hand, suppose,
$\GU\supsetneq\GV$ and $\bxi\in\Bxi_1(\GV)\cap\Bxi_1(\GU)$. Then $\bxi\in\Lambda(\GF)$ for some $\GF\in\CF(\GU)$.
Suppose that $\GV_1\in\GF$ is a subspace such that $\dim\GV_1=\dim\GV$.
If $\GV_1=\GV$, it immediately follows that $\bxi$ is contained in the
RHS of \eqref{eq:newBxi}. Assume that $\GV_1\ne\GV$; in particular, $\dim\GV\ge 1$. Then there exists $\GV_2\in\GF$ such that
$\dim (\GV+\GV_2)=\dim\GV +1$. Put $\GW:=\GV+\GV_2$. Since $\bxi\in\Lambda(\GF)\subset\Bxi_1(\GV_2)$, by Lemma \ref{lem:intersect} we have $\bxi\in\Bxi_1(\GW)$, and so $\bxi$ is contained in the RHS of \eqref{eq:newBxi}.
\enp

\bec
We can re-write \eqref{eq:Bxi} as
\bee\label{eq:Bxibbis}
\Bxi(\GV):=\Bxi_1(\GV)\setminus(\cup_{\GW\supsetneq\GV, \dim\GW=1+\dim\GV}\Bxi_1(\GW)).
\ene
\enc

\bel\label{lem:Upsilon}
Let $\GV\in\CV$ and $\bth\in\Bth_{\tilde k}$. Suppose that $\bxi\in\Bxi(\GV)$ and
both points $\bxi$ and $\bxi+\bth$ are inside $\Lambda(\bth)$. Then $\bth\in\GV$ and $\bxi+\bth\in\Bxi(\GV)$.
\enl
\bep
If $\bth\not\in\GV$, then Lemma \ref{lem:intersect} implies that $\bxi\in\Bxi_1(\GW)$, where
$\GW=\mathrm {span}(\GV,\bth)$, which contradicts our assumption  $\bxi\in\Bxi(\GV)$.

Let us prove that 
$\bxi+\bth\in\Bxi_1(\GV)$. Since $\bxi\in\Bxi(\GV)\subset\Bxi_1(\GV)$, this implies
that $\bxi\in\Lambda(\GF)$ with $\GF\in\CF(\GV)$, $\GF=\{\GV_0=\GX,\GV_1,\dots,\GV_m=\GV\}$.

Let $J$ be the biggest number such that $\bth\not\in\GV_{J-1}$ (obviously, $J\le m:=\dim\GV$). We construct a new flag
$\GF_1=\{\GU_0=\GX,\GU_1,\dots,\GU_m=\GV\}$ such that
\bees
\GU_j=\begin{cases}
\GX,& j=0\\
\mathrm {span}(\GV_{j-1},\bth),& 0<j\le J\\
\GV_j,&j>J
\end{cases}
\enes

We are going to prove that
$\bxi+\bth\in\Lambda(\GF_1)$. Let $\{\bnu_j\}_{j=1}^m$ be a sequence generated by $\GF_1$.
Obviously, for $j>J$ we have $\lu\bxi+\bth,\bnu_j\ru=\lu\bxi,\bnu_j\ru$,
so that $|\lu\bxi+\bth,\bnu_j\ru|\le L_j$ if $|\lu\bxi,\bnu_j\ru|\le L_j$. So, assume that
$j\le J$. If $j=1$, we have $\bnu_1=\bn(\bth_1)$, so the assumption $\bxi+\bth\in\Lambda(\bth)$ implies
$|\lu\bxi+\bth,\bnu_1\ru|\le L_1$. Assume now that $1<j\le J$. Then
$|\bxi_{\GU_j}|\le(|\bxi_{\GV_{j-1}}|+|\bxi_{\GU_1}|) s(\rho_n)^{-1}\le 2(j-1)L_{j-1} s(\rho_n)^{-1}$. Therefore,
$|(\bxi+\bth)_{\GU_j}|\le 2(j-1)L_{j-1} s(\rho_n)^{-1}+|\bth|\le 2(j-1)L_{j-1} s(\rho_n)^{-1}+R(\rho_n)\le L_j$.
Thus, $|\lu\bxi+\bth,\bnu_j\ru|\le |(\bxi+\bth)_{\GU_j}|\le L_j$. This shows that,
indeed, we have $\bxi+\bth\in\Lambda(\GF_1)$ and, therefore, $\bxi+\bth\in\Bxi_1(\GV)$.

Suppose now that $\bxi+\bth\not\in\Bxi(\GV)$. This could only happen if $\bxi+\bth\in\Bxi_1(\GW)$ for some
$\GW\supsetneq\GV$. But then the previous part of the proof would imply that $\bxi\in\Bxi_1(\GW)$, which
contradicts our assumption $\bxi\in\Bxi(\GV)$. Thus, $\bxi+\bth\in\Bxi(\GV)$, which finishes the proof.
\enp

The next definition is almost identical to the corresponding definition from \cite{ParSob}.
\begin{defn}
\label{reachability:defn}
Let $\bth, \bth_1, \bth_2, \dots, \bth_l$ be some vectors from $\Bth'_{\tilde k}$,
which are not necessarily distinct.
\begin{enumerate}
\item \label{1}
We say that two vectors
$\bxi, \boldeta\in\R^d$ are \textsl{$\bth$-resonant congruent}
if both $\bxi$ and $\boldeta$ are inside $\L(\bth)$ and $(\bxi - \boldeta) =l\bth$ with $l\in\Z$.
In this case we write $\bxi \leftrightarrow \boldeta \mod \bth$.
\item\label{2}
For each $\bxi\in\R^d$ we denote by
$\BUps_{\bth}(\bxi)$ the set of all points which are
$\bth$-resonant congruent to $\bxi$.
For $\bth\not = \mathbf 0$ we say that
$\BUps_{\bth}(\bxi) = \varnothing$
if $\bxi\notin\L(\bth)$.
\item\label{3}
We say that $\bxi$ and $\boldeta$
are \textsl{$\bth_1, \bth_2, \dots, \bth_l$-resonant congruent},
if there exists a sequence $\bxi_j\in\R^d, j=0, 1, \dots, l$ such that
$\bxi_0 = \bxi$, $\bxi_l = \boldeta$,
and $\bxi_j\in\BUps_{\bth_j}(\bxi_{j-1})$ for $j = 1, 2, \dots, l$.

\item
We say that $\boldeta\in\R^d$ and
$\bxi\in\R^d$ are \textsl{resonant congruent}, if either $\bxi=\boldeta$ or $\bxi$ and $\boldeta$ are $\bth_1, \bth_2, \dots, \bth_l$-resonant congruent with some
 $\bth_1, \bth_2, \dots, \bth_l\in\Bth_{\tilde k}'$. The set of \textbf{all} points, resonant congruent
to $\bxi$, is denoted by $\BUps(\bxi)$.
For points $\boldeta\in\BUps(\bxi)$ (note that this condition is equivalent to
$\bxi\in\BUps(\boldeta)$) we write $\boldeta\leftrightarrow\bxi$.
\end{enumerate}
\end{defn}

Note that $\BUps(\bxi)=\{\bxi\}$ for any $\bxi\in\Bxi(\GX)$. Now Lemma \ref{lem:Upsilon} immediately implies
\bec\label{cor:Upsilon}
For each $\bxi\in\Bxi(\GV)$ we have $\BUps(\bxi)\subset\Bxi(\GV)$ and thus
\bees
\Bxi(\GV)=\sqcup_{\bxi\in\Bxi(\GV)}\BUps(\bxi).
\enes
\enc
\bel\label{lem:diameter}
The diameter of $\BUps(\bxi)$ is bounded above by $mL_m$, if $\bxi\in\Bxi(\GV)$, $\GV\in\CV_m$.
\enl
\bep
This follows from Lemmas \ref{lem:Upsilon} and \ref{lem:propBUps}.
\enp
\bel\label{lem:finiteBUps}
For each $\bxi\in\Bxi(\GV),\ \GV\ne\R^d$, the set $\BUps(\bxi)$ is finite, and $\card{\BUps(\bxi)}\ll\rho_n^{(d-1)\alpha_{d-1}+0+}$ uniformly in $\bxi\in\R^d\setminus\Bxi(\R^d)$.
\enl
\bep
This immediately follows from Lemmas \ref{lem:propBUps}, \ref{lem:Upsilon}, \ref{lem:diameter}, Conditions A and D
and a standard covering argument.
\enp

\section{Description of the approach}



For any set $\CC\subset\R^d$ by $\CP(\CC)$ we denote  the orthogonal projection onto
$\mathrm{span}\{\be_{\bxi}\}_{\bxi\in\CC}$ in $B_2(\R^d)$ and by $\CP^{L}(\CC)$ the same
projection considered in $L_2(\R^d)$, i.e.
\bee
\CP^{L}(\CC)=\CF^*\chi_{\CC}\CF,
\ene
where $\CF$ is the Fourier transform and $\chi_{\CC}$ is the operator of multiplication
by the characteristic function of $\CC$.
Obviously, $\CP^{L}(\CC)$ is a well-defined (resp. non-zero) projection iff
$\CC$ is measurable (resp. has non-zero measure).
We also denote $\GH:=B_2(\R^d)$. 
Let us fix sufficiently large $n$, and denote (recall that $\la_n=\rho_n^2$)
\bee
\CX_n
:=\{\bxi\in\R^d,\,|\bxi|^2\in [0.7\la_n,
17.5\la_n
]\}.
\ene
We also put
\bee
\CA=\CA_n:=\cup_{\bxi\in\CX_n}\BUps(\bxi).
\ene
Lemma \ref{lem:diameter} implies that for each $\bxi\in\CA$ we have $|\bxi|^2\in[0.5\la_n,18\la_n]$. In particular, we have
\bee\label{eq:CARd}
\CA\cap\BXi(\R^d)=\emptyset.
\ene
For each $\GV\in\CV_m$, $m<d$, we put
\bee
\CA(\GV):=\CA_n\cap\Bxi(\GV).
\ene
We also denote
\bee
\hat\CA:=\{\bxi\not\in\CA,\ |\bxi|^2<\la_n\}
\ene
and
\bee
\check\CA:=\{\bxi\not\in\CA,\ |\bxi|^2>\la_n\}.
\ene


We plan to apply the gauge transform similar to the one used in \cite{ParSob} to the
operator $H$. The details of this procedure will be explained in Sections 8 and 9; here, we just
mention that we are going to introduce two operators: $H_1$ and $H_2$. The operator $H_1$ is unitary equivalent to
$H$: $H_1=U^{-1}HU$, where $U=e^{i\Psi}$ with a bounded pseudo-differential operator $\Psi$
with almost-periodic coefficients
(then Lemma \ref{lem:density2} implies that the densities of states of
$H$ and $H_1$ are the same). Moreover, $H_1=H_2+\CR$, where
$||\CR||\ll\rho_n^{-M+(2-d)}$ and $H_2=-\Delta+W$ is a self-adjoint pseudo-differential
operator with symbol $|\bxi|^2+w$ which satisfies the following property (see section 8 for more discussion about pseudo-differential
operators and their symbols):
\bee\label{eq:b3}
\hat w(\bth,\bxi)=0, \ \mathrm{if} \
(\bxi\not\in\La(\bth)\ \&\ \bxi\in\CA), \  \mathrm{or} \ (\bxi+\bth\not\in\La(\bth)\ \&\ \bxi\in\CA),  \  \mathrm{or} \
(\bth\not\in\Bth_{\tilde k}).
\ene
Now Corollary \ref{cor:H1H2} implies that if we prove that $N(\rho^2;H_2)$ satisfies \eqref{eq:main_lem1}, then $N(\rho^2;H_1)$
(and therefore $N(\rho^2;H)$)
satisfies the same asymptotic formula. This means that it is enough to establish the asymptotic expansion \eqref{eq:main_lem1} for
the operator $H_2$ instead of $H$.
Condition \eqref{eq:b3} implies that
for each $\bxi\in\CA$ the subspace $\CP(\BUps(\bxi))\GH$ is an invariant
subspace of $H_2$ (acting, remember, in $B_2(\R^d)$); we denote its
dimension by $m$ (which is finite by Lemma \ref{lem:finiteBUps}). We put
\bee\label{eq:h3bxi}
H_2(\BUps(\bxi)):=H_2\bigm|_{\BUps(\bxi)\GH}.
\ene
Note that the subspaces
$\CP(\hat\CA)\GH$ and $\CP(\check\CA)\GH$ are invariant as well;
by $H_2(\hat\CA)$ and $H_2(\check\CA)$ we denote the restrictions of
$H_2$ to these subspaces; we also denote by $H_2(\CA)$ the restriction of
$H_2$ to $\CP(\CA)\GH$. Also notice that if we consider the operator $H_2$
acting in $L_2(\R^d)$, then $\CP^L(\hat\CA)L_2(\R^d)$, $\CP^L(\check\CA)L_2(\R^d)$, and $\CP^L(\CA)L_2(\R^d)$
would still be invariant subspaces.
For each $\bxi\in\CA$ the operator
$H_2(\BUps(\bxi))$ is a finite-dimensional self-adjoint operator, so
its spectrum is purely discrete; we denote its eigenvalues (counting
multiplicities) by
$\la_1(\BUps(\bxi))\le\la_2(\BUps(\bxi))\le\dots\le\la_m(\BUps(\bxi))$
and the corresponding orthonormalized eigenfunctions by
$\{h_{j,\BUps(\bxi)}(\bx)\}$. Next, we list all points
$\boldeta\in\BUps(\bxi)$ in increasing order of their absolute
values; thus, we have put into correspondence to each point
$\boldeta\in\BUps(\bxi)$ a natural number $t=t(\boldeta)$ so that
$t(\boldeta)< t(\boldeta')$ if $|\boldeta|< |\boldeta'|$. If two
points $\boldeta=(\eta_1,\dots,\eta_d)$ and
$\boldeta'=(\eta'_1,\dots,\eta'_d)$ have the same absolute values,
we put them in the lexicographic order of their coordinates, i.e. we
say that $t(\boldeta)< t(\boldeta')$ if $\eta_1<\eta'_1$, or
$\eta_1=\eta'_1$ and $\eta_2<\eta'_2$, etc. Now we define the
mapping $g:\CA\to\R$ which puts into correspondence to each point
$\boldeta\in\CA$ the number $\la_{t(\boldeta)}(\BUps(\boldeta))$. This
mapping is an injection from $\CA$ onto the set of eigenvalues of
$H_2$, counting multiplicities (recall that we consider the operator
$H_2$ acting in $B_2(\R^d)$, so there is nothing miraculous about its
spectrum consisting of eigenvalues and their limit points). Moreover,
all eigenvalues of $H_2$ inside the interval $[0.75\la_n,17\la_n]$ have a pre-image
under $g$.
Arguments,
similar to the ones used in \cite{ParSob}, show that $g$ is a
measurable function. Similarly, we define the mapping $h:\CA\to
B_2(\R^d)$ by the formula $h_{\bxi}:=h_{t(\bxi),\BUps(\bxi)}$. Then for each $\bxi\in\CA$ the expression
$(2\pi)^{-d}\sum_{\boldeta\in\BUps(\bxi)}h_{\boldeta}(\bx)\overline{h_{\boldeta}(\by)}$
is the integral kernel of the projection $\CP(\BUps(\bxi))$. Therefore, we have
\bee\label{eq:kernels}
\sum_{\boldeta\in\BUps(\bxi)}h_{\boldeta}(\bx)\overline{h_{\boldeta}(\by)}=
\sum_{\boldeta\in\BUps(\bxi)}\be_{\boldeta}(\bx)\overline{\be_{\boldeta}(\by)}.
\ene
Another, perhaps slightly simpler way of establishing \eqref{eq:kernels} is just to notice that
$\vec\bh=F\vec\be$, where $\vec\bh$ is a column-vector with entries $\{h_{\boldeta}\}_{\boldeta\in\BUps(\bxi)}$,
$\vec\be$ is a column-vector with entries $\{\be_{\boldeta}\}_{\boldeta\in\BUps(\bxi)}$, and $F$ is a unitary matrix.
Then
\bee\label{eq:kernels1}
\sum_{\boldeta\in\BUps(\bxi)}h_{\boldeta}(\bx)\overline{h_{\boldeta}(\by)}=\vec\bh(\bx)^T \overline{\vec\bh(\by)}=
\vec\be(\bx)^TF^T\overline{F} \overline{\vec\be(\by)}=\vec\be(\bx)^T \overline{\vec\be(\by)}=
\sum_{\boldeta\in\BUps(\bxi)}\be_{\boldeta}(\bx)\overline{\be_{\boldeta}(\by)}.
\ene
When $\bxi\not\in\CA$, we put
$g(\bxi):=|\bxi|^2$ and $h_{\bxi}:=\be_{\bxi}$, so that
now the functions $g$ and $h$ are defined on all $\R^d$.
It
follows from the construction that $\{h_{\bxi}\}_{\bxi\in\R^d}$ is
an orthonormal basis in $B_2(\R^d)$.

All this implies that for each $\lambda\in
[0.75\la_n,17\la_n]$ the function
\bee\label{eq:projection}
e(\la;\bx,\by):=(2\pi)^{-d}\int_{G_{\la}}h_{\bxi}(\bx)
\overline{h_{\bxi}(\by)}d\bxi,\,\,\bx,\by\in\R^d,
\ene
is the
integral kernel of the spectral projection $E_{\la}(H_2;B_2(\R^d))$
of the operator $H_2$ in $B_2(\R^d)$; here, we have denoted
\bee\label{eq:mla} G_{\la}:=\{\bxi\in\R^d,\,g(\bxi)\le\la\}. \ene

Notice that $e(\la;\bx,\by)$ also gives the kernel of the
spectral projection of the operator $H_2$ considered in $L_2(\R^d)$.
Since this is the statement we will use in our proof, let us give a little bit
more detailed proof of it. We define a mapping
$$
U:\ f\mapsto
(2\pi)^{-d/2}\int_{\R^d}\overline{h_{\bxi}(\bx)}f(\bx)d\bx
$$
and
\bee\label{eq:CM}
\CM:=
\CA/\leftrightarrow,
\ene
where $\leftrightarrow$ is the equivalence relation introduced in Definition \ref{reachability:defn};
Lemma \ref{lem:finiteBUps} and property \eqref{eq:CARd} imply that $\CM$ is measurable.
It is not hard to see that $U$ is a unitary operator in $L_2(\R^d)$ and
$$
U^*:\ z\mapsto
(2\pi)^{-d/2}\int_{\R^d}{h_{\bxi}(\bx)}z(\bxi)d\bxi.
$$
Indeed, we have (recall the notation \eqref{eq:CM} and identity \eqref{eq:kernels}):
\bee
\bes
&U^*Uf(\bx)=(2\pi)^{-d}\int_{\R^d}\int_{\R^d}h_{\bxi}(\bx)\overline{h_{\bxi}(\by)}f(\by)d\by d\bxi\\
=
&(2\pi)^{-d}\bigl(\int_{\hat\CA}+\int_{\check\CA}+\int_{\CA}\bigr)\int_{\R^d}h_{\bxi}(\bx)\overline{h_{\bxi}(\by)}f(\by)d\by d\bxi\\
=
&(2\pi)^{-d}\bigl(\int_{\hat\CA}+\int_{\check\CA}\bigr)\int_{\R^d}h_{\bxi}(\bx)\overline{h_{\bxi}(\by)}f(\by)d\by d\bxi\\
+
&(2\pi)^{-d}\int_{\CM}\int_{\R^d}\sum_{\boldeta\in\BUps(\bxi)}h_{\boldeta}(\bx)\overline{h_{\boldeta}(\by)}f(\by)d\by d\bxi\\
=
&(2\pi)^{-d}\bigl(\int_{\hat\CA}+\int_{\check\CA}\bigr)\int_{\R^d}\be_{\bxi}(\bx)\overline{\be_{\bxi}(\by)}f(\by)d\by d\bxi\\
+
&(2\pi)^{-d}\int_{\CM}\int_{\R^d}\sum_{\boldeta\in\BUps(\bxi)}\be_{\boldeta}(\bx)\overline{\be_{\boldeta}(\by)}f(\by)d\by d\bxi\\
=
&(2\pi)^{-d}\int_{\R^d}\int_{\R^d}\be_{\bxi}(\bx)\overline{\be_{\bxi}(\by)}f(\by)d\by d\bxi=f(\bx).
\end{split}
\ene

Now we notice that for $\la\in[0.75\la_n,17\la_n]$ the function $e(\la;\bx,\by)$ is the kernel of
$U^*\chi_{G_{\la}}U$. Moreover, $UH_2(\CA)U^*$ is the operator of
multiplication by $g$ acting in $L_2(\CA)$. It remains to notice that the other two restrictions of
$H_2$ satisfy $UH_2(\hat\CA)U^*<0.75\la_n I$
and  $UH_2(\check\CA)U^*>17\la_n I$.
Now it immediately follows that for $\la\in[\la_n,16\la_n]$ $e(\la;\bx,\by)$ is the kernel of the
spectral projection of the operator $H_2$ considered in $L_2(\R^d)$.

Now \eqref{eq:density1} and Theorem 4.1 from \cite{Shu} (see \eqref{eq:densityalmost}) imply the
following result (note that since $g$ is a
measurable function, $G_{\la}$ is a measurable set):
\bel For $\la\in[\la_n,16\la_n]$ being a continuity point of $N(\la;H_2)$
we have:
\bee\label{eq:densityh3}
N(\la;H_2)=(2\pi)^{-d}\vol(G_{\la}).
\ene
\enl
\bep
For the proof it is enough to notice that $|h_{\bxi}(\bx)|=|\be_{\bxi}(\bx)|=1$ for $\bxi\not\in\CA$ and
$$
|h_{\bxi}(\bx)|^2\leq\card{\BUps(\bxi)}\ll\rho_n^{(d-1)\alpha_{d-1}+0+}
$$
for $\bxi\in\CA$ by \eqref{eq:kernels} and Lemma~\ref{lem:finiteBUps} and apply Lebesgue's Limit Theorem.
\enp
Since points of continuity of $N(\la)$ are dense, {\it the asymptotic expansion proven for such $\lambda$ can be extended to all $\la\in[\la_n,16\la_n]$ by taking the limit}. Thus, our next task is to compute $\vol(G_{\la})$.
Let us put
\bee\label{eq:47}
\hat A^+:=\{\bxi\in\R^d,\,g(\bxi)<\rho^2<|\bxi|^2\}
\ene
and
\bee\label{eq:48}
\hat A^-:=\{\bxi\in\R^d,\,|\bxi|^2<\rho^2<g(\bxi)\}.
\ene
\bel\label{lem:new1}
\bee\label{eq:46}
\vol(G_{\la})=w_d\rho^d+\vol\hat A^+-\vol\hat A^-.
\ene
\enl
\bep
We obviously have $G_{\la}=B(\rho)\cup \hat A^+\setminus \hat A^-$. Since $\hat A^-\subset B(\rho)$
and $\hat A^+\cap B(\rho)=\emptyset$, this implies \eqref{eq:46}.
\enp
\ber\label{rem:nnew1}
Properties of the mapping $g$ imply that we have $\hat A^+,\ \hat A^-\subset\CA$. Thus, in order to compute $N(\lambda)$, we need to analyze the behaviour of $g$ only inside $\CA$.
\enr
We will compute volumes of $\hat A^{\pm}$ by means of integrating their characteristic functions in a specially chosen set of
coordinates. The next section is devoted to introducing these coordinates.

\section{Coordinates}
In this section, we do some preparatory work before computing $\vol\hat A^{\pm}$. Namely,
we are going to introduce a convenient set of coordinates in $\Bxi(\GV)$.
Let $\GV\in\CV_m$ be fixed; since $\hat A^{\pm}\cap\BXi(\R^d)=\emptyset$, we will assume that $m<d$. Then, as we have seen, $\bxi\in\Bxi_1(\GV)$ if and only if
$\bxi_{\GV}\in\Om(\GV)$. Let $\{\GU_j\}$ be a collection of all subspaces $\GU_j\in\CV_{m+1}$
such that each $\GU_j$ contains $\GV$. Let $\bmu_j=\bmu_j(\GV)$ be (any) unit vector from $\GU_j\ominus\GV$.
Then it follows from Lemma~\ref{lem:verynew} that for $\bxi\in\Bxi_1(\GV)$, we have  $\bxi\in\Bxi_1(\GU_j)$ if and only if the estimate
$|\lu\bxi,\bmu_j\ru|=|\lu\bxi_{\GV^{\perp}},\bmu_j\ru|\le L_{m+1}$ holds. Thus, formula \eqref{eq:Bxibbis} implies that
\bee\label{eq:Bxibbbis}
\Bxi(\GV)=\{\bxi\in\R^d,\ \bxi_{\GV}\in\Omega(\GV)\ \&\ \forall j \ |\lu\bxi_{\GV^{\perp}},\bmu_j(\GV)\ru| > L_{m+1}\}.
\ene
The collection $\{\bmu_j(\GV)\}$ obviously coincides with
\bee
\{\bn(\bth_{\GV^{\perp}}),\ \bth\in\Bth_{\tilde k}\setminus\GV\}.
\ene

The set $\Bxi(\GV)$ is, in general, disconnected; it consists of several
connected components which we will denote by $\{\Bxi(\GV)_p\}_{p=1}^P$. Let us fix a connected component $\Bxi(\GV)_p$.
Then for some vectors $\{\tilde\bmu_j(p)\}_{j=1}^{J_p}\subset \{\pm\bmu_j\}$ we have
\bee\label{eq:Bxip}
\Bxi(\GV)_p=\{\bxi\in\R^d,\ \bxi_{\GV}\in\Omega(\GV)\ \&\ \forall j \ \lu\bxi_{\GV^{\perp}},\tilde\bmu_j(p)\ru > L_{m+1}\};
\ene
we assume that $\{\tilde\bmu_j(p)\}_{j=1}^{J_p}$ is the minimal set with this property, so that each hyperplane
$$
\{\bxi\in\R^d,\ \bxi_{\GV}\in\Omega(\GV)\ \ \&\ \ \lu\bxi_{\GV^{\perp}},\tilde\bmu_j(p)\ru = L_{m+1}\},\ j=1,\dots,J_p
$$
has a non-empty intersection with the boundary of $\Bxi(\GV)_p$. It is not hard to see that $J_p\ge d-m$. Indeed, otherwise $\Bxi(\GV)_p$ would have non-empty intersection with $\Bxi_1(\GV')$ for some $\GV'$, $\GV\subsetneq\GV'$. We also introduce
\bee\label{eq:Bxiptilde}
\tilde\Bxi(\GV)_p:=\{\bxi\in\GV^{\perp},\ 
\forall j \
\lu\bxi,\tilde\bmu_j(p)\ru > 0\}.
\ene
Note that our assumption that $\Bxi(\GV)_p$ is a connected component of $\Bxi(\GV)$ implies that for any
$\bxi\in\tilde\Bxi(\GV)_p$ and any $\bth\in\Bth_{\tilde k}\setminus\GV$ we have
\bee\label{eq:ne0}
\lu\bxi,\bth\ru=\lu\bxi,\bth_{\GV^{\perp}}\ru\ne 0.
\ene
We also put $K:=d-m-1$.

Let us first assume that the number $J_p$ of `defining planes' is the minimal possible, i.e. $J_p= K+1$. We will carry on this
assumption throughout most of the paper, and only in Section 11 we will discuss how to deal with the more general case of arbitrary
$\Xi(\GV)_p$.
If $J_p= K+1$, then the set
$\{\tilde\bmu_j(p)\}_{j=1}^{K+1}$ is linearly independent.
Let $\ba=\ba(p)$ be a unique
point from $\GV^\perp$ satisfying the following conditions: $\lu\ba,\tilde\bmu_j(p)\ru = L_{m+1}$, $j=1,\dots,K+1$.
Then, since the determinant of the Gram matrix of vectors $\tilde\bmu_j(p)$ is $\gg\rho_n^{0-}$, we have $|\ba|\ll L_{m+1}\rho_n^{0+}$.
We introduce the shifted cylindrical coordinates in $\Bxi(\GV)_p$. These coordinates will be denoted by $\bxi=(r;\tilde\Phi;X)$. Here,
$X=(X_1,\dots,X_m)$ is an arbitrary set of cartesian coordinates in $\Omega(\GV)$. These coordinates do not depend on the choice of the
connected component $\Bxi(\GV)_p$. The rest of the coordinates $(r,\tilde\Phi)$ are shifted spherical coordinates in $\GV^{\perp}$,
centered at $\ba$. This means that
\bee\label{eq:r}
r(\bxi)=|\bxi_{\GV^{\perp}}-\ba|
\ene
and
\bee\label{eq:tildePhi}
\tilde\Phi=\bn(\bxi_{\GV^{\perp}}-\ba)\in S_{\GV^{\perp}}.
\ene
More precisely, $\tilde\Phi\in M$, where $M=M_p:=\{\bn(\bxi_{\GV^{\perp}}-\ba),\ \bxi\in\Bxi(\GV)_p\}\subset S_{\GV^{\perp}}$
is a $K$-dimensional spherical simplex with $K+1$ sides. Note that
\bee\label{eq:Mp}
\bes
M_p&=\{\bn(\bxi_{\GV^{\perp}}-\ba),\ \bxi\in\Bxi(\GV)_p\}
=\{\bn(\bxi_{\GV^{\perp}}-\ba),\  \forall j \ \lu\bxi_{\GV^{\perp}},\tilde\bmu_j(p)\ru > L_{m+1}\}\\
&=\{\bn(\boldeta),\ \boldeta:=\bxi_{\GV^{\perp}}-\ba\in\GV^\perp,\  \forall j \ \lu\boldeta,\tilde\bmu_j(p)\ru > 0\}
=S_{\GV^{\perp}}\cap\tilde\Bxi(\GV)_p.
\end{split}
\ene
We will 
denote by $d\tilde\Phi$ the spherical Lebesgue measure on $M_p$.
For each non-zero vector $\bmu\in\GV^{\perp}$, we denote
\bee
W(\bmu):=\{\boldeta\in\GV^\perp,\ \lu\boldeta,\bmu\ru=0\}.
\ene
Thus, the sides of the simplex $M_p$ are intersections of $W(\tilde\bmu_j(p))$ with the sphere $S_{\GV^{\perp}}$.
Each vertex $\bv=\bv_t$, $t=1,\dots,K+1$ of $M_p$ is an intersection of $S_{\GV^{\perp}}$ with $K$ hyperplanes
$W(\tilde\bmu_j(p))$,
$j=1,\dots,K+1$, $j\ne t$. This means that $\bv_t$ is a unit vector from $\GV^{\perp}$ which is orthogonal to $\{\tilde\bmu_j(p)\}$,
$j=1,\dots,K+1$, $j\ne t$; this defines $\bv$ up to a multiplication by $-1$.
\bel\label{lem:newangles}
Let $\GU_1$ and $\GU_2$ be two strongly distinct subspaces each of which is a linear combination of some of the
vectors from $\{\tilde\bmu_j(p)\}$. Then the angle between them is not smaller than $s(\rho_n)$.
In particular, all non-zero angles between two sides of any dimensions of $M_p$ as well as all the distances between two vertexes $\bv_t$ and $\bv_{\tau}$, $t\ne\tau$,
are bounded below by $s(\rho_n)$.
\enl
\bep
First of all, we remark that $\GU_j$ are not, in general, quasi-lattice subspaces. However, each algebraic sum
$\GW_j:=\GV+\GU_j$ is a quasi-lattice subspace. Moreover, the angle between $\GW_1$ and $\GW_2$ is equal to the
angle between $\GU_1$ and $\GU_2$, so the first statement follows from Condition C. To prove the second statement, it
is enough to notice that any  non-zero angle between two sides (of arbitrary dimension) of $M_p$ is equal to
the angle between two subspaces $\GU_1$ and $\GU_2$ of the type considered in the first statement; the same can be said about the
distance between $\bv_t$ and $\bv_{\tau}$.
\enp
\bel\label{lem:sign}
Let $p$ be fixed.
Suppose, $\bth\in\Bth_{\tilde k}\setminus\GV$ and $\bth_{\GV^{\perp}}=\sum_{j=1}^{K+1} b_j\tilde\bmu_j(p)$.
Then either all coefficients $b_j$ are non-positive, or all of them are non-negative.
\enl
\bep
Suppose not. Then, without loss of generality, we can assume that $b_1>0$ and $b_2<0$.
Let $L$ be a spherical interval joining $\bv_1$ and $\bv_2$, i.e.
\bee
L=\{ \bu\in \overline{M_p}, \lu \bu,\tilde\bmu_j(p)\ru=0,\ j=3,\dots,K+1\}.
\ene
Note that $\lu \bv_1,\bth_{\GV^{\perp}}\ru=b_1\lu \bv_1,\tilde\bmu_1(p)\ru>0$ and
$\lu \bv_2,\bth_{\GV^{\perp}}\ru=b_2\lu \bv_2,\tilde\bmu_2(p)\ru<0$. Therefore, there is a point
$\bu\in L$ such that $\lu \bu,\bth_{\GV^{\perp}}\ru=0$. This means that $W(\bth_{\GV^{\perp}})$ has a
non-empty intersection with $M_p$, which contradicts \eqref{eq:ne0}.
\enp

Assume that the diameter of $M_p$ is $\le (100d^2)^{-1}$, which we can always achieve by taking sufficiently large $\tilde k$. We put $\Phi_q:=\frac{\pi}{2}-\phi(\bxi_{\GV^\perp}-\ba,\tilde\bmu_q(p))$, $q=1,\dots,K+1$.
The geometrical meaning of these coordinates is simple: $\Phi_q$ is the spherical distance between
$\tilde\Phi=\bn(\bxi_{\GV^{\perp}}-\ba)$ and $W(\tilde\bmu_q(p))$. The reason why we have introduced $\Phi_q$ is that in these coordinates some important objects will be especially simple (see e.g. Lemma~\ref{lem:products} below) which is very convenient for integration in Section 10. At the same time, the set of coordinates $(r,\{\Phi_q\})$ contains $K+2$ variables, whereas we only need $K+1$ coordinates in
$\GV^{\perp}$. Thus, we have one constraint for variables $\Phi_j$. Namely,
let $\{\be_j\}$, $j=1,\dots,K+1$ be a fixed orthonormal basis in $\GV^{\perp}$ chosen in such a way that the $K+1$-st axis
passes through $M_p$.
Then we have $\be_j=\sum_{l=1}^{K+1}a_{jl}\tilde\bmu_l$ with some matrix $\{a_{jl}\}$, $j,l=1,\dots,K+1$, and $\tilde\bmu_l=\tilde\bmu_l(p)$.
Therefore (recall that we denote $\boldeta:=\bxi_{\GV^{\perp}}-\ba$),
\bee\label{eq:etaj}
\eta_j=\lu\boldeta,\be_j\ru=r\sum_{q=1}^{K+1}a_{jq}\sin\Phi_q
\ene
and, since $r^2(\bxi)=|\boldeta|^2=\sum_{j=1}^{K+1}\eta_j^2$, this
implies that
\bee\label{eq:odin}
\sum_j(\sum_q a_{jq}\sin\Phi_q)^2=1,
\ene
which is our constraint.

Let us also put
\bee\label{eq:etajdash}
\eta_j':=\frac{\eta_j}{|\boldeta|}=\sum_{q=1}^{K+1}a_{jq}\sin\Phi_q.
\ene
Then we can write the surface element $d\tilde\Phi$ in the coordinates $\{\eta_j'\}$ as
\bee\label{eq:surfaceelement}
d\tilde\Phi=\frac{d\eta_1'\dots d\eta_K'}{\eta_{K+1}}=\frac{d\eta_1'\dots d\eta_K'}{(1-\sum_{j=1}^K(\eta_{j}')^2)^{1/2}},
\ene
where the denominator is bounded below by $1/2$ by our choice of the basis $\{\be_j\}$.

\bel\label{lem:Al}
For each $p,l$ we have $|a_{pl}|\le s(\rho_n)^{-1}$.
\enl
\bep
This follows from the fact that for each $p$
$|a_{pl}|$ is the length of the projection of $\be_p$ onto
$\tilde\bmu_l$ parallel to the linear space spanned by all $\tilde\bmu_j$, $j\ne l$. Since the absolute value of the sine of the angle between
$\tilde\bmu_l$ and the linear space spanned by all $\tilde\bmu_j$, $j\ne l$, is at least $s(\rho_n)$, this implies that
for each $l,p$ we have $|a_{pl}|\le s(\rho_n)^{-1}$, which finishes the proof.
\enp
\bel\label{lem:anglebelow}
We have $\max_j\sin\Phi_j(\boldeta)\ge s(\rho_n) d^{-3/2}$.
\enl
\bep
Suppose not. Then for each $l$ we have $\sin\Phi_l(\boldeta)< s(\rho_n)d^{-3/2}$. Since all
$\sin\Phi_l$ are positive, Lemma \ref{lem:Al} implies
\bees
\sum_j(\sum_l a_{jl}\sin\Phi_l)^2<d(d s(\rho_n)^{-1} s(\rho_n) d^{-3/2})^2=1,
\enes
which contradicts \eqref{eq:odin}.
\enp

The next lemma describes the dependence on $r$ of all possible inner products $\lu\bxi,\bth\ru$, $\bth\in\Bth_{\tilde k}$, $\bxi\in\Bxi(\GV)_p$.
\bel\label{lem:products}
Let
$\bxi\in\Bxi(\GV)_p$, $\GV\in\CV_m$, and $\bth\in\Bth_{\tilde k}$.

(i) If $\bth\in\GV$, then $\lu\bxi,\bth\ru$ does not depend on $r$.

(ii) If $\bth\not\in\GV$ and $\bth_{\GV^{\perp}}=\sum_{q}b_q\tilde\bmu_q(p)$, 
then
\bee\label{eq:innerproduct1}
\lu\bxi,\bth\ru=\lu X,\bth_{\GV}\ru+L_{m+1}\sum_{q}b_q+r(\bxi)\sum_{q}b_q\sin\Phi_q.
\ene

In the case (ii) all the coefficients $b_q$ are either non-positive or non-negative and
each non-zero coefficient $b_q$ satisfies
\bee\label{eq:n10}
\rho_n^{0-}\le |b_q| \le \rho_n^{0+}.
\ene
\enl
\bep
We begin by noticing that
\bee
\lu\bxi,\bth\ru=\lu X,\bth_{\GV}\ru+\lu \bxi_{\GV^{\perp}},\bth_{\GV^{\perp}}\ru,
\ene
from which
part (i) immediately follows. Recalling that $\bxi_{\GV^{\perp}}=\ba+\boldeta$,
we obtain
\bee\label{eq:innerproduct3}
\bes
&\lu\bxi,\bth\ru=\lu X,\bth_{\GV}\ru+\lu\ba,\bth_{\GV^{\perp}}\ru+\lu\boldeta,\bth_{\GV^{\perp}}\ru\\
&=\lu X,\bth_{\GV}\ru+\sum_q b_q\lu\ba,\tilde\bmu_q\ru+\sum_q b_q\lu\boldeta,\tilde\bmu_q\ru\\
&=\lu X,\bth_{\GV}\ru+\sum_{q} b_qL_{m+1}+r\sum_q b_q\sin\Phi_q.
\end{split}
\ene
The last statement follows from Lemmas \ref{lem:sign} and \ref{lem:coefficients}. The application of Lemma \ref{lem:sign}
is straightforward, so let us discuss the application of Lemma \ref{lem:coefficients}. Suppose, $\bth_{\GV^{\perp}}=\sum_{q=1}^{K+1}b_q\tilde\bmu_q(p)$,
where $\bth$ belongs to $\Bth_{\tilde k}$ (but $\tilde\bmu_q(p)$, in general, is not in  $\Bth_{\tilde k}$).
We know that the linear span of each $\tilde\bmu_q(p)$ and $\GV$ is an element of $\CV_{m+1}$. Therefore, for each $q=1,\dots,K+1$, there
is a vector $\bnu_q\in\Bth_{\tilde k}$ such that $\tilde\bmu_q(p)$ is proportional to $(\bnu_q)_{\GV^\perp}$,
$\tilde\bmu_q(p)=C(q)(\bnu_q)_{\GV^\perp}$, where $\rho_n^{0-}\le|C(q)|\le\rho_n^{0+}$. Now we choose arbitrary linearly independent
vectors $\bnu_{K+2},\dots,\bnu_d\in(\GV\cap\Bth_{\tilde k})$. Then we can write
$\bth=\bth_{\GV^\perp}+\bth_{\GV}=\sum_{q=1}^{K+1}C(q)b_q(\bnu_q)_{\GV^\perp}+\sum_{q=K+2}^{d}b_q\bnu_q$. Now we can apply Lemma \ref{lem:coefficients} directly.
\enp



\section{Pseudo-differential operators}

In this and the next sections, we construct operators $H_1$ and $H_2$ described in Section 6. Most of the material in these two sections
is very similar to the corresponding sections of \cite{ParSob}, as are the proofs of most of the statements.
Therefore, we will often omit the proofs, instead referring the reader to \cite{Sob}, \cite{Sob1}, and \cite{ParSob}.

\subsection{Classes of PDO's}\label{classes:subsect}
Before we define the pseudo-differential operators (PDO's), we
introduce the relevant classes of symbols.

For any $f\in L_2(\R^d)$ we define the Fourier transform:
\begin{equation*}
(\mathcal F f)(\bxi)
= \frac{1}{(2\pi)^{\frac{d}{2}}} \int_{\R^d} e^{-i\bxi
\bx}f(\bx) d\bx,\ \bxi\in\R^d.
\end{equation*}
Let us now define the symbols we will consider and operators associated with
them. Let $b = b(\bx, \bxi)$, $\bx, \bxi\in\R^d$, be an
almost-periodic (in $\bx$) complex-valued function, i.e. for some countable set $\hat{\Bth}$ of frequencies (we always assume $\hat\Bth$ to be symmetric and to contain $0$; starting from the middle of this section, the set
$\hat\Bth$ will be assumed to be finite)
\begin{equation}\label{eq:sumf}
b(\bx, \bxi) = \sum\limits_{\bth\in\hat{\Bth}}\hat{b}(\bth, \bxi)\be_{\bth}(\bx)
\end{equation}
where
\bees
\hat{b}(\bth, \bxi):=\BM_\bx(b(\bx,\bxi)\be_{-\bth}(\bx))
\enes
are Fourier coefficients of $b$ (recall that $\BM$ is the mean of an almost-periodic function). We always assume that \eqref{eq:sumf}
converges absolutely.
Put $\lu \bt \ru := \sqrt{1+|\bt|^2},\
\forall \bt\in\R^d$. We notice that
\begin{equation}\label{weight:eq}
\lu \bxi + \boldeta\ru\le 2\lu\bxi\ru \lu\boldeta\ru, \ \forall
\bxi, \boldeta\in \R^d.
\end{equation}
We say that the symbol $b$ belongs to the
class $\BS_{\a}=\BS_{\a}(\beta) = \BS_{\a}(\beta,\,\hat{\Bth})$,\ $\a\in\R$, $0<\beta\leq1$, if
for any $l\ge 0$ and any non-negative $s\in\Z$ the condition
\begin{equation}\label{1b1:eq}
\1 b \1^{(\a)}_{l, s} :=
\max_{|\bs| \le s}
 \sum\limits_{\bth\in\hat{\Bth}}\lu \bth\ru^{l}\sup_{\bxi}\,\lu\bxi\ru^{(-\a + |\bs|)\beta}
|\BD_{\bxi}^{\bs} \hat b(\bth, \bxi)|<\infty, \ \ |\bs| = s_1+ s_2 + \dots + s_d,
\end{equation}
is fulfilled.
The quantities \eqref{1b1:eq} define norms on the class $\BS_\a$. Note that $\BS_\a$ is an increasing function of $\a$,
i.e. $\BS_{\a}\subset\BS_{\g}$ for $\a < \g$. For later reference we
write here the following convenient bound that follows from definition
\eqref{1b1:eq} and property \eqref{weight:eq}:
\bee
\sum\limits_{\bth\in\hat{\Bth}}\lu\bth\ru^{l}\sup_{\bxi}\,\lu\bxi\ru^{(-\a+s+1)\b}(|\BD^{\bs}_{\bxi}\hat b(\bth, \bxi+ \boldeta) -
\BD^{\bs}_{\bxi}\hat b(\bth, \bxi)|)\le C\1 b\1^{(\a)}_{l, s+1}  \lu\boldeta\ru^{|\a-s-1|\b}
|\boldeta|, \ s = |\bs|, \label{differ:eq}
\ene
with a constant $C$ depending only on $\a, s$, and $\beta$. For a vector $\boldeta\in\R^d$ introduce
the symbol
\begin{equation}\label{bboldeta:eq}
b_{\boldeta}(\bx, \bxi) = b(\bx, \bxi+\boldeta), \boldeta\in\R^d,
\end{equation}
so that $\hat b_{\boldeta}(\bth, \bxi) = \hat b(\bth, \bxi+\boldeta)$ .
The bound \eqref{differ:eq} implies that for all $|\boldeta|\le C$ we have
\begin{equation}\label{differ1:eq}
\1 b - b_{\boldeta}\1^{(\a-1)}_{l, s}\le C \1 b\1^{(\a)}_{l, s+1}|\boldeta|,\
\end{equation}
uniformly in $\boldeta$: $|\boldeta|\le C$.

Now we define the PDO $\op(b)$ in the usual way:
\begin{equation}\label{eq:deff}
\op(b)u(\bx) = \frac{1}{(2\pi)^{\frac{d}{2}}} \int  b(\bx, \bxi)
e^{i\bxi \bx} (\mathcal Fu)(\bxi) d\bxi,
\end{equation}
the integral being over $\Rd$. Under the condition $b\in\BS_\a$
 the integral in the r.h.s. is clearly finite for any
$u$ from the Schwarz
class $\plainS(\Rd)$. Moreover, the condition $b\in \BS_0$
guarantees the boundedness of $\op(b)$ in $L_2(\Rd)$, see
Proposition \ref{bound:prop}. Unless otherwise stated, from now on
$\plainS(\Rd)$ is taken as a natural domain for all PDO's at hand, when they act in $L_2$.
Applying the standard regularization procedures to definition \eqref{eq:deff} (see, e.g., \cite{Shu0}), we can also consider the action of
$\op(b)$ when we apply it to an exponential $\be_{\bnu}$. Then we have
\bee\label{eq:actionbe}
\op(b)\be_{\bnu}=\sum_{\bth\in\hat\Bth}\hat{b}(\bth, \bnu)\be_{\bnu+\bth}.
\ene
This action can be extended by linearity to all quasi-periodic functions (i.e. finite linear combinations of $\be_{\bnu}$).
Moreover, if the order $\al=0$, by continuity this action can be extended to all of $B_2(\R^d)$; this extension
 has the same norm as $\op(b)$ acting in $L_2$ (see \cite{Shu0}). Thus, in what follows, when we speak about a pseudo-differential operator
 with almost-periodic symbol
 acting in $B_2$, we mean that its domain is whole $B_2$ (when the order is non-positive), or the the space of all quasi-periodic functions
 (for operators with positive order). And, when we make a statement about the norm of a pseudo-differential operator with almost-periodic symbol,
 we will not specify whether the operator acts in $L_2(\R^d)$ or  $B_2(\R^d)$, since these norms are the same.
Notice that the operator $\op(b)$ is
symmetric if its symbol satisfies the condition
\begin{equation}\label{selfadj:eq}
\hat b(\bth, \bxi) = \overline{\hat b(-\bth, \bxi+\bth)}.
\end{equation}
We shall call such symbols \textsl{symmetric}.


We note that in the very beginning when we consider \eqref{eq:Sch}, our operator $\op(b)$ is a multiplication by a function $b$ (in particular, $b\in\BS_0$). However, during modifications and transformations below our perturbation will eventually become a pseudo-differential operator.
Thus, it is convenient in abstract statements to consider $b$ a pseudo-differential symbol from some $\BS_\a$ class.

\subsection{Some basic results on the calculus of almost-periodic PDO's}
We begin by listing some elementary results for almost-periodic PDO's. The proof is very similar (with obvious changes) to the proof of analogous statements in \cite{Sob}. In what follows, {\it if we need to calculate a product of two (or more) operators with some symbols $b_j\in\BS_{\a_j}(\hat{\Bth}_j)$ we will always consider that $b_j\in\BS_{\a_j}(\sum_j\hat{\Bth}_j)$ where, of course, all added terms are assumed to have zero coefficients in front of them}.


\begin{prop}\label{bound:prop}
Suppose that $\1 b\1^{(0)}_{0, 0}<\infty$. Then $\op(b)$ is bounded in both $L_2(\R^d)$ and $B_2(\R^d)$ and $\|\op(b)\|\le \1 b \1^{(0)}_{0, 0}$.
\end{prop}

Since $\op(b) u\in\plainS(\Rd)$ for any $b\in\BS_{\a}$ and $u\in
\plainS(\Rd)$, the product $\op(b) \op(g)$, $b\in \BS_{\a}(\hat{\Bth}_1), g\in
\BS_{\g}(\hat{\Bth}_2)$, is well defined on $\plainS(\Rd)$. A straightforward
calculation leads to the following formula for the symbol $b\circ g
$ of the product $\op(b)\op(g)$:
\begin{equation*}
(b\circ g)(\bx, \bxi) = \sum_{\bth\in\hat{\Bth}_1,\, \bphi\in\hat{\Bth}_2}
\hat b(\bth, \bxi +\bphi) \hat g(\bphi, \bxi) e^{i(\bth+\bphi)\bx},
\end{equation*}
and hence
\begin{equation}\label{prodsymb:eq}
\widehat{(b\circ g)}(\bchi, \bxi) =
\sum_{\bth +\bphi = \bchi} \hat b (\bth, \bxi +\bphi) \hat g(\bphi,
\bxi),\ \bchi\in\hat{\Bth}_1+\hat{\Bth}_2,\ \bxi\in \Rd.
\end{equation}

We have

\begin{prop}\label{product:prop}
Let $b\in\BS_{\a}(\hat{\Bth}_1)$,\ $g\in\BS_{\g}(\hat{\Bth}_2)$. Then $b\circ g\in\BS_{\a+\g}(\hat{\Bth}_1+\hat{\Bth}_2)$
and
\begin{equation*}
\1 b\circ g\1^{(\a+\g)}_{l,s} \le C \1 b\1^{(\a)}_{l,s} \1 g\1^{(\g)}_{l+(|\a|+s)\beta,s},
\end{equation*}
with a constant $C$ depending only on $l,\ \alpha,\ s$.
\end{prop}

We are also interested in the estimates for symbols of commutators.
For PDO's $A, \Psi_l, \ l = 1, 2, \dots ,N$, denote
\begin{gather*}
\ad(A; \Psi_1, \Psi_2, \dots, \Psi_N)
= i\bigl[\ad(A; \Psi_1, \Psi_2, \dots, \Psi_{N-1}), \Psi_N\bigr],\\
\ad(A; \Psi) = i[A, \Psi],\ \ \ad^N(A; \Psi) = \ad(A; \Psi, \Psi,
\dots, \Psi),\ \ad^0(A; \Psi) = A.
\end{gather*}
For the sake of convenience we use the notation $\ad(a;  \psi_1,
\psi_2, \dots, \psi_N)$ and $\ad^N(a, \psi)$ for the symbols of
multiple commutators. It follows from \eqref{prodsymb:eq} that the
Fourier coefficients of the symbol $\ad(b,g)$ are given by
\begin{multline}\label{comm:eq}
\widehat{\ad(b, g)}(\bchi, \bxi) = i
\sum_{\bth +\bphi = \bchi} \bigl[\hat b(\bth, \bxi +\bphi) \hat
g(\bphi, \bxi) - \hat b(\bth, \bxi)
\hat g(\bphi, \bxi + \bth)\bigr],\ \ \bxi\in \Rd.
\end{multline}

\begin{prop}\label{commut0:prop}
Let $b\in \BS_{\a}(\hat{\Bth})$ and $g_j\in\BS_{\g_j}(\hat{\Bth}_j)$,\ $j = 1, 2, \dots, N$.
Then $\ad(b; g_1, \dots, g_N) \in\BS_{\g}(\hat{\Bth}+\sum_j\hat{\Bth}_j)$ with
$$
\g = \a+\sum_{j=1}^N(\g_j-1),
$$
and
\begin{equation}\label{commutator:eq}
\1 \ad(b; g_1, \dots, g_N)\1^{(\g)}_{l,s} \le C \1 b\1^{(\a)}_{p,s+N}
\prod_{j=1}^N \1 g_j\1^{(\g_j)}_{p,s+N-j+1},
\end{equation}
where $C$ and $p$ depend on $l,s,N,\a$ and $\g_j$.
\end{prop}

\subsection{Partition of the perturbation}
From now on we fix $\b:\ 0<\beta<\alpha_1$, and put $\hat{\Bth}:=\Bth$ which is finite. The symbols we are going to construct will depend on $\rho_n$; this dependence will usually be omitted from the notation.

Let $\iota\in \plainC\infty(\R)$ be a non-negative function such
that
\begin{equation}\label{eta:eq}
0\le\iota\le 1,\ \ \iota(z) =
\begin{cases}
& 1,\  z \le \frac{1}{4};\\
& 0,\  z \ge \frac{1.1}{4}.
\end{cases}
\end{equation}
For $\bth\in \Bth, \bth\not = \mathbf 0$,
define the following $\plainC\infty$-cut-off functions:
\begin{equation}\label{el:eq}
\begin{cases}
e_{\bth}(\bxi) =&\ \iota\biggl({\biggl|\dfrac{|\bxi
+\bth/2|-3\rho_n}{10\rho_n}\biggr|}
 \biggr),\\[0.5cm]
\ell^{>}_{\bth}(\bxi) = &\ 1 - \iota\biggl(\dfrac{|\bxi +
\bth/2|-3\rho_n}{10\rho_n}
\biggr),\\[0.5cm]
\ell^{<}_{\bth}(\bxi) = &\ 1 - \iota\biggl(\dfrac{3\rho_n-|\bxi +
\bth/2|}{10\rho_n}\biggr),
\end{cases}
\end{equation}
and
\begin{equation}\label{phizeta:eq}
\begin{cases}
\z_{\t}(\bxi) =&\ \iota\biggl(\dfrac{|\lu\bth,\bxi + \bth/2\ru|}
{\rho_n^\beta |\bth|}\biggr),\\[0.5cm]
\varphi_{\t}(\bxi) = &\  1 - \z_{\bth}(\bxi).
\end{cases}
\end{equation}
Note that $e_{\bth}+\ell^{>}_{\bth} + \ell^{<}_{\bth} = 1$. The
function $\ell^{>}_{\bth}$ is supported on the set
$|\bxi+\bth/2|\geq11\rho_n/2$, and $\ell^{<}_{\bth}$ is supported on the
set $|\bxi+\bth/2|\leq \rho_n/2$. The function $e_{\bth}$ is supported
in the shell $\rho_n/4\le |\bxi+\bth/2|\le 23\rho_n/4$. Using the notation $\ell_{\bth}$ for any of the functions
$\ell^{>}_{\bth}$ or $\ell^{<}_{\bth}$, we point out that
\begin{equation}\label{symmetry:eq}
\begin{cases}
e_{\bth}(\bxi) = e_{-\bth}(\bxi + \bth), \ \ell_{\bth}(\bxi)
= \ell_{-\bth}(\bxi + \bth),\\[0.2cm]
\varphi_{\bth}(\bxi) =  \varphi_{-\bth}(\bxi + \bth),\ \
\z_{\bth}(\bxi) = \z_{-\bth}(\bxi + \bth).
\end{cases}
\end{equation}
Note that the above functions satisfy the estimates
\begin{equation}\label{varphi:eq}
\begin{cases}
|\BD^{\bs}_{\bxi} e_{\bth}(\bxi)|
+ |\BD^{\bs}_{\bxi}\ell_{\bth}(\bxi)|\ll \rho_n^{-|\mathbf s|},\\[0.2cm]
|\BD^{\bs}_{\bxi}\varphi_{\bth}(\bxi)| + |\BD^{\bs}_{\bxi}
\z_{\bth}(\bxi)| \ll \rho_n^{-\beta |\bs|}.
\end{cases}
\end{equation}
Using the above cut-off functions, for any symbol $b\in\BS_{\a }(\b)$
we introduce five new symbols $b^{\downarrow}, b^o,
b^{\ssharp}, b^{\natural}, b^{\flat}$ in the following way:
\begin{gather}
b^{\ssharp}(\bx, \bxi; \rho_n) =
\sum_{\bth\in\Bth'} \hat b(\bth, \bxi)
\ell^{>}_{\bth}(\bxi)
 e^{i\bth \bx},\label{sharp:eq}\\
b^{\natural}(\bx, \bxi; \rho_n) =
\sum_{\bth\in\Bth'} \hat b(\bth, \bxi) \varphi_{\bth}(\bxi)
e_{\bth}(\bxi)  e^{i\bth \bx},
\label{natural:eq}\\
b^{\flat}(\bx, \bxi; \rho_n) =
\sum_{\bth\in\Bth'} \hat b(\bth, \bxi)\z_{\bth}(\bxi)
e_{\bth}(\bxi)
  e^{i\bth \bx},\label{flat:eq}\\
b^{\downarrow}(\bx, \bxi; \rho_n) =
\sum_{\bth\in\Bth'} \hat b(\bth, \bxi) \ell^{<}_{\bth}(\bxi)
  e^{i\bth \bx},\label{downarrow:eq}\\
b^o(\bx, \bxi; \rho_n) = b^o(\bxi; \rho_n) =
\hat b(0, \bxi).\label{o:eq}
\end{gather}
The superscripts here are chosen to mean correspondingly: `large energy', `non-resonant',
`resonant', `small energy' and $0$-th Fourier coefficient. The
corresponding operators are denoted by
\begin{equation*}
\begin{split}
B^{\ssharp} = \op(b^{\ssharp}),\ B^{\natural} = \op(b^{\natural}),\\
B^{\flat} = \op(b^{\flat}),\ B^{\downarrow} = \op(b^{\downarrow}),\
B^o = \op (b^o).
\end{split}
\end{equation*}
By definitions \eqref{eta:eq}, \eqref{el:eq} and \eqref{phizeta:eq}
\begin{equation*}
b = b^o + b^{\downarrow}+ b^{\flat} + b^{\natural} + b^{\ssharp}.
\end{equation*}
The role of each of these operator is easy to explain. Note that on the support of the functions
$\hat b^{\natural}(\bth, \ \cdot\ ; \rho_n)$ and $\hat b^{\flat}(\bth, \
\cdot\ ; \rho_n)$ we have
\begin{equation}\label{ds:eq}
|\bth|\le \rho_n^{0+},\ \frac{1}{4}\rho_n \le |\bxi+\bth/2|\le
\frac{23}{4}\rho_n, \ \frac{1}{4}\rho_n - \frac{1}{2}\rho_n^{0+}\le
|\bxi| \le \frac{23}{4}\rho_n + \frac{1}{2}\rho_n^{0+}.
\end{equation}
On the support of $b^{\downarrow}(\bth, \ \cdot\ ; \rho_n)$ we have
\begin{equation}\label{supportell<:eq}
\biggl|\bxi+\frac{\bth}{2}\biggr|\le \frac{1}{2}\rho_n,\ |\bxi|\le
\frac{1}{2}\rho_n + \frac{1}{2}\rho_n^{0+}.
\end{equation}
On the support of $b^{\ssharp}(\bth, \ \cdot\ ; \rho_n)$ we have
\begin{equation}\label{supportell>:eq}
\biggl|\bxi+\frac{\bth}{2}\biggr|\ge \frac{11}{2}\rho_n,\ |\bxi|\ge
\frac{11}{2}\rho_n - \frac{1}{2}\rho_n^{0+}.
\end{equation}
The introduced symbols play a central role in the proof of Lemma
\ref{main_lem}. As we have seen in Section 6, due
to \eqref{supportell<:eq} and \eqref{supportell>:eq} the symbols
$b^\downarrow$  and $b^\ssharp$ make only a negligible contribution
to the spectrum of the operator $H$ near the point $\l =
\rho^{2},\ \rho\in I_n$. The only significant components of $b$ are the symbols
$b^\natural, b^{\flat}$ and $b^o$. The symbol $b^o$ will remain as
it is, and the symbol $b^\natural$ will be transformed in the next
Section to another symbol, independent of $\bx$.

We will often combine $B^{\flat}$, $B^{\ssharp}$ and $B^{\downarrow}$: for instance $B^{\flat, \ssharp} = B^{\flat} +
B^{\ssharp}$, $B^{\flat, \ssharp, \downarrow} = B^{\flat, \ssharp} +
B^{\downarrow}$. A similar convention applies to the symbols. Under
the condition $b\in\BS_{\a}(\b)$ the above symbols belong to the same
class $\BS_{\a}(\b)$ and the following bounds hold:
\begin{equation}\label{subord:eq}
\1 b^{\flat}\1^{(\a)}_{l, s} + \1 b^{\natural}\1^{(\a)}_{l, s} + \1
b^{\ssharp}\1^{(\a)}_{l, s} + \1 b^{o}\1^{(\a)}_{l, s} + \1
b^{\downarrow}\1^{(\a)}_{l, s} \ll
\1 b\1^{(\a)}_{l, s}.
\end{equation}
Indeed, let us check this  for the symbol $b^{\natural}$, for
instance. According to \eqref{ds:eq} and \eqref{varphi:eq}, on the
support of the function $\hat b^{\natural}(\bth, \ \cdot\ ; \rho_n)$
we have
\begin{gather*}
|\BD^{\bs}\varphi_{\bth}(\bxi)|\ll \rho_n^{-\b|\bs|}
\ll  \lu\bxi\ru^{-|\bs|\b},\\[0.2cm]
|\BD^{\bs}\ell^{>}_{\bth}(\bxi)| + |\BD^{\bs}\ell^{<}_{\bth}(\bxi)|
+ |\BD^{\bs}e_{\bth}(\bxi)|\ll  \rho_n^{-|\bs|}\ll \lu\bxi\ru^{-|\bs|\b}.
\end{gather*}
This immediately leads to the bound of the form \eqref{subord:eq}
for the symbol $b^{\natural}$.

The introduced operations also preserve symmetry. Indeed, let us
calculate using \eqref{symmetry:eq}:
\begin{align*}
\overline{\hat b^{\flat}(-\bth, \bxi + \bth)} = & \ \overline{\hat
b(-\bth, \bxi + \bth)} \z_{-\bth}(\bxi+\bth)
e_{-\bth}(\bxi+\bth)\\
= &\ \hat b(\bth, \bxi) \z_{\bth}(\bxi) e_{\bth}(\bxi) =
\hat b^{\flat}(\bth, \bxi).
\end{align*}
Therefore, by \eqref{selfadj:eq} the operator
$B^{\flat}$ is symmetric if so is $\op(b)$. The proof is similar for the
rest of the operators introduced above.

Let us list some other elementary properties of the introduced
operators. In the lemma below we use the projection $\CP(\CC),
\CC\subset\R$ whose definition was given in Section 6.

\begin{lem}\label{smallorthog:lem}
Let $b\in \BS_{\a}(\b)$ with
some $\a\in\R$. Then the following hold:
\begin{itemize}
\item[(i)]
The operator $B^{\downarrow}$ is bounded and
\begin{equation*}
\|B^{\downarrow}\|\ll \1 b \1^{(\a)}_{0, 0} \rho_n^{\b\max(\a,
0)}.
\end{equation*}
Moreover,
\begin{equation*}
\bigl(I - \CP (B(2\rho_n/3 )\bigr) B^{\downarrow} =
B^{\downarrow} \bigl(I - \CP (B(2\rho_n/3)\bigr)  = 0.
\end{equation*}

\item[(ii)] The operator $B^\flat$ satisfies the following relations
\begin{align}\label{bflatorthog:eq}
\CP(B(\rho_n/8)) B^{\flat} =  &\ B^{\flat} \CP(B(\rho_n/8)) \notag\\[0.2cm]
= &\ \bigl(I - \CP(B(6\rho_n)\bigr) B^{\flat} =  B^{\flat} \bigl(I
- \CP(B(6\rho_n))\bigr) = 0,
\end{align}
and similar relations hold for the operator $B^{\natural}$ as well.

Moreover, for any $\g \in\R$ one has $b^{\natural}, b^{\flat}
\in\BS_{\g}$ and
\begin{equation}\label{nat:eq}
\1 b^{\natural}\1^{(\g)}_{l, s} + \1 b^{\flat}\1^{(\g)}_{l, s} \ll
\rho_n^{\b(\a - \g)}\1 b\1^{(\a)}_{l, s},
\end{equation}
for all $l$ and $s$, with
an implied constant 
independent of $b$ and $n\ge 1$. In particular, the operators
$B^{\natural}, B^{\flat}$ are bounded and
\begin{equation*}
\| B^{\natural} \| + \|B^{\flat}\| \ll \rho_n^{\b \a } \1 b
\1^{(\a)}_{0, 0}.
\end{equation*}

\item[(iii)]
\begin{equation*}
\CP\bigl(B(5\rho_n)\bigr)B^{\ssharp} =
B^{\ssharp}\CP\bigl(B(5\rho_n)\bigr) = 0.
\end{equation*}

\end{itemize}
\end{lem}

\begin{proof}
\underline{Proof of (i).} It follows from \eqref{1b1:eq} and \eqref{supportell<:eq}  that
\begin{equation}\label{flatdecay:eq}
\sum_{\bth}\sup_{\bxi}\,|\hat b^{\downarrow} (\bth, \bxi; \rho_n)| \ll \rho_n^{\b\max(\a, 0)}\sum_{\bth}\sup_{\bxi}\,\lu\bxi\ru^{-\b\a}|\hat b (\bth, \bxi; \rho_n)|=\1 b \1^{(\a)}_{0, 0}\rho_n^{\b\max(\a, 0)}.
\end{equation}
By Proposition \ref{bound:prop} this implies the sought bound for
the norm $\|B^{\downarrow}\|$.

In view of \eqref{supportell<:eq}, the second part of statement (i)
follows from \eqref{downarrow:eq}.

\underline{Proof of (ii).}
 Relations \eqref{bflatorthog:eq} follow from definitions
\eqref{flat:eq} and \eqref{natural:eq} in view of \eqref{ds:eq}.

Furthermore, by  \eqref{ds:eq} and \eqref{subord:eq},
\begin{equation*}
\begin{split}
&\sum_{\bth}\lu\bth\ru^{l}\sup_{\bxi}\,\lu\bxi\ru^{(-\g+s)\b}|\BD^{\bs}_{\bxi} \hat b^{\natural}(\bth, \bxi; \rho_n)|
\ll \rho_n^{\b(\a-\g)}\sum_{\bth}\lu\bth\ru^{l}\sup_{\bxi}\,\lu\bxi\ru^{(-\a+s)\b}|\BD^{\bs}_{\bxi} \hat b^{\natural}(\bth, \bxi; \rho_n)| \cr
&\leq\1 b^{\natural}
\1^{(\a)}_{l, s}\rho_n^{\b(\a - \g)}\ll\1 b
\1^{(\a)}_{l, s}\rho_n^{\b(\a - \g)}.
\end{split}
\end{equation*}
This means that $b^{\natural}\in\BS_\g$ for any $\g\in\R$
and \eqref{nat:eq} holds for $b^{\natural}$. The bound for the norm follow from
\eqref{nat:eq} with $\g = 0$, and Proposition \ref{bound:prop}. The proof for $b^{\flat}$ is analogous.

\underline{Proof of (iii)} is similar to (i). The required result
follows from \eqref{supportell>:eq}.

\end{proof}

\section{Gauge transform and the symbol of the resulting operator}

\subsection{Preparation}
Our strategy will be to find a unitary operator which reduces $H =
H_0+\op(b)$, $H_0:=-\Delta$, to another PDO, whose symbol, essentially,
depends only on $\bxi$ (notice that now we have started to distinguish
between the potential $b$ and the operator of multiplication by it $\op(b)$). More precisely,
we want to find operators $H_1$ and $H_2$ with the properties discussed in Section 6.
The unitary operator
will be constructed in the form $U = e^{i\Psi}$ with a suitable
bounded self-adjoint quasi-periodic PDO $\Psi$. This is why we
sometimes call it a `gauge transform'. It is useful to
consider $e^{i\Psi}$ as an element of the group
\begin{equation*}
U(t) = \exp\{ i \Psi t\},\ \ \forall t\in\R.
\end{equation*}
\textbf{We assume that the operator $\ad(H_0, \Psi)$ is bounded, so
that $U(t) D(H_0) = D(H_0)$}. This assumption will be justified
later on. Let us express the operator
\begin{equation*}
A_t := U(-t)H U(t)
\end{equation*}
via its (weak) derivative with respect to $t$:
\begin{equation*}
A_t = H +  \int_0^t U(-t') \ad(H; \Psi) U(t') dt'.
\end{equation*}
By induction it is easy to show that
\begin{gather}
A_1 = H + \sum_{j=1}^{\tilde k} \frac{1}{j!}
\ad^j(H; \Psi) +  R^{(1)}_{{\tilde k}+1},\label{decomp:eq}\\
R^{(1)}_{{\tilde k}+1} :=  \int_0^1 d t_1 \int_0^{t_1} d t_2\dots
\int_0^{t_{\tilde k}} U(-t_{{\tilde k}+1}) \ad^{{\tilde k}+1}(H; \Psi) U(t_{{\tilde k}+1})
dt_{{\tilde k}+1}.\notag
\end{gather}
The
operator $\Psi$ is sought in the form
\begin{equation}\label{psik:eq}
\Psi = \sum_{j=1}^{\tilde k} \Psi_j,\ \Psi_j = \op(\psi_j),
\end{equation}
with symbols $\psi_j$ from some suitable class $\BS_{\s_j}$
to be specified later on. Substitute this formula in
\eqref{decomp:eq} and rewrite, regrouping the terms:
\begin{gather}
A_1 = H_0 + \op(b) + \sum_{j=1}^{\tilde k} \frac{1}{j!} \sum_{l=j}^{\tilde k} \sum_{k_1+
k_2+\dots + k_j = l}
\ad(H; \Psi_{k_1}, \Psi_{k_2}, \dots, \Psi_{k_j})\notag\\
 +
 R^{(1)}_{{\tilde k}+1} + R^{(2)}_{{\tilde k}+1},\notag\\
R^{(2)}_{{\tilde k}+1}: = \sum_{j=1}^{\tilde k} \frac{1}{j!}\sum_{k_1+ k_2+\dots + k_j
\ge  {\tilde k}+1} \ad(H; \Psi_{k_1}, \Psi_{k_2}, \dots, \Psi_{k_j}).
\label{rtilde:eq}
\end{gather}
Changing this expression yet again produces
\begin{gather*}
A_1 = H_0 + \op(b) + \sum_{l=1}^{\tilde k} \ad(H_0; \Psi_l) + \sum_{j=2}^{\tilde k}
\frac{1}{j!} \sum_{l=j}^{\tilde k} \sum_{k_1+ k_2+\dots + k_j = l}
\ad(H_0; \Psi_{k_1}, \Psi_{k_2}, \dots, \Psi_{k_j})\\
+ \sum_{j=1}^{\tilde k} \frac{1}{j!} \sum_{l=j}^{\tilde k} \sum_{k_1+ k_2+\dots + k_j
= l} \ad(\op(b); \Psi_{k_1}, \Psi_{k_2}, \dots, \Psi_{k_j}) +
R^{(1)}_{{\tilde k}+1} + R^{(2)}_{{\tilde k}+1}.
\end{gather*}
Next, we switch the summation signs and decrease $l$ by one in the
second summation:
\begin{gather*}
A_1 = H_0 + \op(b) + \sum_{l=1}^{\tilde k} \ad(H_0; \Psi_l) + \sum_{l=2}^{\tilde k}
\sum_{j=2}^l \frac{1}{j!}\sum_{k_1+ k_2+\dots + k_j = l}
\ad(H_0; \Psi_{k_1}, \Psi_{k_2}, \dots, \Psi_{k_j})\\
+\sum_{l=2}^{{\tilde k}+1} \sum_{j=1}^{l-1} \frac{1}{j!}\sum_{k_1+ k_2+\dots
+ k_j = l-1} \ad(\op(b); \Psi_{k_1}, \Psi_{k_2}, \dots, \Psi_{k_j}) +
R^{(1)}_{{\tilde k}+1} + R^{(2)}_{{\tilde k}+1}.
\end{gather*}
Now we introduce the notation
\begin{gather}
B_1 := \op(b),\notag\\
B_l := \sum_{j=1}^{l-1} \frac{1}{j!}
 \sum_{k_1+ k_2+\dots + k_j = l-1}
\ad(\op(b); \Psi_{k_1}, \Psi_{k_2}, \dots, \Psi_{k_j}),
\ l\ge 2,\label{bl:eq}\\
T_l := \sum_{j=2}^l \frac{1}{j!} \sum_{k_1+ k_2+\dots + k_j = l}
\ad(H_0; \Psi_{k_1}, \Psi_{k_2}, \dots, \Psi_{k_j}),\  l\ge 2.
\label{tl:eq}
\end{gather}
We emphasise that the operators $B_l$ and $T_l$ depend only on
$\Psi_1, \Psi_2, \dots, \Psi_{l-1}$. Let us make one more
rearrangement:
\begin{gather}
A_1 =  H_0 + \op(b) + \sum_{l=1}^{\tilde k} \ad(H_0, \Psi_l) + \sum_{l=2}^{\tilde k} B_l
+   \sum_{l=2}^{{\tilde k}} T_{l} + R_{{\tilde k}+1},\notag\\
R_{{\tilde k}+1} =  B_{{\tilde k}+1} + R^{(1)}_{{\tilde k}+1} + R^{(2)}_{{\tilde k}+1}.\label{r:eq}
\end{gather}
Now we can specify our algorithm for finding $\Psi_j$'s. The symbols
$\psi_j$ will be found from the following system of commutator
equations:
\begin{gather}
\ad(H_0; \Psi_1) + B_1^{\natural} = 0,\label{psi1:eq}\\
\ad(H_0; \Psi_l) + B_l^{\natural} + T_l^{\natural} = 0,\ l\ge
2,\label{psil:eq}
\end{gather}
and hence
\begin{equation}\label{lm:eq}
\begin{cases}
A_1 = H_0 + Y^{(o)}_{\tilde k} + Y_{\tilde k}^{\flat}
+ Y_{{\tilde k}}^{\downarrow, \ssharp} + R_{{\tilde k}+1},\\[0.3cm]
Y_{\tilde k} =  \sum_{l=1}^{{\tilde k}} B_l + \sum_{l=2}^{{\tilde k}}
T_l.
\end{cases}
\end{equation}
Below we denote by $y_{\tilde k}$ the symbol of the PDO $Y_{\tilde k}$. Recall that by
Lemma \ref{smallorthog:lem}(ii), the operators $B_l^{\natural},
T_l^{\natural}$ are bounded, and therefore, in view of
\eqref{psi1:eq}, \eqref{psil:eq}, so is  the commutator $\ad(H_0;
\Psi)$. This justifies the assumption made in the beginning of the
formal calculations in this section.

\subsection{Commutator equations}

Put
$$
\tilde{\chi}_{\bth}(\bxi):=e_{\bth}(\bxi)\varphi_{\bth}(\bxi)(|\bxi+\bth|^2-|\bxi|^2)^{-1}=\frac{e_{\bth}(\bxi)\varphi_{\bth}(\bxi)}
{2\lu\bth,\bxi+\frac{\bth}{2}\ru}
$$
when $\bth\not={\bf 0}$, and $\tilde{\chi}_{\bf 0}(\bxi)=0$.
We have
\begin{lem} \label{commut:lem}
Let $A = \op(a)$ be a symmetric PDO with $a\in\BS_{\om }$. Then the
PDO $\Psi$ with the Fourier coefficients of the symbol $\psi(\bx,
\bxi)$ given by
\begin{equation}\label{psihat:eq}
\hat\psi(\bth, \bxi) = i\,{\hat a}(\bth, \bxi)\tilde{\chi}_{\bth}(\bxi)
\end{equation}
solves the equation
\begin{equation}\label{adb:eq}
\ad(H_0; \Psi) + \op(a^{\natural})= 0.
\end{equation}
Moreover, the operator $\Psi$ is bounded and self-adjoint, its
symbol $\psi$ belongs to $\BS_{ \g}$ with any $\g \in\R$ and the
following bound holds:
\begin{equation}\label{psitau:eq}
\1 \psi\1^{(\g)}_{l, s} \ll \rho_n^{\b(\om-\g-1)}r(\rho_n)^{-1}\1 a\1^{(\om)}_{l-1, s}\ll \rho_n^{\b(\om-\g-1)+0+}\1 a\1^{(\om)}_{l-1, s}.
\end{equation}
\end{lem}
Using Propositions~\ref{bound:prop},\ref{product:prop},\ref{commut0:prop}, Lemma~\ref{commut:lem}, and repeating arguments from the proof of Lemma 4.2 from \cite{ParSob}, we obtain the following estimates for the symbols introduced above:
\bel\label{estimateskm:lem}
Let $b\in\BS_0(\beta)$ be a symmetric symbol. Then $\psi_j,\, b_j,\,t_j\in\BS_\gamma(\beta)$ for any $\gamma\in\R$ and
\bee\label{estpsi}
\1\psi_j\1^{(\gamma)}_{l,s}\leq C_j\rho_n^{\beta(1-\gamma-2j)}r(\rho_n)^{-j}\left(\1 b\1^{(0)}_{l_j,s_j}\right)^j,\ \ j\geq 1;
\ene
\bee\label{estbt}
\1 b_j\1^{(\gamma)}_{l,s}+\1 t_j\1^{(\gamma)}_{l,s}\leq C_j\rho_n^{\beta(2-\gamma-2j)}r(\rho_n)^{-j+1}\left(\1 b\1^{(0)}_{l_j,s_j}\right)^j,\ \ j\geq 2.
\ene
Here $C_j,\,l_j,\,s_j$ depend only on $j,l,s$ and $\gamma$. Moreover, assuming $\rho_0$ is large enough (depending on $l,s,\gamma,b$ and ${\tilde k}$) we get
\bee\label{estpsitotal}
\1\psi\1^{(\gamma)}_{l,s}\ll\rho_n^{-\beta(1+\gamma)}r(\rho_n)^{-1}\1 b\1^{(0)}_{l,s};
\ene
\bee\label{esty}
\1 y_{\tilde k}\1^{(0)}_{l,s}\leq 2\1 b\1^{(0)}_{l,s};
\ene
\bee\label{esterror}
\| R_{{\tilde k}+1}\|\ll \rho_n^{-2\beta{\tilde k}}r(\rho_n)^{-{\tilde k}}\left(\1 b\1^{(0)}_{l_{{\tilde k}+1},s_{{\tilde k}+1}}\right)^{{\tilde k}+1}.
\ene
\enl
Now, we take
\bee\label{eq:kM}
{\tilde k}>(M+(d-2))/\beta
\ene
and assume that $k$ is large enough so that $r(\rho_n)^{-1}\ll\rho_n^{0+}\ll\rho_n^{\beta}$. Then
$$
\| R_{{\tilde k}+1}\|\ll\rho_n^{-M+(2-d)}
$$
and we can disregard $R_{{\tilde k}+1}$ due to Corollary~\ref{cor:H1H2}. More precisely, let $W=W_{\tilde k}$ be the operator with symbol
\bee\label{eq:newy}
w_{\tilde k}(\bx,\bxi):=y_{\tilde k}(\bx,\bxi)-y_{\tilde k}^{\natural}(\bx,\bxi),\ \ \hbox{i.e.}\ \
\hat w_{\tilde k}(\bth,\bxi)=\hat y_{\tilde k}(\bth,\bxi)(1-e_{\bth}(\bxi)\varphi_{\bth}(\bxi)).
\ene
We put $H_1:=A_1$ and $H_2:=-\Delta+W$. Then $||H_1-H_2||\ll\rho_n^{-M+(2-d)}$ and, moreover, the symbol $w$ satisfies condition \eqref{eq:b3}. This means that all the constructions of Section 6
are valid, and all we need to do is to compute $\vol(G_{\lambda})$.

\subsection{Computing the symbol of the operator after gauge transform}
The following lemma provides us with more explicit form of the
symbol ${y}_{\tilde k}$.
\bel\label{lem:symbol} We have $\hat{y}_{\tilde k}(\bth,\bxi)=0$ for $\bth\not\in\Bth_{\tilde k}$. Otherwise,
\begin{equation}\label{symbol}
\begin{split}
&\hat{y}_{\tilde k}(\bth,\bxi)=\hat{b}(\bth)+\sum\limits_{s=1}^{{\tilde k}-1}\sum  C_s(\bth,\bxi)\hat{b}(\bth_{s+1})\prod\limits_{j=1}^s \hat{b}(\bth_j)\tilde{\chi}_{\bth_j'}(\bxi+\bphi_j')\cr
&=\hat{b}(\bth)+\sum\limits_{s=1}^{{\tilde k}-1}\sum
C_s(\bth,\bxi)\hat{b}(\bth_{s+1})\prod\limits_{j=1}^s \hat{b}(\bth_j)\frac{e_{\bth_j'}(\bxi+\bphi_j')\varphi_{\bth_j'}(\bxi+\bphi_j')}
{2\lu\bth_j',\bxi+\bphi_j'+\frac{\bth_j'}{2}\ru},
\end{split}
\end{equation}
where the second sums are taken over all $\bth_j\in\Bth$,
$\bth_j',\bphi_j'\in\Bth_{s+1}$ and
\begin{equation}\label{indsconst}
C_s(\bth,\bxi)=\sum\limits_{p=1}^s\sum\limits_{\bth_j '',\bphi_j
''\in\Bth_{s+1}\ (1\leq j\leq p)}C_s^{(p)}(\bth) \prod\limits_{j=1}^p
e_{\bth_j ''}(\bxi+\bphi_j '')\varphi_{\bth_j ''}(\bxi+\bphi_j '').
\end{equation}
Here $C_s^{(p)}(\bth)$ depend on $s,\ p$ and all vectors
$\bth,\bth_j,\bth_j',\bphi_j',\bth_j '',\bphi_j ''$. At the
same time, coefficients $C_s^{(p)}(\bth)$ can be bounded uniformly
by a constant which depends on $s$ only. We apply the convention that $0/0=0$.
\enl
\begin{proof}
We will prove the lemma by induction. Namely, let $\ell\geq2$. We claim that:

1) For any $m=1,\dots,\ell-1$, $\hat{\psi}_m(\bth,\bxi)=0$ for $\bth\not\in\Bth_m$. Otherwise,
\begin{equation}\label{inds1}
\hat{\psi}_m(\bth,\bxi)=\sum C_m'(\bth,\bxi)\prod\limits_{j=1}^m \hat{b}(\bth_j)\tilde{\chi}_{\bth_j'}(\bxi+\bphi_j'),
\end{equation}
where the sum is taken over all $\bth_j\in\Bth$,
$\bth_j',\bphi_j'\in\Bth_m$ and $C_m'(\bth,\bxi)$ admit representation similar to \eqref{indsconst}.

2) For any $s=1,\dots,\ell-1$ and any $k_1,\dots,k_p$ $(p\geq1)$ such that $k_1+\dots+k_p=s$, $\widehat{ad(\op(b);\Psi_{k_1},\dots,\Psi_{k_p})}(\bth,\bxi)=0$ for $\bth\not\in\Bth_{s+1}$. Otherwise,
\begin{equation}\label{inds2}
\widehat{ad(\op(b);\Psi_{k_1},\dots,\Psi_{k_p})}(\bth,\bxi)=\sum C_s''(\bth,\bxi)\hat{b}(\bth_{s+1})\prod\limits_{j=1}^s \hat{b}(\bth_j)\tilde{\chi}_{\bth_j'}(\bxi+\bphi_j'),
\end{equation}
where the sum is taken over all $\bth_j\in\Bth$,
$\bth_j',\bphi_j'\in\Bth_{s+1}$ and $C_s''(\bth,\bxi)$
admit representation similar to \eqref{indsconst}.

3) For any $s=2,\dots,\ell$ and any $k_1,\dots,k_p$ $(p\geq2)$ such that $k_1+\dots+k_p=s$,
$\widehat{ad(H_0;\Psi_{k_1},\dots,\Psi_{k_p})}(\bth,\bxi)=0$ for $\bth\not\in\Bth_{s}$. Otherwise,
\begin{equation}\label{inds3}
\widehat{ad(H_0;\Psi_{k_1},\dots,\Psi_{k_p})}(\bth,\bxi)=\sum C_{s}'''(\bth,\bxi)\hat{b}(\bth_{s})\prod\limits_{j=1}^{s-1} \hat{b}(\bth_j)\tilde{\chi}_{\bth_j'}(\bxi+\bphi_j'),
\end{equation}
where the sum is taken over all $\bth_j\in\Bth$,
$\bth_j',\bphi_j'\in\Bth_{s}$ and $C_{s}'''(\bth,\bxi)$ admit
representation similar to \eqref{indsconst}.

For $\ell=2$ statements 1)--3) can be easily checked. Indeed,
\begin{equation}\label{ind1}
\hat{\psi}_1(\bth,\bxi)=i\hat{b}(\bth)\tilde{\chi}_{\bth}(\bxi),
\end{equation}
\begin{equation}\label{ind3}
\widehat{ad(\op(b);\Psi_1)}(\bth,\bxi)=\sum\limits_{\bchi+\bphi=\bth}\left(\hat{b}(\bchi)\hat{b}(\bphi)\tilde{\chi}_{\bphi}(\bxi+\bchi)-
\hat{b}(\bchi)\hat{b}(\bphi)\tilde{\chi}_{\bphi}(\bxi)\right),
\end{equation}
\begin{equation}\label{ind2}
\widehat{ad(H_0;\Psi_1,\Psi_1)}(\bth,\bxi)=-\sum\limits_{\bchi+\bphi=\bth}\left(
\hat{b}^{{\natural}}(\bchi)\hat{b}(\bphi)\tilde{\chi}_{\bphi}(\bxi+\bchi)-
\hat{b}^{{\natural}}(\bchi)\hat{b}(\bphi)\tilde{\chi}_{\bphi}(\bxi)\right).
\end{equation}
Now, we complete the induction in several steps.

Step 1. First of all, notice that due to \eqref{bl:eq}, \eqref{tl:eq}, for any $m=2,\dots,\ell$ symbol of $B_m$ admits representation of the form \eqref{inds2} with $s=m-1$, and symbol of $T_m$ admits representation of the form \eqref{inds3} with $s=m$. Then it follows from Lemma~\ref{commut:lem} and \eqref{psil:eq} that $\Psi_\ell$ admits representation of
the form \eqref{inds1}.

Step 2. Proof of \eqref{inds2} with $s=\ell$. Let $k_1+\dots+k_p=\ell$. If $p\geq2$ then
$$
ad(\op(b);\Psi_{k_1},\dots,\Psi_{k_p})=ad(ad(\op(b);\Psi_{k_1},\dots,\Psi_{k_{p-1}});\Psi_{k_p}).
$$
Since $k_1+\dots+k_{p-1}\leq\ell-1$ and $k_p\leq\ell-1$ we can apply
\eqref{inds1} and \eqref{inds2}. Combined with \eqref{comm:eq} it gives
representation of the form \eqref{inds2}. If $p=1$ then $ad(\op(b);\Psi_\ell)$ satisfies
\eqref{inds2} because of \eqref{comm:eq} and step 1.

Step 3. Proof of \eqref{inds3} with $s=\ell+1$. Let $k_1+\dots+k_p=\ell+1$, $p\geq2$. If $p\geq3$ then (cf. step 2)
$$
ad(H_0;\Psi_{k_1},\dots,\Psi_{k_p})=ad(ad(H_0;\Psi_{k_1},\dots,\Psi_{k_{p-1}});\Psi_{k_p}).
$$
Since $k_1+\dots+k_{p-1}\leq\ell$, $p-1\geq2$ and $k_p\leq\ell-1$ we can apply
\eqref{inds1} and \eqref{inds3}. Together with \eqref{comm:eq} it gives
representation of the form \eqref{inds3}. If $p=2$ then (see \eqref{psil:eq})
$$
ad(H_0;\Psi_{k_1},\Psi_{k_2})=ad(ad(H_0;\Psi_{k_1});\Psi_{k_2})=-ad(B_{k_1}^{{\natural}}+T_{k_1}^{{\natural}};\Psi_{k_2}).
$$
Since $k_1\leq\ell$ and $k_2\leq\ell$, the representation of the form \eqref{inds3} follows from \eqref{comm:eq} and step 1. (Formally exceptional case $k_1=1,\ k_2=\ell$ can be treated separately in the same way using \eqref{psi1:eq} instead of \eqref{psil:eq}.)

Induction is complete.

Now, \eqref{inds2}, \eqref{inds3} and \eqref{bl:eq}, \eqref{tl:eq}, \eqref{lm:eq} prove the lemma.
\end{proof}


%


\section{Contribution from various resonance regions}

\subsection{Summing the contributions from individual eigenvalues}
Let us fix a subspace $\GV\in\CV_m$, $m<d$, and a component $\Bxi_p$ of the resonance
region $\Bxi(\GV)$.
Our aim is to compute the contribution to the density of states from each
component $\Bxi_p$. This means that we define $\hat A^+(\Bxi_p):=\hat A^+\cap\Bxi(\GV)_p$
and $\hat A^-(\Bxi_p):=\hat A^-\cap\Bxi(\GV)_p$ and try to compute
\bee\label{eq:n3}
\vol\hat A^+(\Bxi_p)-\vol\hat A^-(\Bxi_p).
\ene
Since formulas \eqref{eq:46} and \eqref{eq:CARd}
obviously imply that
\bee\label{eq:n2}
\vol(G_{\la})=w_d\rho^d+\sum_{m=0}^{d-1}\sum_{\GV\in \CV_m}\sum_{p}\bigl(\vol\hat A^+(\Bxi_p)-\vol\hat A^-(\Bxi_p)\bigr),
\ene
if we manage to compute \eqref{eq:n3} (or at least prove that this expression admits a complete
asymptotic expansion in $\rho$), Lemma \ref{main_lem} would be proved. Thus, we fix $\GV$ 
and, moreover,
we fix a component $\Bxi(\GV)_p$ of the resonance region. Recall that $K=d-m-1$.

Note that if $\bxi\in\Bxi_p$, then
we also have that $\BUps(\bxi)\subset \Bxi_p$. We denote
\bees
H_2(\bxi):=\CP(\BUps(\bxi))H_2\CP(\BUps(\bxi))
\enes
as an operator acting in $\GH_{\bxi}:=\CP(\BUps(\bxi))\GH$ (recall that $\GH_{\bxi}$ is an invariant subspace of $H_2$ acting in $B_2(\R^d)$).
Suppose now that two points $\bxi$ and $\boldeta$ have the
same coordinates $X$ and $\tilde\Phi$ and different coordinates $r$. Then $\bxi\in\Bxi_p$ implies $\boldeta\in\Bxi_p$ and
$\BUps(\boldeta)=\BUps(\bxi)+(\boldeta-\bxi)$. This shows that two spaces $\GH_{\bxi}$
and $\GH_{\boldeta}$ have the same dimension and, moreover,
there is a natural isometry
$F_{\bxi,\boldeta}:\GH_{\bxi}\to\GH_{\boldeta}$ given by $F:\be_{\bnu}\mapsto\be_{\bnu+(\boldeta-\bxi)}$,
$\bnu\in\BUps(\bxi)$. This isometry allows us to `compare' operators acting in $\GH_{\bxi}$ and
$\GH_{\boldeta}$. Thus, abusing slightly our notation, we can assume that $H_2(\bxi)$ and $H_2(\boldeta)$
act in the same (finite dimensional) Hilbert space $\GH(X,\tilde\Phi)$.
We will fix the values $(X,\tilde\Phi)$ and study how these
operators depend on $r$. Thus, we denote by $H_2(r)=H_2(r;X,\tilde\Phi)$ the operator
$H_2(\bxi)$ with $\bxi=(X,r,\tilde\Phi)$, acting in $\GH(X,\tilde\Phi)$.

As we have seen from the previous sections,
the symbol of the operator $H_2$ satisfies
\bee\label{eq:nn1}
h_2(\bx,\bxi)=|\bxi|^2+{w}_{\tilde k}(\bx,\bxi)=r^2+2r\lu\ba,\bn(\boldeta)\ru+|\ba|^2+{w}_{\tilde k}(\bx,\bxi)+|X|^2,
\ene
where
the Fourier coefficients of $w_{\tilde k}$ satisfy \eqref{esty}, \eqref{eq:newy}, \eqref{symbol}, \eqref{indsconst} and we denote, as usual,
$\boldeta=\bxi_{\GV^{\perp}}-\ba$. This immediately implies that
the operator $H_2(r)$ is monotonically increasing in $r$; in particular, all its eigenvalues
$\la_j(H_2(r))$ are increasing in $r$. Thus, the function $g(\bxi)$ (defined in Section 6) is an increasing function
of $r(\bxi)$ if we fix other coordinates of $\bxi$, so the equation
\bee\label{eq:tau}
g(\bxi)=\rho^2
\ene
has a unique solution if we fix the values $(X,\tilde\Phi)$; we
denote the $r$-coordinate of this solution by $\tau=\tau(\rho)=\tau(\rho;X,\tilde\Phi)$,
so that
\bee\label{eq:tau1}
g(\bxi(X,\tau,\tilde\Phi)=\rho^2.
\ene
By $\tau_0=\tau_0(\rho)=\tau_0(\rho;X,\tilde\Phi)$
we denote the value of $\tau$ for the unperturbed operator, i.e. $\tau_0$ is a unique
solution of the equation
\bee
|\bxi(X,\tau_0,\tilde\Phi)|=\rho.
\ene
Obviously, we can write down a precise analytic expression for $\tau_0$
(and we have done this in \cite{ParSht} in the two-dimensional case) and show that it allows an expansion in powers of
$\rho$ and $\ln\rho$, but we will not need it.
The definition of the sets $\hat A^{\pm}$ implies that the intersection
\bee
\hat A^+\cap\{\bxi(X,r,\tilde\Phi),\ r\in\R_+\}
\ene
consists of points with $r$-coordinate belonging to the interval $[\tau_0(\rho),\tau(\rho)]$
(where we assume the interval to be empty if $\tau_0>\tau$). Similarly, the intersection
\bee
\hat A^-\cap\{\bxi(X,r,\tilde\Phi),\ r\in\R_+\}
\ene
consists of points with $r$-coordinate belonging to the interval $[\tau(\rho),\tau_0(\rho)]$.
Therefore,
\bee
\hat A^+(\Bxi_p)=\{\bxi=\bxi(X,r,\tilde\Phi),\ X\in\Omega(\GV),\tilde\Phi\in M_p,\ r\in[\tau_0(\rho;X,\tilde\Phi),\tau(\rho;X,\tilde\Phi)]\}
\ene
and
\bee
\hat A^-(\Bxi_p)=\{\bxi=\bxi(X,r,\tilde\Phi),\ X\in\Omega(\GV),\tilde\Phi\in M_p,\ r\in[\tau(\rho;X,\tilde\Phi),\tau_0(\rho;X,\tilde\Phi)]\}.
\ene
This implies that
\bee\label{eq:n4}
\bes
&\vol\hat A^+(\Bxi_p)-\vol\hat A^-(\Bxi_p)=
\int_{\Omega(\GV)}dX\int_{M_p}d\tilde\Phi\int_{\tau_0(\rho;X,\tilde\Phi)}^{\tau(\rho;X,\tilde\Phi)}r^{K}dr\\
&=(K+1)^{-1}\int_{M_p}d\tilde\Phi\int_{\Omega(\GV)}dX(\tau(\rho;X,\tilde\Phi)^{K+1}-\tau_0(\rho;X,\tilde\Phi)^{K+1}).
\end{split}
\ene

Obviously, it is enough to compute the part of \eqref{eq:n4} containing $\tau$, since the second
part (containing $\tau_0$) can be computed analogously. We start by considering
\bee\label{eq:n5}
\int_{\Omega(\GV)}\tau(\rho;X,\tilde\Phi)^{K+1}dX.
\ene
First of all, we notice that if $\bxi,\boldeta\in\BXi(\GV)$ are equivalent points then, according to Lemma \ref{lem:Upsilon},
all vectors $\bth_j$ from the definition \ref{reachability:defn} of equivalence belong to $\GV$. This naturally leads to
the definition of equivalence for projections $\bxi_{\GV}$ and $\boldeta_{\GV}$. Namely, we say that two points $\bnu$ and
$\bmu$ from $\Omega(\GV)$ are $\GV$-equivalent (and write $\bnu\leftrightarrow_{\GV}\bmu$) if $\bnu$ and $\bmu$ are equivalent
in the sense of Definition \ref{reachability:defn} with additional requirement that all $\bth_j\in\GV$. Then
$\bxi\leftrightarrow\boldeta$ implies $\bxi_{\GV}\leftrightarrow_{\GV}\boldeta_{\GV}$. For $\bnu\in\Omega(\GV)$ we denote by
$\BUps_{\GV}(\bnu)$ the class of equivalence of $\bnu$ generated by $\leftrightarrow_{\GV}$. Then $\BUps_{\GV}(\bxi_{\GV})$
is a projection of $\BUps(\bxi)$ to $\GV$ and is, therefore, finite.

Denote by $\CM_{\GV}$ the quotient space $\CM_{\GV}:=\Omega(\GV)/\leftrightarrow_{\GV}$.
Since $\BUps_{\GV}(\bnu)$ is a finite set for each $\bnu\in\Omega(\GV)$,
there is a natural measure on $\CM_{\GV}$ generated
by the Lebesgue measure on $\Omega(\GV)$. Therefore, we can re-write \eqref{eq:n5} as
\bee\label{eq:n6}
\int_{\CM_{\GV}}\sum_{X\in\BUps_{\GV}(\bnu)}\tau(\rho;X,\tilde\Phi)^{K+1}d\bnu
\ene
and try to compute
\bee\label{eq:n7}
\sum_{X\in\BUps_{\GV}(\bnu)}\tau(\rho;X,\tilde\Phi)^{K+1}.
\ene

Let us denote by $S=S(r)$ the operator with symbol $2r\lu\ba,\bn(\boldeta)\ru+|\ba|^2+{w}_{\tilde k}(\bx,\bxi)+|X|^2$
acting in $\GH(X,\tilde\Phi)$, so that
$H_2(r)=r^2 I+S(r)$.
\ber\label{analytic}
We always assume that $\bxi\in\CA$, so that $0.7\rho_n\le|\bxi|\le 5\rho_n$ and all functions $e_{\bth}(\bxi+\bphi)$ from \eqref{eq:newy}--\eqref{indsconst} are equal to $1$.
Note that if $\bth\in\Bth_{\tilde k}$, $\bphi\in\Bth_{\tilde k}$, and $\bth\not\in\GV$, then (see Lemma \ref{lem:products} and \eqref{phizeta:eq})
$\varphi_{\bth}(\bxi+\bphi)=1$. This means that all cut-off functions from \eqref{eq:newy}--\eqref{indsconst} are equal
to $1$ unless $\bth\in\GV$. If, on the other hand, $\bth\in\GV$, then $\varphi_{\bth}(\bxi+\bphi)$
depends only on the projection $\bxi_{\GV}$ and thus is a function only of the coordinates $X$. Thus, equations \eqref{eq:newy}--\eqref{indsconst} show that $H_2(r)$ depends on $r$ analytically, so we can and will consider the family
$H_2(z)$ with complex values of the parameter $z$.
\enr

Formulas \eqref{eq:newy}, \eqref{symbol} imply
\bee\label{eq:newS1}
||S(r)||\ll\rho_n^{1+\al_{d}+0+},\ \ ||S'(r)||\ll \rho_n^{\al_{d}+0+},
\ene
and
\bee\label{eq:newS2}
||\frac{d^l}{dr^l}S(r)||\ll \rho_n^{-l},\ \ l\ge 2.
\ene
Let $\gamma:\ \{|z-\rho|=\rho_n/8\}$ be a circle in the complex plane going in the positive direction.
Then for $\rho\in I_n$ all $\tau(\rho;X,\tilde\Phi)$ lie inside $\gamma$. It is not hard to see that estimates \eqref{eq:newS1} and \eqref{eq:newS2} hold inside and on $\gamma$ (indeed, formulas \eqref{eq:newy}--\eqref{indsconst} give matrix elements of $S(z)$ in an orthonormal
basis even for complex $z$).

A version of the Jacobi's formula states
that for any differentiable invertible matrix-valued function $F(z)$ we have
$$
\tr[F'(z)F^{-1}(z)]=(\det[F(z)])'(\det[F(z)])^{-1}
$$
(it can be proved, for example, using the expansion of the determinant along rows and the induction in the size of $F$)
Let $\#(S,\gamma)$ be the total number of zeros (counting multiplicity) of $\det[S(z)+z^2I-\rho^2I]$ inside $\gamma$. We have
\bee\label{eq:residues0}
\bes
&\#(S,\gamma)=\frac{1}{2\pi i}\oint_\gamma (\det[S(z)+z^2I-\rho^2I])'(\det[S(z)+z^2I-\rho^2I])^{-1}dz\\
&=\frac{1}{2\pi i}\oint_\gamma \tr[(2zI+S'(z))(S(z)+z^2I-\rho^2I)^{-1}]dz=\tr[(1+O(\rho_n^{\alpha_{d}-1+0+}))I],
\end{split}
\ene
where $I=I_{\BUps_{\GV}(\bnu)}$.
Since $\card{\BUps_{\GV}(\bnu)}\ll\rho_n^{(d-1)\alpha_{d-1}+0+}$ by Lemma \ref{lem:finiteBUps} and $d\al_d<1$, we conclude that
there are precisely  $\card{\BUps_{\GV}(\bnu)}$ zeros (counting multiplicities) of $\det[S(z)+z^2I-\rho^2I]$ inside $\gamma$,
and thus the points
$\tau(\rho;X,\tilde\Phi)$ are the only zeros of $\det[S(z)+z^2I-\rho^2I]$ inside $\gamma$.

Then we have by the residue theorem:
\bee\label{eq:residues}
\bes
&\sum_{X\in\BUps_{\GV}(\bnu)}\tau(\rho;X,\tilde\Phi)^{K+1}\\
&=\frac{1}{2\pi i}\oint_\gamma z^{K+1}(\det[S(z)+z^2I-\rho^2I])'(\det[S(z)+z^2I-\rho^2I])^{-1}dz\\
&=\frac{1}{2\pi i}\oint_\gamma \tr[z^{K+1} (2zI+S'(z))(S(z)+z^2I-\rho^2I)^{-1}]dz\\
&=\frac{1}{2\pi i}\oint_\gamma \tr[(2z^{K+2}I+z^{K+1}S'(z))(z^2-\rho^2)^{-1}\sum_{l=0}^\infty (-1)^lS^l(z)(z^2-\rho^2)^{-l}]dz\\
&=\frac{1}{2\pi i}\sum_{l=0}^\infty(-1)^l\oint_\gamma \tr[ (2z^{K+2}I+z^{K+1}S'(z))S^l(z)(z-\rho)^{-(l+1)}(z+\rho)^{-(l+1)}]dz\\
&=\sum_{l=0}^\infty\frac{(-1)^l}{l!}\tr\frac{d^l}{dr^l}[(2r^{K+2}I+r^{K+1} S'(r))S^l(r)(r+\rho)^{-(l+1)}]\bigm|_{r=\rho}.
\end{split}
\ene

Formula \eqref{eq:n4} shows that in order to compute the contribution to the density of states
from $\Bxi(\GV)_p$, we need to integrate the RHS of \eqref{eq:residues} against $dX$ (or rather $d\bnu$) and $d\tilde \Phi$. We are going to integrate against $d\tilde \Phi$ first. We will prove that this integral is a convergent
series of products of powers of $\rho$ and $\ln\rho$. The coefficients in front of all terms will be bounded functions
of $X$, so afterwards we will just integrate these coefficients to obtain the desired asymptotic expansion.


Let us discuss how the RHS of \eqref{eq:residues} depends on the coordinates $X$ and $\tilde\Phi$ (or rather $\Phi$). Equations
\eqref{eq:newy}--\eqref{indsconst}, Lemma \ref{lem:products} and Remark \ref{analytic} show that the RHS of \eqref{eq:residues} is a sum of
terms of the following form:
\bee\label{eq:terms}
C\rho^{p}f_1(X)f_2(\Phi)f_3(X;\rho;\Phi).
\ene
Here, $f_1$ is a uniformly bounded function of $X$ coordinates only. It consists of contributions from
the cut-off functions $\varphi_{\bth}$ with $\bth\in\GV$ and from the terms in \eqref{symbol}, \eqref{indsconst} corresponding
to $\bth'_j, \bth_j''\in\GV$. The function $f_2(\Phi)$ is a product of powers of $\{\sin\Phi_q\}$. This
function 
comes from differentiating
\eqref{eq:innerproduct1} with respect to $r$. Finally, $f_3$ is of the following form:
\bee\label{eq:f3}
f_3(X;\rho;\Phi)=\prod_{t=1}^T
(l_t
+\rho\sum_{q} b_q^t\sin(\Phi_q))^{-k_t
}.
\ene
This function corresponds to the negative powers of inner products $\lu\bxi,\bth_t\ru$ given by
Lemma \ref{lem:products}, part (ii).
Here, $\{b_q^t\}$ are coefficients in the decomposition $(\bth_t)_{\GV^{\perp}}=\sum_q b_q^t\tilde \bmu_q$;
recall that these numbers are all of the same sign and satisfy \eqref{eq:n10}. Without loss of
generality we will assume that all $b_q^t$ are non-negative. 
The number $l_t=l(b_1^t,\dots,b_{K+1}^t):=\lu X,(\bth_t)_{\GV}\ru+L_{m+1}\sum_{q}b_q^t$ satisfies
$\rho_n^{\al_{m+1}}\rho_n^{0-}\ll l_t\ll  \rho_n^{\al_{m+1}}\rho_n^{0+}$, since our
assumptions imply $|\lu X,\bth_{\GV}\ru|\ll\rho_n^{\al_m}$. This number depends on $X$,
but not on $\Phi$ or $\rho$. The number $k_t=k(b_1^t,\dots,b_{K+1}^t)$ is positive, integer, and independent of $\bxi$. Our next objective is to compute the integrals
of \eqref{eq:terms} over the domain $\{\tilde\Phi\in M_p\}$ and prove that these integrals enjoy
asymptotic behaviour \eqref{eq:main_lem1} with uniformly bounded coefficients (as functions of $X$).
The calculations will be rather messy technically, although the main ideas of computing them are not
too difficult.


\subsection{Computing the model integral}
Before computing the integral of \eqref{eq:terms}, we will deal with a simpler integral
\begin{equation}\label{eq:in1}
J_K:=\int_{0}^{\gamma}\int_0^{\hat\Phi_K}\dots\int_{0}^{\hat\Phi_2}\frac{\hat\Phi_1^{n_1}\dots \hat\Phi_K^{n_K}\,d\hat\Phi_1\dots d\hat\Phi_K}{\prod_{t=1}^T (l_t+\rho\sum_{j=1}^K b_j^t \hat\Phi_j)^{k_t}
(c_t+\sum_{j=1}^{K} \tilde b_j^t\hat\Phi_j)^{k'_t}
}
\end{equation}
and then we will discuss how to reduce our initial integral to \eqref{eq:in1}. Here, $n_j,k_t,k'_t\in(\N\cup\{0\})$,
$\gamma\le 1$,
$\rho_n^\beta\ll\rho_n^{\al_1+0-}\ll l_t\ll\rho_n^{\al_d+0+}\ll\rho_n^{1/2}$,
$0<c_t\ll\rho_n^{1/2}$,
$\rho_n^{-\delta_0}\ll b_j^t\ll \rho_n^{\delta_0}$, and $\rho_n^{-\delta_0}\ll \tilde b_j^t\ll \rho_n^{\delta_0}$, where $\delta_0>0$ is sufficiently small (for the sake of definiteness, we put  $\delta_0:=\frac{1}{3^d 3000}$; obviously, we assume that these inequalities hold only for non-zero
values of $b_j^t$ and $\tilde b_j^t$).
We introduce the following notation:
\bee\label{eq:notation}
P:=\sum_j n_j, \ \ \ \ \ \ Q:=\sum_t k_t,\ \ \ \ \ \
Q':=\sum_t k'_t
\ene
and we sometimes will denote the integral \eqref{eq:in1} as $J_K(P,Q,Q')$.
We will also need the auxiliary positive numbers $p_j,\ q_j,\
j=0,\dots,d$, defined by
$$
q_d=\frac{1}{3^d 300},\ \ \ p_j=q_j+\frac{1}{3^d 300},\ \ \ q_{j-1}=q_j+p_j+\frac{1}{3^d 300}=2p_j.
$$
Obviously, $p_0<1/100$. 
\bel\label{lem:integral}
Assume that $c_t\gg \rho_n^{-q_K}$. Then we have:
\begin{equation}\label{eq:JK}
J_K(P,Q,Q')=
\sum_{q=0}^K(\ln({\rho}))^q\sum_{p=0}^\infty
e(p,q;P,Q,Q'){\rho}^{-p},
\ene
where
\bee\label{eq:estimatee}
|e(p,q;P,Q,Q')|\ll\rho_n^{(2/3-p_K)p}\rho_n^{-Q\beta}2^{Q'}\prod_{t=1}^T c_t^{-k_t'}.
\end{equation}
These estimates are uniform in the following regions of variables:
\bee\label{eq:range}
\bes
&0\leq\gamma\le 1,\,
\rho_n^{\beta}\ll l_t\ll\rho_n^{1/2},\,
\rho_n^{-q_K}\ll c_t\ll\rho_n^{1/2},\\
&
\rho_n^{-\delta_0}\ll b_j^t\ll \rho_n^{\delta_0},\, \rho_n^{-\delta_0}\ll \tilde b_j^t\ll \rho_n^{\delta_0},\, \rho_n^{2/3-q_K}<{\rho}.
\end{split}
\ene
\enl
\ber
The estimates \eqref{eq:estimatee} are more natural than they may look. Indeed, each time we run the
inductive argument in the proof (i.e. each time we increase $K$), we apply the geometric series expansion, which results in a slight worsening of the
estimates. This accounts for the need to have $p_K$ in the exponent of $\rho_n$.
\enr
\bep
The proof will go by induction in $K$. The base of induction ($K=0$) is trivial (and the
case $K=1$ has been discussed in \cite{ParSht}). Suppose, we have proved this statement for
$K=S-1$, and let us prove it for $K=S$.

Step I. First, we consider the area where $\hat\Phi_S\geq\rho_n^{-p_S}$. We
do not change terms in the denominator of ${J}_K$ where $b_S^t=0$.
If $b_S^t\not=0$ then we proceed with the following transformations:
\begin{equation}\label{transone}
\begin{split}
&({\rho}\sum_j b_j^t \hat\Phi_j+l_t)^{-k_t}=({\rho}\sum_j b_j^t
\hat\Phi_j)^{-k_t}\left(1+\frac{l_t}{{\rho}\sum_j b_j^t
\hat\Phi_j}\right)^{-k_t}=\cr & ({\rho}\sum_j b_j^t
\hat\Phi_j)^{-k_t}\sum_{m=0}^\infty {{m+k_t-1}\choose
m}\left(\frac{-l_t}{{\rho}\sum_j b_j^t
\hat\Phi_j}\right)^m=\sum_{m=0}^\infty
C_m(k_t)\left(\frac{1}{{\rho}\sum_j b_j^t \hat\Phi_j}\right)^{m+k_t},
\end{split}
\end{equation}
where constants $C_m(k_t)$ satisfy the estimate
\bee\label{eq:estimateC}
|C_m(k_t)|\leq \rho_n^{m/2} {{m+k_t-1}\choose m}\le\rho_n^{m/2}2^{m+k_t-1}.
\ene
Now, we can move powers of ${\rho}$ out of the integral and denote $\tilde{c}_t:=b_S^t
\hat\Phi_S$. Obviously, $\tilde{c}_t$ satisfies \eqref{eq:range} with
$K=S-1$. For terms which do not contain ${\rho}$ we just denote
$\hat{c}_t:=c_t+\tilde b_S^t \hat\Phi_S$. Then, we can apply the induction
assumption for $K=S-1$. Corresponding coefficients will depend on
$\hat\Phi_S$ uniformly. 
As a result, we obtain 
the following expression as the contribution to $J_K$ from the region $\{\hat\Phi_S\geq\rho_n^{-p_S}\}$
(we denote $m:=m_1+\dots+m_T$ and assume for simplicity
that
$b_K^t\not=0$ for all $t$): 
\begin{equation}\label{eq:JK1}
\sum_{m_1=0}^{\infty}\dots\sum_{m_T=0}^\infty C_{m_1}(k_1)\dots C_{m_T}(k_T)\rho^{-m
-Q}
J_{S-1}(P,0,m
+Q+Q').
\ene
Of course, each time we write $J_{S-1}(P,0,m+Q+Q')$, it denotes a different integral of the form
\eqref{eq:in1}, but by the assumption of induction all of them satisfy  
\begin{equation}\label{eq:JK1n}
J_{S-1}(P,0,Q+m+Q')=
\sum_{q=0}^{S-1}(\ln({\rho}))^q\sum_{p=0}^\infty
\check e(p,q;P,0,Q+m+Q'){\rho}^{-p}
\ene
with
\bee\label{eq:estimatee1n}
|\check e(p,q;P,0,Q+m+Q')|\ll\rho_n^{(2/3-p_{S-1})p}\rho_n^{(Q+m)q_0}2^{Q+m+Q'}\prod_{t=1}^T c_t^{-k_t'}
\end{equation}
(when we use \eqref{eq:estimatee} as the induction hypothesis,
we replace $\tilde{c}_t$ and $\hat{c}_t$  by the corresponding lower bounds
$\rho_n^{-q_0}$ and $c_t$).

Contribution from this region of integration into coefficients $e(p,q)$ of the integral
$J_S$
can therefore be
estimated from above by the following expression:
\begin{equation}
\begin{split}
&\sum\limits_{m=0}^{p-Q}
\rho_n^{(2/3-p_{S-1})(p-m-Q)}\rho_n^{(Q+m)q_{0}}\rho_n^{m/2}2^{m+Q-T}2^{Q+m+Q'}\left(\prod_{t=1}^T c_t^{-k_t'}\right)
\sum\limits_{m_1+\dots+m_T=m,\ m_j\geq0}1\leq\\
&\rho_n^{(2/3-p_S)p}\rho_n^{-Q\beta}2^{Q'}\left(\prod_{t=1}^T c_t^{-k_t'}\right)\times\\ &\rho_n^{Q\beta+Qq_{0}-(2/3-p_{S-1})Q}2^{2Q-T}\sum\limits_{m=0}^\infty
\rho_n^{mq_0+m/2-(2/3-p_{S-1})m}2^{2m}2^{m+T-1}\ll\\
&\rho_n^{(2/3-p_S)p}\rho_n^{-Q\beta}2^{Q'}\prod_{t=1}^T c_t^{-k_t'}.
\end{split}
\end{equation}
Notice that at this step we have $S-1$ as the
largest power of $\ln({\rho})$.

Step II. From now we are in the area $\hat\Phi_S\leq\rho_n^{-p_S}$. Then we
can transform all terms
not containing ${\rho}$ in the denominator of ${J}_K$:
\begin{equation}\label{transtwo}
(\sum_j \tilde b_j^t \hat\Phi_j+c_t)^{-k'_t}=(c_t)^{-k'_t}\left(1+\frac{\sum_j
\tilde b_j^t \hat\Phi_j}{c_t}\right)^{-k'_t}=\sum_{m=0}^\infty C'_m(k'_t)\left(\sum_j
\tilde b_j^t \hat\Phi_j\right)^{m},
\end{equation}
where
\begin{equation}\label{constraz1}
|C'_m(k'_t)|\leq c_t^{-k_t'-m}2^{m+k'_t-1}\leq\rho_n^{q_S m}c_t^{-k_t'}2^{m+k'_t-1}.
\end{equation}
Then, only terms with ${\rho}$ are left in the denominator. We
change variables $x_j:=\hat\Phi_j{\rho}$ and obtain the integral
\bee\label{eq:SM1}
\rho^{-P-S}\int_{0}^{{\rho}\rho_n^{-p_S}}\int_0^{x_S}\dots\int_{0}^{x_2}\frac{x_1^{n_1}\dots
x_S^{n_S}\,dx_1\dots dx_S} {\prod_{t} (l_t+\sum_{j=1}^S b_j^t
x_j)^{k_t}}\prod_t\sum_{m_t=0}^\infty \rho^{-m_t}C'_{m_t}(k'_t)\left(\sum_{j=1}^S
\tilde b_j^t x_j\right)^{m_t}.
\ene
First, we consider the integral along the region where $0\leq
x_S\leq\rho_n^{2/3-q_{S-1}}$. Obviously, the corresponding contribution to $J_S$ can be
computed as
$$
{\rho}^{-P-S}\sum\limits_{m=0}^\infty{\rho}^{-m}{\tilde C}_m,
$$
where
$$
|{\tilde C}_m|\leq
\rho_n^{(2/3-q_{S-1})(P+S+m)}\rho_n^{\delta_0
m}S^{m}\rho_n^{-Q\beta}
2^{m+T-1}2^{m+Q'-T}\rho_n^{q_S m}\prod_{t=1}^T c_t^{-k_t'}.
$$
Thus, if we put $m=p-P-S$ we obtain the following estimate for the
coefficient in front of ${\rho}^{-p}$ (notice that
$p_S<q_{S-1}-q_S-2\delta_0$ for $S\geq1$):
$$
\rho_n^{(2/3-p_S)p}\rho_n^{-Q\beta}\rho_n^{-(P+S)q_S}2^{Q'}\prod_{t=1}^T c_t^{-k_t'}.
$$

Step III. The case when $x_S\in
(\rho_n^{2/3-q_{S-1}},{\rho}\rho_n^{-p_S})$. Once again if
$b_S^t=0$ then we leave such terms unchanged. If $b_S^t\not=0$ and
thus $b_S^t\geq\rho_n^{-\delta_0}$, we perform the following transform to \eqref{eq:SM1} (cf. \eqref{transone}):
\begin{equation}\label{transthree}
(l_t+\sum\limits_{j=1}^{S}b_j^t x_j)^{-k_t}=\sum_{m=0}^\infty
C_m(k_t)\left(\frac{1}{\sum_j b_j^t x_j}\right)^{m+k_t}
\end{equation}
and introduce new variables $z_j:=x_j/x_S,\ j=1,\dots,S-1$. Thus, we reduce the problem to the integrals of the following form:
$$
\int_{\rho_n^{2/3-q_{S-1}}}^{{\rho}\rho_n^{-p_S}}\int_0^{1}\int_0^{z_{S-1}}\dots\int_{0}^{z_2}\frac{z_1^{\tilde{n}_1}\dots
z_{S-1}^{\tilde{n}_{S-1}}x_S^{\tilde{n}_S}\,dz_1\dots dz_{S-1}dx_S} {\prod_t
(l_t+x_S\sum_{j=1}^{S-1} b_j^t z_j)^{\tilde{k}_t}(b_S^t
x_S+x_S\sum_{j=1}^{S-1} b_j^t z_j)^{\tilde{k}'_t}}.
$$
Now we can remove $x_S$ from the second bracket in the denominator and then
apply the induction assumption for the internal $S-1$ integrals (i.e. the integrals against
$dz_1\dots dz_{S-1}$ with
$c_t:=b_S^t$ and ${\rho}:=x_S$. This induction assumption guarantees that these internal integrals
can be expressed as a series in powers of $x_S$ and $\ln x_S$, with the biggest power of $\ln x_S$ being $S-1$.
Then, we multiply this expansion by
a (possibly negative) power of $x_S$ and integrate the product against $dx_S$. As a result, we obtain a decomposition
\eqref{eq:JK} of $J_S$, with the biggest power of $\ln\rho$ being equal to $S$. The estimate of the contribution of Step III to the coefficients $e(p,q)$ is similar (but rather more tedious) to the estimates in the first two steps, and we will
skip it.
\enp

\subsection{Reduction to the model integral}
Now we will discuss how to deal with our initial integral
\begin{equation}\label{eq:in2}
\hat J_K:=\int\limits_{M_p}\frac{(\sin\Phi_1)^{n_1}\dots (\sin\Phi_K)^{n_K}(\sin\Phi_{K+1})^{n_{K+1}}\,d\tilde\Phi}{\prod_{t=1}^T (l_t+\rho\sum_{j=1}^{K+1} b_j^t \sin\Phi_j)^{k_t}
}.
\end{equation}
The main problem with reducing the integral along $M_p$ (or even along
$M_p\cap\{\Phi_1\le\dots\le\Phi_{K}\le\Phi_{K+1}\}$) to the model integral \eqref{eq:in1} is the limits of integration: the upper limit of integration against $d\Phi_K$ is not a constant (since the collection of points where $\Phi_K=\Phi_{K+1}$
has variable coordinate $\Phi_K$). In order to rectify this, we define
\bee\label{eq:hatM}
\hat M:=\{\sin\Phi_1\le\dots\le\sin\Phi_{s-1}\le\rho_n^{-p_d},\  \rho_n^{-p_d}\le\sin\Phi_{q}, \ q=s,\dots,K+1\}.
\ene
It is clear that $M_p$ can be represented as a union of several domains of this type. Lemma \ref{lem:anglebelow} shows that
we always have at least one `large' variable in $\hat M$, i.e. $s\le K+1$.
We also introduce the `spherical' coordinates
in the $(\Phi_{s},\dots,\Phi_{K+1})$-subspace: we put
\bee\label{eq:spherical}
\bes
\sin\Phi_j&=\hat\Phi_j,\ \ \ j=1,\dots,s-1,\\
\sin\Phi_{s}-\rho_n^{-p_d}&=\hat r\cos\hat\Phi_{s},\\
\sin\Phi_{s+1}-\rho_n^{-p_d}&=\hat r\sin\hat\Phi_{s}\cos\hat\Phi_{s+1},\\
&\dots\\
\sin\Phi_{K}-\rho_n^{-p_d}&=\hat r\sin\hat\Phi_{s}\sin\hat\Phi_{s+1}\dots\sin\hat\Phi_{K-1}\cos\hat\Phi_{K},\\
\sin\Phi_{K+1}-\rho_n^{-p_d}&=\hat r\sin\hat\Phi_{s}\sin\hat\Phi_{s+1}\dots\sin\hat\Phi_{K-1}\sin\hat\Phi_{K},
\end{split}
\ene
so that $(\hat r,\hat\Phi_{1},\dots,\hat\Phi_{K})$ are the new coordinates. Since only $K$ of the coordinates
$\Phi_{j}$ were independent, we can consider the new variables $(\hat\Phi_{1},\dots,\hat\Phi_{K})$
as independent (and $\hat r$ as a function of the independent variables). We remind that (see Lemma~\ref{lem:anglebelow}) we always have $\max_j\sin\Phi_{j}\gg\rho_n^{0-}$ and thus
\bee\label{hatrest}
\rho_n^{0-}\ll\hat r\ll 1.
\ene

The point in introducing these variables is that the limits
of integration over $\hat M$ become simple:
\bee
\int_{0}^{\pi/2}d\hat\Phi_K\int_{0}^{\pi/2}d\hat\Phi_{K-1}\int_{0}^{\pi/2}d\hat\Phi_{s}\int_{0}^{\rho_n^{-p_d}}
d\hat\Phi_{s-1}\int_{0}^{\hat\Phi_{s-1}}d\hat\Phi_{s-2}\dots\int_{0}^{\hat\Phi_2}d\hat\Phi_1.
\ene

When we insert the values of $\sin\Phi_q$, $q=1,\dots,K+1$ given by \eqref{eq:spherical} into \eqref{eq:odin}, we obtain the
quadratic equation for finding $\hat r$:
\bee\label{eq:quadraticr}
\hat A_s{\hat r}^2+2\hat B_s\hat r+(\hat C_s-1)=0,
\ene
where
\bee
\bes
\hat A_s&=\sum_j(a_{js}\cos\hat\Phi_{s}+a_{j\,s+1}\sin\hat\Phi_{s}\cos\hat\Phi_{s+1}+\dots\\
&+a_{j\,K+1}\sin\hat\Phi_{s}\sin\hat\Phi_{s+1}\dots\sin\hat\Phi_{K-1}\sin\hat\Phi_{K})^2>0,
\end{split}
\ene
\bee
\bes
\hat B_s&=\sum_j(\sum_{q=1}^{s-1}a_{jq}\hat\Phi_q+\rho_n^{-p_d}\sum_{q=s}^{K+1}a_{jq})
(a_{js}\cos\hat\Phi_{s}+a_{j\,s+1}\sin\hat\Phi_{s}\cos\hat\Phi_{s+1}+\dots\\
&+a_{j\,K+1}\sin\hat\Phi_{s}\sin\hat\Phi_{s+1}\dots\sin\hat\Phi_{K-1}\sin\hat\Phi_{K}),
\end{split}
\ene
and
\bee
\hat C_s=\sum_j(\sum_{q=1}^{s-1}a_{jq}\hat\Phi_q+\rho_n^{-p_d}\sum_{q=s}^{K+1}a_{jq})^2>0.
\ene
Therefore, we have
\bee\label{eq:quadratic11}
\hat r=\frac{-\hat B_s\pm\sqrt{\hat B_s^2-\hat A_s\hat C_s+\hat A_s}}{\hat A_s}.
\ene
Note that Cauchy-Schwarz inequality implies $\hat A_s\hat C_s\ge \hat B_s^2$. Since equation \eqref{eq:quadraticr} obviously has at least one real solution, we have $0\le \hat A_s+\hat B_s^2-\hat A_s\hat C_s\le \hat A_s$.
Next, due to Lemma~\ref{lem:Al} we have
\bee\label{hatBsCs}
|\hat B_s|\ll\rho_n^{-p_d+0+},\ \ \ |\hat C_s|\ll\rho_n^{-2p_d+0+}.
\ene
Then using \eqref{hatrest}, \eqref{eq:quadraticr}, and \eqref{hatBsCs} we get
\bee\label{hatAs}
\rho_n^{0+}\gg\frac{2}{{\hat r}^2}\geq\hat A_s\geq \frac{1}{2{\hat r}^2}\gg 1.
\ene
Since $\hat r$ is positive we obviously have
$$
\hat r=\frac{-\hat B_s+\sqrt{\hat B_s^2-\hat A_s\hat C_s+\hat A_s}}{\hat A_s}=\frac{-\hat B_s+\sqrt{\hat A_s}
\sqrt{1-(\hat C_s-\hat B_s^2{\hat A_s}^{-1})}}{\hat A_s},
$$
and thus $\hat r$ is {\it analytic} with respect to $\hat B_s,\ \hat C_s$, i.e. with respect to all $\hat\Phi_j,\ j=1,\dots,s-1$, {\it uniformly} in $\hat{\Phi}_l,\ l=s,\dots,K$, inside $\hat M$. It is easy to see from \eqref{eq:etajdash}, \eqref{eq:surfaceelement} and \eqref{eq:spherical} that the same is true for the Jacobian $\frac{\partial(\tilde\Phi)}{\partial(\hat\Phi_{1},\dots,\hat\Phi_{K})}$.

We also notice that, if we denote
\bee\label{hatr}
\hat r_0:=\hat r-{\hat A_s}^{-1/2},
\ene
then $\hat r_0$ satisfies the same analyticity properties as $\hat r$ and $\hat r_0=O(\rho_n^{-p_d+0+})$.

Thus, we arrive at the integrals of the following form:
\begin{equation}\label{eq:in3}
\begin{split}
&\int_{0}^{\pi/2}d\hat\Phi_K\int_{0}^{\pi/2}d\hat\Phi_{K-1}\dots\int_{0}^{\pi/2}d\hat\Phi_{s}\times\cr
&\int_{0}^{\rho_n^{-p_d}}
d\hat\Phi_{s-1}\int_{0}^{\hat\Phi_{s-1}}d\hat\Phi_{s-2}\dots\int_{0}^{\hat\Phi_2}d\hat\Phi_1
\frac{F(\hat\Phi_{s},\dots,\hat\Phi_K)\hat\Phi_1^{n_1}\dots\hat\Phi_{s-1}^{n_{s-1}}}
{\prod_{t=1}^T (l_t+\rho S(\hat\Phi_{1},\dots,\hat\Phi_K))^{k_t}}.
\end{split}
\end{equation}
Here the function $F(\hat\Phi_{s},\dots,\hat\Phi_K)$ is uniformly bounded with respect to $\hat\Phi_{s},\dots,\hat\Phi_K$ in $\hat M$ and
\bee\label{Sid}
S=S(\hat\Phi_{1},\dots,\hat\Phi_K)
:=\sum_{j=1}^{s-1} b_j^t \hat\Phi_j+\sum_{j=s}^{K+1} b_j^t\rho_n^{-p_d}+({\hat A_s}^{-1/2}+\hat r_0) \tilde{F}(\hat\Phi_{s},\dots,\hat\Phi_K),
\ene
where
$$
\tilde{F}:=b_{s}^t\cos\hat\Phi_{s}+b_{s+1}^t\sin\hat\Phi_{s}\cos\hat\Phi_{s+1}+\dots
+b_{K+1}^t\sin\hat\Phi_{s}\sin\hat\Phi_{s+1}\dots\sin\hat\Phi_{K-1}\sin\hat\Phi_{K}.
$$

Now we apply the construction from the STEP I of the proof of Lemma~\ref{lem:integral}.
We do not change terms $(l_t+\rho S)$ with $b_j^t=0$ for all $j=s,\dots,K+1$ (such terms are equal to $(l_t+\rho\sum_{j=1}^{s-1} b_j^t\hat\Phi_j)$). Otherwise, we write (cf. \eqref{transone})
$$
(l_t+{\rho}S)^{-k_t}=\sum_{m=0}^\infty
C_m(k_t)\left(\frac{1}{{\rho}S}\right)^{m+k_t}.
$$
It remains to notice that
$$
S^{-1}=(\sum_{j=1}^{s-1} b_j^t \hat\Phi_j+\sum_{j=s}^{K+1} b_j^t\rho_n^{-p_d}+{\hat A_s}^{-1/2} \tilde{F})^{-1}
\left(1+\frac{\hat r_0\tilde{F}}{\sum_{j=1}^{s-1} b_j^t \hat\Phi_j+\sum_{j=s}^{K+1} b_j^t\rho_n^{-p_d}+{\hat A_s}^{-1/2} \tilde{F}}\right)^{-1},
$$
and decompose the last expression using geometric progression. We remind (see \eqref{hatAs} and \eqref{hatr}) that
$\hat r_0{\hat A_s}^{1/2}\ll\rho_n^{-p_d+0+}$. Since $\hat r_0$ is analytic in $\hat\Phi_1,\dots,\hat\Phi_{s-1}$, we end up with the model integrals $J_{s-1}$ (with $c_t:=\sum_{j=s}^{K+1} b_j^t\rho_n^{-p_d}+{\hat A_s}^{-1/2} \tilde{F}$) uniformly depending on the parameters $\hat\Phi_s,\dots\hat\Phi_{K+1}$. Summing this and using Lemma \ref{lem:integral}, we have proved the following result.

\bel\label{lem:integral1}
We have:
\begin{equation}\label{eq:JK3}
\hat J_K=
\sum_{q=0}^K(\ln({\rho}))^q\sum_{p=0}^\infty
e(p,q){\rho}^{-p},
\ene
where
\bee\label{eq:estimatee1}
|e(p,q)|\ll\rho_n^{(2/3-p_K)p}\rho_n^{-Q\beta}.
\end{equation}
These estimates are uniform in the following regions of variables:
\bee\label{eq:range1}
\rho_n^{\beta}\ll l_t\ll\rho_n^{1/2},\ \ \rho_n^{-\delta_0}\ll b_j^t\ll \rho_n^{\delta_0},\ \ \rho_n^{2/3-q_K}<{\rho}.
\ene
\enl

Now Lemma \ref{lem:integral1}, Remark \ref{analytic}, and equation \eqref{eq:residues} show that the integral \eqref{eq:n4} admits decomposition of the form \eqref{eq:JK3} for $0.7\rho_n<\rho<5\rho_n$. This, together with equations \eqref{eq:n4}, \eqref{eq:n2},
\eqref{eq:densityh3}, Corollary \ref{cor:H1H2} and the observation that the number of different quasi-lattice subspaces $\GV$
is $\le\rho_n^{0+}$, completes the proof of Lemma \ref{main_lem} and, thus, of our main theorem in the case of all domains $M_p$ being simplexes. It remains to discuss
how to reduce the case of general region $\Xi(\GV)_p$ to the case of a simplex.

\section{Integration in non-simplex domains}

Now let us consider the case when the number $J_p$ of defining hyperplanes is bigger than $K+1$. Recall that we have
\bee\label{eq:Bxip11}
\Bxi(\GV)_p=\{\bxi\in\R^d,\ \bxi_{\GV}\in\Omega(\GV)\ \ \&\ \ \lu\bxi_{\GV^{\perp}},\tilde\bmu_j(p)\ru > L_{m+1},\ j=1,\dots,J_p\}.
\ene
In this section, we give only a brief description of how to prove the main statements, since the complete proof would be too long and tedious.
The complexity of the proof is mostly caused by the fact that we need to describe a procedure working in all dimensions. If one is interested
only in the cases $d=2$ or $d=3$, the proofs become substantially simpler (indeed, the case $d=2$ has already been proved, since
then we have $K=1$ or $K=0$, and any polyhedron in dimensions $0$ or $1$ is obviously a simplex).

Consider first the case when $\Bxi(\GV)_p$ is a `cone', i.e. there is a point $\ba=\ba(p)\in\GV^\perp$
such that $\lu\ba,\tilde\bmu_j(p)\ru = L_{m+1}$, $j=1,\dots,J_p$ (such point, if it exists, is always unique, and the existence of it is automatic in the
simplex case, i.e. when $J_p=K+1$). Then we can introduce the coordinates $(r,\tilde\Phi)$ by formulas \eqref{eq:r} and
\eqref{eq:tildePhi} with $\tilde\Phi\in M_p$, where $M_p$ is still given by \eqref{eq:Mp}.
\bel\label{lem:pc}
Suppose, $\bth\in\Bth_{\tilde k}$. Then we can write $\bth_{\GV^{\perp}}=\sum_{q}\tilde b_q\tilde\bmu_q(p)$, where either all $\tilde b_q$ are
non-positive, or all of them are non-negative (but such a decomposition is not necessarily unique).
\enl
\bep
Our assumptions imply that for each $\boldeta\in \tilde\Bxi(\GV)_p$ (where $\tilde\Bxi(\GV)_p$ is defined by \eqref{eq:Bxiptilde}) we have $\lu\boldeta,\bth_{\GV^{\perp}}\ru\ne 0$. Assume for
definiteness that $\lu\boldeta,\bth_{\GV^{\perp}}\ru> 0$. This property can be reformulated like this: whenever $\bz\in\GV^{\perp}$
is a vector with $\lu\bz,\bth_{\GV^{\perp}}\ru< 0$, there is at least one vector $\tilde\bmu_j(p)$ such that
$\lu\bz,\tilde\bmu_j(p)\ru< 0$. This, in turn, is equivalent to saying that whenever $\bz\in\GV^{\perp}$
is a vector with $\lu\bz,\bth_{\GV^{\perp}}\ru> 0$, there is at least one vector $\tilde\bmu_j(p)$ such that
$\lu\bz,\tilde\bmu_j(p)\ru> 0$. Consider the set $S:=\{\bz=\sum_{q}\tilde b_q\tilde\bmu_q(p), \tilde b_q\ge 0
\}$. We need
to prove that $\bth_{\GV^{\perp}}\in S$. Suppose not. Let $\bz_0$ be a nearest to $\bth_{\GV^{\perp}}$ point from $S$. Then
$\lu\bth_{\GV^{\perp}}-\bz_0,\bth_{\GV^{\perp}}\ru>0$, because otherwise the point
$|\bth_{\GV^{\perp}}|^2\lu\bz_0,\bth_{\GV^{\perp}}\ru^{-1}\bz_0\in S$ is closer to $ \bth_{\GV^{\perp}}$ than $\bz_0$.
Thus, there is a value of $j$, say $j=1$, so that
$\lu\bth_{\GV^{\perp}}-\bz_0,\tilde\bmu_1(p)\ru>0$. But then for sufficiently small $\epsilon>0$ the vector
$\bz_0+\epsilon\tilde\bmu_1(p)\in S$
is closer to $\bth_{\GV^{\perp}}$ than $\bz_0$. This contradiction proves our lemma.
\enp

This result shows that there is a big similarity between cases of the cone and the simplex. The only difference from the simplex case is that the number of sides of $M_p$ is now greater than $K+1$. This however means that
if we introduce the angular coordinates $\Phi$ in the same way as in Section 7, we will have difficulties trying to get rid of several of them.
Instead, we will follow a different strategy: we will cut the spherical polyhedron $M_p$ into several simplexes:
\bee\label{eq:qMp}
M_p=\sqcup_q {}^q M_p
\ene
 and then perform
integration over each simplex ${}^q M_p$ in the same way as we did in sections 7--10. The only thing we need to make sure is that
the lengths of all sides (edges) of  ${}^q M_p$, as well as all non-zero angles between two sides of any dimensions of
${}^q M_p$, is $\gg\rho_n^{0-}$ (let us call this {\it angles and sides property}). We, however, know what that angles and sides property
holds for the original polyhedron $M_p$ because of Lemma
\ref{lem:newangles} (strictly speaking, Lemma \ref{lem:newangles} was proved for simplexes, rather than for cones, but the proof is the same). Thus, the only problem we face is how to cut a polyhedron $M_p$ into simplexes ${}^q M_p$ without drastically decreasing sides or angles.
We do it by induction in $K$. For $K=1$ the statement is obvious (each $1$-dimensional connected polyhedron is a simplex, i.e. an interval). Assume that
$K$ is arbitrary. Also assume for simplicity that $M_p$ is not a spherical, but a Euclidean polyhedron (we can achieve this by
projecting $M_p$ onto any hyperplane tangent to it; obviously, this projection keeps the angles and sides property invariant).

Step I. We find a simplex $\hat M_p\subset M_p$ satisfying the angles and sides property. To do this, we consider a ball centered at any vertex
$\bv$ of $M_p$ of radius $\gg\rho_n^{0-}$, but sufficiently small so that the intersection of this ball with $M_p$ is a cone. The intersection of the boundary of this ball with $M_p$ is a polyhedron $N_p$ of dimension $K-1$. Running the induction argument, we can find a $(K-1)$-dimensional simplex $\hat N_p\subset N_p$
satisfying the angles and sides property. Now we define $\hat M_p$ as a convex hull of $\bv$ and vertexes of $\hat N_p$. A straightforward  geometrical
argument implies that $\hat M_p$ satisfies the angles and sides property. In particular, the volume of $\hat M_p$ is
$\gg\rho_n^{0-}$.

Step II. We find a point $\boldeta^*\in M_p$ such that the distance from $\boldeta^*$ to each of the $(K-1)$-dimensional sides of $M_p$ is
$\gg\rho_n^{0-}$ (the distance to the side is the length of the perpendicular dropped to the hyperplane containing this side). This point can be chosen to be the center of gravity of $\hat M_p$. Indeed, the distance from $\boldeta^*$ to a $(K-1)$-dimensional side of $M_p$ is the average
of the distances from the vertexes of $\hat M_p$ to this side. Thus, we need to show that the distance from at least one of the vertexes of $\hat M_p$ to this side is $\gg\rho_n^{0-}$. But if this were not the case, then the breadth of  $\hat M_p$ in the direction orthogonal to that side were
$\ll\rho_n^{0-}$, and so the volume of $\hat M_p$ were $\ll\rho_n^{0-}$, which would contradict estimates from Step I.

Step III. Now we use the inductive assumption and cut each $(K-1)$-dimensional side of $M_p$ into simplexes. Taking convex hulls of $\boldeta^*$
with these simplexes, we obtain the required decomposition of $M_p$ into simplexes ${}^q M_p$. It is a geometric exercise to check that
such constructed simplexes ${}^q M_p$ satisfy the angles and sides property.

Let us denote by ${}^q\Bxi(\GV)_p$ the infinite cone with the vertex $\ba$ and a cross-section ${}^q M_p$, i.e.
\bee\label{eq:qXip}
{}^q\Bxi(\GV)_p:=\ba+\{\bxi\in\R^d, \bxi_{\GV}\in\Omega(\GV)\ \& \ \bn(\bxi_{\GV^{\perp}})\in {}^q M_p\}.
\ene
Then we obviously have
\bee\label{eq:qXip1}
\Bxi(\GV)_p=\sqcup_q {}^q \Bxi(\GV)_p.
\ene
Now let us discuss how to perform the integration over ${}^q \Bxi(\GV)_p$. Let us fix $q$ and $p$ and denote by
$\bnu_1,\dots,\bnu_{K+1}$ the interior unit normal vectors to the faces of ${}^q \Bxi(\GV)_p$.
We denote, as before, $\Phi_q:=\frac{\pi}{2}-\phi(\bxi_{\GV^\perp}-\ba,\bnu_q(p))$, $q=1,\dots,K+1$.
In order to perform the integration, we need to check that Lemma \ref{lem:products} is still valid in the cone case. So, let $\bth\in\Bth_{\tilde k}$. Applying Lemma \ref{lem:pc}, we deduce that
\bee\label{eq:n111}
\bth_{\GV^{\perp}}=\sum_{q}\tilde b_q\tilde\bmu_q(p),
\ene
where either all $\tilde b_q$ are
non-positive, or all of them are non-negative; assume for definiteness that all of them are non-negative. We also have (applying, for example, the same lemma) that each vector $\tilde \bmu_q(p)$ admits a decomposition
$\tilde \bmu_q(p)=\sum_{l=1}^{K+1}\hat b_{ql}\bnu_l$ with all coefficients $\hat b_{ql}$ being non-negative. Now we have (denoting, as usual,
$\boldeta:=\bxi_{\GV^{\perp}}-\ba$ and putting $b_l:=\sum_{q=1}^{J_p}\tilde b_q\hat b_{ql}\ge 0$, $l=1,\dots,K+1$):
\bee\label{eq:innerproduct3n}
\bes
&\lu\bxi,\bth\ru=\lu X,\bth_{\GV}\ru+\lu\ba,\bth_{\GV^{\perp}}\ru+\lu\boldeta,\bth_{\GV^{\perp}}\ru\\
&=\lu X,\bth_{\GV}\ru+\sum_{q=1}^{J_p} \tilde b_q\lu\ba,\tilde\bmu_q\ru+\sum_{l=1}^{K+1} b_l\lu\boldeta,\bnu_l\ru\\
&=\lu X,\bth_{\GV}\ru+\sum_{q} \tilde b_qL_{m+1}+r\sum_l b_l\sin\Phi_l.
\end{split}
\ene
The estimates $\rho_n^{0-}\le |b_q| \le \rho_n^{0+}$ can be proved in the same way as Lemma \ref{lem:coefficients}. Finally,
\eqref{eq:n111} implies that $\sum_{q} \tilde b_q\gg\rho_n^{0-}$ for any $\bth\not\in\GV$. Multiplying \eqref{eq:n111} by $\ba$, we deduce that $\sum_{q} \tilde b_q\ll\rho_n^{0+}$.

This finishes the proof for the cone case. Now let us discuss the general case.

Let $\ba$ be any point inside $\Bxi(\GV)_p$ such that
for all $j$ we have $L_{m+1}\le\lu\ba,\tilde\bmu_j(p)\ru \ll L_{m+1}\rho_n^{0+}$. For each $l=0,\dots,J_p$, we define
\bee
\bes
\Bxi(\GV)^l_p:=&\{\bxi\in\R^d,\ \bxi_{\GV}\in\Omega(\GV)\ \& \\
&\lu\bxi_{\GV^{\perp}},\tilde\bmu_j(p)\ru > \lu\ba,\tilde\bmu_j(p)\ru,\ j=1,\dots,l, \ \&\\
&\lu\bxi_{\GV^{\perp}},\tilde\bmu_j(p)\ru > L_{m+1},\ j=l+1,\dots,J_p
\}
\end{split}
\ene
and
\bee
\check \Bxi(\GV)^l_p:=\Bxi(\GV)^l_p\setminus \Bxi(\GV)^{l+1}_p.
\ene
Then we obviously have
\bee
\Bxi(\GV)_p=\Bxi(\GV)^{J_p}_p\sqcup(\sqcup_{l=0}^{J_p-1} \check \Bxi(\GV)^l_p )
\ene
(as usual, modulo boundary points). The domain $\Bxi(\GV)^{J_p}_p$ is a cone, so we already know how to deal with it. Now let us consider
$\check \Bxi(\GV)^l_p$ for some $l$. Let us introduce a coordinate $t=t(\bxi):=\lu\bxi,\tilde\bmu_{l+1}(p)\ru$. Then for
$\bxi\in\check \Bxi(\GV)^l_p$ we have $t\in[L_{m+1},\lu\ba,\tilde\bmu_{l+1}(p)\ru]$. For each $t\in[L_{m+1},\lu\ba,\tilde\bmu_{l+1}(p)\ru]$
we denote $O^l(t):=\{\bxi\in\check \Bxi(\GV)^l_p,\ \lu\bxi,\tilde\bmu_{l+1}(p)\ru=t\}$. It is easy to see that the domain
$O^l(t)$ is of the same type as the domain  $\Bxi(\GV)_p$, i.e. it can be written as
\bee\label{eq:Bxipn}
\{\bxi\in\R^d,\ \bxi_{\GV}\in\Omega(\GV), \lu\bxi,\tilde\bmu_{l+1}(p)\ru=t \ \ \&\ \ \forall j\ne l+1\  \ \lu\bxi_{\GV^{\perp}},\tilde\bnu_j\ru > s_j
\},
\ene
where $\tilde\bnu_j$ are, of course, normalized projections of some $\tilde\bmu_j(p)$ onto the plane orthogonal to $\tilde\bmu_{l+1}(p)$.
Our aim is to compute
\bee\label{eq:n3new}
\vol\hat A^+\cap\check \Bxi(\GV)^l-\vol\hat A^-\cap\check \Bxi(\GV)^l
\ene
(or at least prove that this expression admits an asymptotic expansion in powers of $\rho$ and $\ln\rho$). But \eqref{eq:n3new}
is obviously equal to
\bee\label{eq:n3new1}
\int_{L_{m+1}}^{\lu\ba,\tilde\bmu_{l+1}(p)\ru}\bigl(\vol(\hat A^+\cap O^l(t))-\vol(\hat A^-\cap O^l(t))\bigr)dt.
\ene
If $\dim O^l(t)=1$, or, more generally, if $O^l(t)$ is a simplex, then we can perform integration over $O^l(t)$ as described above,
since the formula \eqref{eq:innerproduct3n} would still be valid, in the sense that
\bee
\lu\bxi,\bth\ru=C(X,t)+r\sum_l b_l\sin\Phi_l,
\ene
where $L_{m+1}\rho^{0-}_n\ll C(X,t)\ll L_{m+1}\rho^{0+}_n$ and the coordinates $(r,\{\Phi_l\})$ are the shifted polar coordinates in  $O^l(t)$.
Results of Section 10 show that for each fixed $t$ the expression $\bigl(\vol(\hat A^+\cap O^l(t))-\vol(\hat A^-\cap O^l(t))\bigr)$ admits
the asymptotic expansion in powers of $\rho$ and $\ln\rho$, with coefficients being uniformly bounded in $t$ (and $X$). Now it remains to integrate this expansion against $dt$ (and $dX$). If $O^l(t)$ is a cone, we cut it onto simplexes as described earlier in this section,
and then integrate over each simplex separately. Finally, if $O^l(t)$ is not a cone, we continue the process of reducing dimension
until the dimension of $O^l(t)$ becomes equal one.


\begin{thebibliography}{99}

\bibitem{BarPar}G. Barbatis and L. Parnovski
\emph{Bethe-Sommerfeld conjecture for pseudo-differential perturbation},
Comm.P.D.E., \textbf{34}(4) (2009), 383--418.

\bibitem{HelMoh} B. Helffer, A. Mohamed, \emph{Asymptotics of the density of states for the
Schr\"odinger operator with periodic electric potential}, Duke Math. J. \textbf{92} (1998), 1--60.

\bibitem{HitPol} M. Hitrik, I. Polterovich, \emph{Regularized traces and Taylor expansions for the heat semigroup},  J. London Math. Soc.   \textbf{68}(2)  (2003),  402--418.

\bibitem{HitPol1} M. Hitrik, I. Polterovich, \emph{Resolvent expansions and trace regularizations for Schr\"odinger operators}, Advances in Differential Equations and Mathematical Physics, Contemporary Mathematics, American Mathematical Society, 2003.

\bibitem{Kar} Yu. Karpeshina, \emph{On the density of states for the periodic Schr\"odinger operator},
Ark. Mat. \textbf{38} (2000), 111--137.

\bibitem{Kar1} Yu. Karpeshina, \emph{Perturbation theory for the Schr\"odinger operator with a periodic potential},
Lecture Notes in Math., Vol. 1663, Springer Berlin 1997.

\bibitem{Kat} T. Kato, \emph{Perturbation Theory for Linear Operators}, Springer 1980.

\bibitem{KorPush} E. Korotyaev and A. Pushnitski, \emph{On the High-Energy Asymptotics of the Integrated Density of States},
Bull. LMS \textbf{35} (2003), No. 6, 770--776.

\bibitem{Nai} M. A. Naimark \emph{Normed rings}, P.Noordhoff N.V., Netherlands, 1959.

\bibitem{Par} L. Parnovski  \emph{Bethe-Sommerfeld conjecture} Annales H. Poincar\'e, \textbf{9}(3) (2008), 457--508.



\bibitem{ParSht} L. Parnovski, R. Shterenberg \emph{Asymptotic expansion of the integrated density of states of a two-dimensional periodic Schroedinger operator},
Inv.Math., \textbf{176}(2) (2009), 275--323.


\bibitem{ParSob} L. Parnovski, A. Sobolev \emph{Bethe-Sommerfeld conjecture for periodic operators with strong perturbations}, 	 arXiv:0907.0887v1, submitted for publication.

\bibitem{ReeSim} M. Reed, B. Simon, \emph{Methods of Modern Mathematical Physics}, Academic Press, 1978.

\bibitem{Sav} A. V. Savin, \emph{Asymptotic expansion of the density of states
for one-dimensional Schr\"odinger and Dirac operators with
almost periodic and random potentials},
Sb. Nauchn. Tr. , I.F.T.P., Moscow (Russian), 1988.

\bibitem{Shu0} M. Shubin \emph{Almost periodic functions and partial differential operators},
Russian Math. Surveys \textbf{33}(2) (1978), 1--52.

\bibitem{Shu} M. Shubin \emph{The spectral theory and the index of elliptic operators with almost periodic coefficients},
Russian Math. Surveys \textbf{34}(2) (1979), 109--157.

\bibitem{ShuSche} D. Shenk and M. Shubin, \emph{Asymptotic expansion of the state density and the
spectral function of a Hill operator}, Math. USSR Sbornik \textbf{56} (1987), No. 2, 473--490.

\bibitem{Skr} M. Skriganov,
\emph{Geometrical and arithmetical methods in the spectral theory
of the multi-dimensional periodic operators}, Proc. Steklov Math.
Inst., Vol. 171,   1984.

\bibitem{Sob} A. V. Sobolev, \emph{Asymptotics of the integrated density of states for periodic elliptic pseudo-differential operators in dimension one},
Rev. Mat. Iberoam. \textbf{22} (2006),  no. 1, 55--92.

\bibitem{Sob1} A. V. Sobolev, \emph{Integrated Density of States for the Periodic
Schr\"odinger Operator in Dimension Two}, Ann. Henri Poincar\'e \textbf{6} (2005), 31--84.

\bibitem{Vel} O. A. Veliev, \emph{On the spectrum of multidimensional
periodic operators}, Theory of Functions, functional analysis and their applications,
Kharkov University, \textbf{49} (1988), 17--34  (in Russian).

\end{thebibliography}
\end{document}